\documentclass[hidelinks]{amsart}

\usepackage{mathtools}
\usepackage{xcolor,enumerate}       
 \definecolor{darkgreen}{rgb}{0,0.5,0}

\usepackage[T1]{fontenc}    
\usepackage{hyperref}       
\usepackage{url}            
\usepackage{booktabs}       
\usepackage{amsfonts}       
\usepackage{nicefrac}       
\usepackage{microtype}      
\usepackage{amssymb,latexsym,amsfonts,amsmath}
\usepackage{mathrsfs}

\usepackage{graphicx}
\usepackage{color}
\usepackage{mathabx}
\usepackage{algorithm}
\usepackage[noend]{algpseudocode}
\usepackage{bigints}
\usepackage{float}

\usepackage{color}
\usepackage{caption,subfigure, mdframed, tcolorbox, xcolor}

\usepackage{tikz}
\usetikzlibrary{arrows, positioning}

\tikzset{every node/.style={font=\small}}

\DeclareMathOperator\erfc{erfc}
\DeclareMathOperator\tr{tr}

\newcommand{\real}{{\mathbb{R}}}
\newcommand{\R}{{\mathbb{R}}}
\newcommand{\PP}{{\mathbb{P}}}
\newcommand{\E}{{\mathbb{E}}}
\newcommand{\G}{{\mathcal{G}}}
\newcommand{\D}{{\mathcal{D}}}
\newcommand{\F}{{\mathcal{F}}}
\newcommand{\C}{{\mathcal{C}}}
\newcommand{\K}{{\mathcal{K}}}
\newcommand{\Cinit}{\subscr{\C}{init}}
\newcommand{\Cund}{\subscr{\C}{und}}

\newcommand{\Cdyn}{\subscr{\C}{dyn}}

\newtheorem{theorem}{Theorem}[section]
\newtheorem{corollary}{Corollary}[section]
\newtheorem{proposition}{Proposition}[section]
\newtheorem{lemma}{Lemma}[section]
\newtheorem{sublemma}{Sublemma}[section]
\theoremstyle{definition}

\newtheorem{remark}{Remark}[section]

\newcommand{\eps}{\varepsilon}

\newcommand{\map}[3]{#1:#2 \rightarrow #3}

\newcommand\subscr[2]{#1_{\textup{#2}}}

\newcommand{\longthmtitle}[1]{\mbox{}\textup{\textbf{(#1):}}}

\newcommand{\N}{\mathcal{N}}

\newcommand{\oprocendsymbol}{\hbox{$\diamond$}}
\newcommand{\oprocend}{\relax\ifmmode\else\unskip\hfill\fi\oprocendsymbol}

\begin{document}

\begin{abstract} 
We provide the first known algorithm that provably achieves $\varepsilon$-optimality within $\widetilde{\mathcal{O}}(1/\varepsilon)$ function evaluations for the discounted discrete-time LQR problem with unknown parameters, without relying on two-point gradient estimates. These estimates are known to be unrealistic in many settings, as they depend on using the exact same initialization, which is to be selected randomly, for two different policies. Our results substantially improve upon the existing literature outside the realm of two-point gradient estimates, which either leads to $\widetilde{\mathcal{O}}(1/\varepsilon^2)$ rates or heavily relies on stability assumptions. \end{abstract}

\title{Sample Complexity of the Linear Quadratic Regulator: A Reinforcement Learning Lens}

\author[Amirreza Neshaei Moghaddam]{Amirreza Neshaei Moghaddam}
\address{Department of Electrical and Computer Engineering\\
University of California at Los Angeles,
Los Angeles}
\email{amirnesha@ucla.edu}
\author[Alex Olshevsky]{Alex Olshevsky}
\address{Department of Electrical and Computer Engineering\\
       Boston University}
\email{alexols@bu.edu}
\author[Bahman Gharesifard]{Bahman Gharesifard}
\address{Department of Electrical and Computer Engineering\\
       University of California, Los Angeles}
\email{gharesifard@ucla.edu}

\maketitle

\section{Introduction}

The Linear-Quadratic Regulator (LQR) has been used as a benchmark in optimal control theory since the sixties, see~\cite{EBL-LM:67}. The key distinguishing property of LQR problems is that the optimal controller is linear and can be fully characterized by the celebrated Riccati equation~\cite{DPB:95}. Naturally, with the recent increase in interest in model-free and data-driven methods, the study of LQR problems has resurfaced in the literature in scenarios where the model parameters are unknown and either need to be estimated, or model-free strategies need to be used. Even though such settings fall within the realm of adaptive control, the majority of classical studies addressing this issue have concentrated on system identification or examining asymptotic outcomes~\cite{TLL-CZW:82, HFC-LG:87, HFC-JFZ:90, MCC-PRK:98, SB-MCC:06}.

Recently, the problem has been examined from a machine learning standpoint in both online and offline contexts. In online settings, least-square estimators have been demonstrated to achieve sublinear regret. This area has seen extensive research focusing on the details of these estimations~\cite{YA-CS:11, AC-TK-YM:19,HM-ST-BR:19,MA-BG-TL:22-jmlr,MS-DF:20}. This paper focuses on the offline setting and builds on a sequence of breakthrough results through a reinforcement learning lens, starting with~\cite{MF-RG-SK-MM:18}. By establishing a gradient domination/Polyak-Lojasiewicz property, the results of~\cite{MF-RG-SK-MM:18} first demonstrate that exact gradient descent, in the model-based case, converges to the global optimal solution, despite the non-convex landscape of the LQR problem under study. Using this and in the model-free settings,  gradient estimations are derived from samples of the cost function value, leading to policy gradient methods. For the undiscounted discrete-time LQR under the random initialization setting, global convergence guarantees are provided using so-called one-point gradient estimates. As also explicitly pointed out in later work~\cite{DM-AP-KB-KK-PLB-MJW:18}, the convergence rate for obtaining an $\eps$-optimal policy established in~\cite{MF-RG-SK-MM:18} is only of the order $\widetilde{\mathcal{O}} (1/\eps^4)$ in zero-order evaluations. Note that by zero-order methods, we mean a setup where gradients are not available and can only be approximated using samples of the function value. The two most common such methods in the LQR problem are the one-point and two-point estimates where the former is obtained from a single function evaluation and the latter from two different such evaluations\footnote{We give the formal definitions of these estimates in equations~\eqref{eq: one-point} and~\eqref{eq: two-point}.}.

The next significant development related to our work is presented in~\cite{DM-AP-KB-KK-PLB-MJW:18}, which considers the discounted discrete-time LQR and employs zero-order methods for gradient estimation. For the essential case of one-point gradient estimation, an enhanced analysis is proposed. This analysis does not rely on stability assumptions (i.e., it does not assume a priori that the policies remain stable throughout the algorithm), yet improves the convergence rate reported in~\cite{MF-RG-SK-MM:18} from $\widetilde{\mathcal{O}} (1/\eps^4)$ to $\widetilde{\mathcal{O}} (1/\eps^2)$. Remarkably, with a two-point gradient estimate, $\eps$-optimality can be achieved using only $\widetilde{\mathcal{O}} (1/\eps)$ function evaluations. Similar findings are reported in~\cite{HM-AZ-MS-MRJ:22}, which are somewhat restrictive in terms of scaling of probability bounds with respect to dimensions. The substantial improvement in~\cite{DM-AP-KB-KK-PLB-MJW:18} stems from the application of sharp probabilistic estimates on stability regions using martingale techniques, a method we also heavily rely on. It should be noted that in both mentioned works, a constant learning rate is employed for the policy update. Interestingly, it is not difficult to observe that there is no advantage in using time-varying learning rates when the technique developed in~\cite{DM-AP-KB-KK-PLB-MJW:18} is applied directly.

It is worth pointing out the literature related to the discrete-time LQR problem with time-average cost. 
For instance, \cite{ZY-YC-MH-ZW:19} employs an actor-critic approach to achieve a sample complexity of $\widetilde{\mathcal{O}}(1/\eps^5)$. Similarly, using actor-critic methods, \cite{MZ-JL:23} demonstrates that a sample complexity of $\widetilde{\mathcal{O}}(1/\eps)$ is achievable, assuming almost sure stability and boundedness of the policy size throughout the algorithm.
However, the assumption of boundedness may not always be realistic, and more so is the assumption on stability, considering the inherently noisy dynamics. For example, this issue is echoed in the recent work~\cite{CJ-GK-GL:23}, which presupposes the boundedness of policies at every iteration.

As part of our contributions, 
and somewhat inspired by REINFORCE \cite{RJW:92, RSS-DM-SS-YM:99}, we propose a different gradient estimate scheme. Our approach relies on a new take on using policy gradient for gradient estimation based on appropriate sampling of deterministic policies, and only requires a single noisy cost evaluation, unlike two-point methods that require two evaluations under an identical noise realization~\cite{DM-AP-KB-KK-PLB-MJW:18}.
We are able to achieve high-probability upper bounds on our gradient estimations using moment concentration inequalities. Coupled with the adoption of time-varying learning rates, our methodology enables us to reach a $\widetilde{\mathcal{O}}(1/\eps)$ convergence rate, circumventing the need for two-point gradient estimations. 

Similar to~\cite{DM-AP-KB-KK-PLB-MJW:18}, our gradient estimate relies on an oracle that returns noisy zero-order evaluations of the cost function. Moreover, we assume access to a single state observation drawn randomly from the discounted state distribution. We consider this assumption milder than that of~\cite{SD-HM-NM-BR-ST:20}, which requires access to an entire state trajectory, or~\cite{DM-AP-KB-KK-PLB-MJW:18}, whose two-point method implicitly assumes the ability to both observe and \emph{select} a specific random initial state for a second policy rollout-something that is rarely feasible in realistic systems.

\section{Problem statement}
We start with a few mathematical notations that will be used throughout. 
For arbitrary matrix $M \in \R^{m \times n}$, we use $\lVert M \rVert$, $\lVert M \rVert_F$, and $\sigma_{\mathrm{min}} (M)$ to denote the 2-norm, Frobenius norm, and the minimum singular value of $M$ respectively. In addition, for a square matrix $\tilde{M} \in \R^{n \times n}$, $\rho (\tilde{M})$ denotes the spectral radius of $\tilde{M}$, $\tr(\tilde{M})$ the trace of $\tilde{M}$, and $\mathcal{K} (\tilde{M})$ the Kreiss constant of $\tilde{M}$:
\begin{equation}
    \mathcal{K} (\tilde{M}) := \sup_{\substack{\vert z \vert > 1, z \in \mathbb{C}}} (\vert z \vert - 1) \lVert (z I - \tilde{M})^{-1} \rVert.
\end{equation}
We also use $\langle M_1, M_2 \rangle := \tr (M_1^\top M_2)$ to denote the inner product of the matrices $M_1, M_2 \in \R^{m \times n}$.

Let us now define the problem under study. We consider the discrete-time infinite-horizon discounted LQR problem
\begin{equation}\label{eq: main-LQR}
    \min \E \left[ \sum_{t \geq 0} \gamma^t c_t \right] \quad \text{s.t.} \quad x_{t+1} = A x_t + B u_t + z_t,
\end{equation}
where $x_t \in \R^{n}$ is the system state at time $t$, initialized (deterministically or randomly) at $x_0$; $u_t \in \R^m$ is the control input at time $t$; and $z_t \in \R^{n}$ is the additive noise of the system at time $t$. The stage cost is defined as
\[
c_t := x_t^\top Q x_t + u_t^\top R u_t,
\]
where $Q \in \R^{n \times n}$ and $R \in \R^{m \times m}$ are positive-definite matrices that parameterize the quadratic costs. The system matrices are $A \in \R^{n \times n}$ and $B \in \R^{n \times m}$. In most  of what follows, we assume that the pair $(A,B) $ is controllable. 

As noted above, randomness is introduced in two different ways in the above problem formulation: through the initialization or as an added disturbance to the dynamics. This has led to two separate scenarios considered in the literature: 
\begin{itemize}
    \item \textbf{Random initialization:} where it is assumed that the additive noise $z_t$ is zero for all $t \geq 0$, and that the initial state $x_0$ is randomly chosen from an initial distribution $\D_0$. Given the initial state $x_0$, we let $\C_{\text{init},\gamma} (K;x_0)$ be the random variable representing the cost of implementing the linear policy $K \in \R^{m \times n}$, i.e., choosing $u_t = -K x_t$ for $t \geq 0$, from the initial state $x_0$:
    \begin{equation}
        \Cinit (K;x_0) := \sum_{t = 0}^{\infty} \gamma^t (x_t^\top Q x_t + u_t^\top R u_t), \label{eq: initial_cost_wrt_state}
    \end{equation}
    where $0 < \gamma \leq 1$ is the discount factor, and the dynamics is given by~\eqref{eq: main-LQR} with $z_t=0$. That is, in this case the trajectories satisfy the dynamics 
     \begin{align}
        x_{t+1} =& A x_t + B u_t, \cr
        u_t =& -K x_t.\label{eq:sys-no-noise}
    \end{align}
    Naturally, the objective is to minimize the population cost defined as
    \begin{equation}
        \Cinit(K) := \E_{x_0 \sim \D_0} [\Cinit(K; x_0)]
    \end{equation}
    over choices of the policy $K$.
    \item \textbf{Noisy dynamics:} where it is assumed $z_t$ is drawn i.i.d. for each $t$ from a distribution $\D_{\text{add}}$, and that the initial state $x_0$ is set deterministically to zero. Given a sequence of random variables $\mathcal{Z} = \{ z_t \}_{t \geq 0}$, we let $\Cdyn(K; \mathcal{Z})$ be the random variable representing the cost of implementing the linear policy $K$ on a system where the additive noise is drawn from $\mathcal{Z}$, i.e.,
    \begin{equation}
    \Cdyn(K; \mathcal{Z}) := \sum_{t = 0}^{\infty} \gamma^t (x_t^\top Q x_t + u_t^\top R u_t),
    \end{equation}
    where we have set $x_0 = 0$, the dynamics is given by~\eqref{eq: main-LQR} with $u_t = -K x_t$ for each $t \geq 0$, and $0 < \gamma < 1$ is the discount factor. 
    In contrast to the random initialization setting, the discount factor in this setting obeys $\gamma < 1$ to prevent the cost from diverging to infinity for all $K$ due to the accumulation of noise over time.
    Once again, the objective is to minimize the population cost
    \begin{equation}
    \Cdyn(K) := \E_{\mathcal{Z} \sim
      \D_{\text{add}}^{\mathbb{N}}} [\Cdyn(K; \mathcal{Z})].
    \end{equation}
\end{itemize}

By classical results in optimal control theory, see e.g.,~\cite{REK-etal:60, EBL-LM:67}, the optimal controller in both cases is linear and can be expressed as~$u_t = - K^* x_t $
where $t \geq 0$ and $K^* \in \mathbb{R}^{m \times n}$ is the controller gain, and can be explicitly computed. When the system matrices are known, which is not the case in this paper, the policy $K^*$ can be derived as follows 
\begin{equation}
    K^* = \gamma (R + \gamma B^\top P B)^{-1} B^\top P A, \label{eq: K* explicit formulation}
\end{equation}
where $P$ denotes the unique positive definite solution to the discounted discrete-time algebraic Riccati equation~\cite{DPB:95}:
\begin{equation}
    P = \gamma A^\top P A - \gamma^2 A^\top P B (R + \gamma B^\top P B)^{-1} B^\top P A + Q. \label{eq: discounted discrete-time ARE}
\end{equation}
Throughout this paper, we closely follow the notation and terminology that is introduced in the seminal work~\cite{DM-AP-KB-KK-PLB-MJW:18}. To start, for a random variable $v \sim \D$ where $\D \in \left\{ \D_0, \D_{\text{add}} \right\}$, we assume that
\begin{equation}
    \E[v] = 0, \quad \E[v v^\top] = I, \ \text{and} \ \lVert v \rVert^2 \leq C_m \ \ \text{a.s.} \label{eq: noise_assumption}
\end{equation}
where as per usual, ``a.s.'' refers to almost surely.  
The assumption on the covariance being identity is without loss of generality, see~\cite{DM-AP-KB-KK-PLB-MJW:18}.
Moreover, it is noteworthy to mention that using the definition \eqref{eq: initial_cost_wrt_state} with the trajectories following \eqref{eq:sys-no-noise}, the cost for the random initialization setting can be rewritten as
\begin{equation}
    \Cinit (K; x_0) = x_0^\top P_K x_0, \label{eq: cost wrt state via P_K}
\end{equation}
where $P_K$ is the symmetric positive semi-definite solution to the fixed point equation:
\begin{equation}
    P_K = Q + K^\top R K + \gamma (A-BK)^\top P_K (A-BK).
\end{equation}
Consequently, it also holds that
\begin{align}
    \Cinit(K) &= \E_{x_0 \sim \D_0} [\Cinit(K; x_0)] \cr
    &= \E_{x_0 \sim \D_0} [x_0^\top P_K x_0] \cr
    &= \E_{x_0 \sim \D_0} [\tr (P_K x_0 x_0^\top)] \cr
    &= \tr (P_K \E_{x_0 \sim \D_0} [x_0 x_0^\top]) \cr
    &\overset{\mathrm{(i)}}{=} \tr (P_K), \label{eq: cost defined via P_K}
\end{align}
where (i) follows from assumption~\eqref{eq: noise_assumption} on the randomness. Although this formulation is stated for the cost under the random initialization setting, it turns out that the two costs are essentially equivalent when the respective systems are driven by noise with the same first two moments, in the sense that is shown in Lemma~\ref{lem: noisy-random cost equivalence} to follow. For this reason, we focus on the random initialization scenario henceforth.

Let us now state the problem that we consider throughout this paper. We recall here that we assume that the pair $(A,B)$ is controllable, however, \emph{unknown}. A policy $K$ is said to stabilize the system $(A,B)$ if we have $\rho(A-BK) < 1$. Note that by the controllability assumption, there exists some policy $K$ satisfying the condition $\rho(A-BK) < 1$. Furthermore, we assume access to~\emph{some} stable policy $K_0$; this is a mild assumption that can be satisfied in a variety of ways; we refer the reader to~\cite{MF-RG-SK-MM:18, SD-HM-NM-BR-ST:18}. We use $K_0$ to initialize our algorithms, which we shortly introduce. 

With this in mind, the main objective of this paper is to find an $\eps$-optimal policy $\hat{K}$, i.e., one satisfying
\[
\Cinit(\hat{K}) - \Cinit(K^*) \leq \eps,
\]
where $K^*$ is an optimal policy. The proposed scheme in the literature crucially involves forming an estimation of the gradient of the cost function~\eqref{eq: initial_cost_wrt_state}, which is then used for a gradient update with an appropriate learning rate. 

To make our later comparisons precise and to clarify the discussions emphasized earlier, we now recall the standard forms of the one-point and two-point estimates. The one-point estimate at a policy $K \in \mathbb{R}^{m \times n}$ is computed as
\begin{equation}
    \textsl{g}_r^1 (K) := \Cinit (K + r U;x_0) \cdot \frac{m n}{r} \, U, \label{eq: one-point}
\end{equation}
for a smoothing radius $r \in \mathbb{R}$ and a random matrix $U \in \mathbb{R}^{m \times n}$ drawn uniformly over matrices with unit Frobenius norm. The two-point estimate instead uses
\begin{equation}
    \textsl{g}_r^2 (K) := \left[ \Cinit (K + r U;x_0) - \Cinit (K - r U;x_0) \right] \cdot \frac{m n}{2r} \, U, \label{eq: two-point}
\end{equation}
which requires cost evaluations under two different policies, $K + rU$ and $K - rU$, with respect to the \emph{same} initial condition $x_0$. This is often unrealistic in practice, since $x_0$ is typically random and not something the algorithm can choose or reproduce across rollouts. The estimator we propose later avoids this assumption and instead works by just using a single noisy cost evaluation along one perturbed trajectory.

In accordance with this, we present an algorithm here, displayed as Algorithm~\ref{alg: LQR_policy_gradient}, where we use an estimate inspired by the REINFORCE method \cite{RJW:92, RSS-DM-SS-YM:99} with a time-varying learning rate to achieve $\eps$-optimality. Below, we present a brief roadmap of the key contributions and supporting arguments developed in this paper.

\begin{figure}[H]
\centering
\begin{tikzpicture}[
  node distance=1.2cm and 1.2cm,
  box/.style = {draw, rounded corners, text width=4cm, align=center, font=\small}
]

\node[box] (lem_stability) {Lemma~\ref{lem: gamma condition} \\ (stability guarantees)};
\node[box, right=of lem_stability] (lem_size_bounds) {\textbf{Lemma~\ref{lem: gradient estimate bounds} \\ (probabilistic bounds on the estimate)}};
\node[box, below=of lem_stability] (prop_unbiased_estimate) {\textbf{Proposition~\ref{prop: gradient estimate expectation} \\ (unbiased estimate)}};
\node[box, right=of prop_unbiased_estimate] (lem_estimate_bias) {Lemma~\ref{lem: gradient estimate conditioned} \\ (estimate bias conditioned on bounded size)};
\node[box, below=of lem_estimate_bias] (lem_regularity) {\textbf{Lemmas~\ref{lem:lipschitz_gradient_lqr},~\ref{lem: PL_lqr} \\ (regularity properties)}};
\node[box, right=of lem_estimate_bias] (thm_convergence) {Theorem~\ref{thm: lqr policy gradient} \\ (main convergence statement)};

\draw[->] (lem_stability) -- (lem_size_bounds);
\draw[->] (lem_size_bounds) -- (lem_estimate_bias);
\draw[->] (prop_unbiased_estimate) -- (lem_estimate_bias);
\draw[->] (lem_size_bounds) -- (thm_convergence);
\draw[->] (lem_estimate_bias) -- (thm_convergence);
\draw[->] (lem_regularity) -- (thm_convergence);

\end{tikzpicture}
\caption{Roadmap of the main technical results.}
\label{fig:result_roadmap}
\end{figure}

In this diagram, we omit most intermediate steps and highlight (in bold) the main components that the convergence theorem ultimately depends on. Among these, \textbf{Lemma~\ref{lem: gradient estimate bounds}} and \textbf{Proposition~\ref{prop: gradient estimate expectation}} are the core novel contributions of this paper. The regularity properties (Lemmas~\ref{lem:lipschitz_gradient_lqr},~\ref{lem: PL_lqr}), which we will discuss in detail in the next section, are adapted from prior work~\cite{DM-AP-KB-KK-PLB-MJW:18} and included here for completeness.

\subsection{Regularity properties}
We introduce some notations related to the regularity properties of the cost functions; these will play a crucial role in some of our bounds; the next few results are borrowed from~\cite{DM-AP-KB-KK-PLB-MJW:18}. 

\begin{lemma}[LQR Cost is locally Lipschitz] \cite[Lemma 4]{DM-AP-KB-KK-PLB-MJW:18}
\label{lem:lipschitz_cost_lqr}
Given any linear policy $K$ with finite cost, there exist positive scalars $(\lambda_{K}, \widetilde{\lambda_{K}}, \zeta_{K})$, depending on the function value $\Cinit(K)$, such that for all policies $K'$
satisfying ${\lVert K' - K \rVert_F \leq \zeta_{K}}$, and for all initial states $x_0$, we have
\begin{subequations}
\begin{align}
\vert \Cinit(K') - \Cinit(K) \vert \leq & \lambda_{K} \lVert K' - K \rVert_F, \text{ and} \\
\vert \Cinit(K'; x_0) - \Cinit(K; x_0) \vert \leq & \widetilde{\lambda_{K}} \lVert K' - K \rVert_F.
\end{align}
\end{subequations}
\end{lemma}

\begin{lemma}[LQR Cost has locally Lipschitz Gradients] \cite[Lemma~5]{DM-AP-KB-KK-PLB-MJW:18}
\label{lem:lipschitz_gradient_lqr}
Given any linear policy $K$ with finite cost, there exist positive scalars $(\beta_{K}, \phi_{K})$, depending on the function value $\Cinit(K)$, such that for all policies $K'$ satisfying $\lVert K' - K \rVert_F \leq \beta_{K}$, we have
\begin{align}
\lVert \nabla \Cinit(K') - \nabla \Cinit(K) \rVert_F \leq \phi_{K} \lVert K' - K \rVert_F.
\end{align}
\end{lemma}

\begin{lemma}[LQR satisfies PL] \cite[Lemma 6]{DM-AP-KB-KK-PLB-MJW:18}
\label{lem: PL_lqr}
There exists a universal constant $\mu_\text{lqr} > 0$ such that for all stable policies $K$, we have
\begin{align}
  \lVert \nabla \Cinit(K) \rVert_F^2 \geq \mu_{\text{lqr}} \left( \Cinit(K) - \Cinit(K^*) \right), \label{eq: PL_LQR}
\end{align}
where $K^*$ is a global minimizer of the cost function $\Cinit$.
\end{lemma}
For the sake of exposition, these properties are stated here without specifying the various smoothness and PL constants. The explicit expressions for $\{\lambda_{K}, \widetilde{\lambda_{K}}, \phi_{K}, \beta_{K}, \zeta_{K}, \mu_\text{lqr}\}$ in terms of the parameters of the LQR problem are provided in~\cite[Appendix~A] {DM-AP-KB-KK-PLB-MJW:18}. Remark~\ref{rem: regularity parameters} to follow will provide further elaboration on these parameters as well.

\begin{lemma}[Equivalence of population costs up to scaling] \cite[Lemma~7]{DM-AP-KB-KK-PLB-MJW:18}
\label{lem: noisy-random cost equivalence}
For all policies $K$, we have
\begin{align*}
\Cdyn(K) = \frac{\gamma}{1 - \gamma} \Cinit(K).
\end{align*}
\end{lemma}

This result shows that the noisy dynamics and random initialization population costs behave identically when their respective sources of randomness have the same first two moments. Therefore, we focus on the random initialization cost from now on and remind the reader that $\C (K) := \Cinit (K)$ for ease of notation.

We define the set
\begin{equation}
    \mathcal{G}^\text{lqr} := \{K \ | \ \C (K) - \C (K^*) \leq 10 \C (K_0) \}. \label{eq: G_LQR}
\end{equation}
Since $\C$ is $(\zeta_K, \lambda_K)$ locally Lipschitz and $(\beta_K, \phi_K)$ locally smooth, both properties hold simultaneously within a Frobenius norm radius $\omega_K := \min\{\beta_K,\zeta_K\}$ of a point $K \in \mathcal{G}^\text{lqr}$. We define the quantities
\begin{equation*}
    \phi_\text{lqr} := \sup_{K \in \mathcal{G}^\text{lqr}} \phi_K, \qquad \lambda_\text{lqr} := \sup_{K \in \mathcal{G}^\text{lqr}} \lambda_K, \quad \text{and} \quad \omega_\text{lqr} := \inf_{K \in \mathcal{G}^\text{lqr}} \omega_K.
\end{equation*}
It is noteworthy to mention that these values are non-zero and finite, and their explicit formulation is provided in \cite[Appendix~A]{DM-AP-KB-KK-PLB-MJW:18}, see Remark~\ref{rem: regularity parameters} to follow for further clarification.

Observe that by the definition of these quantities, one can immediately show that for any $K \in \G^{\text{lqr}}$ and $K' \in \R^{m \times n}$ such that $\lVert K' - K \rVert_F \leq \omega_{\text{lqr}}$, we have that
\begin{align*}
    \vert \C(K') - \C(K) \vert \leq & \lambda_{\text{lqr}} \lVert K' - K \rVert_F, \text{ and} \\
    \lVert \nabla \C(K') - \nabla \C(K) \rVert_F \leq & \phi_{\text{lqr}} \lVert K' - K \rVert_F.
\end{align*}

\begin{remark}{\em \label{rem: regularity parameters}
    We now describe how to specify the set of parameters  $\{\lambda_{K}, \widetilde{\lambda_{K}}, \phi_{K}, \beta_{K}, \zeta_{K}, \mu_\text{lqr}\}$ in our setting. We start by recalling that a set of parameters $\{c_{K_0}, c_{K_1}, \dotsc, c_{K_9} \}$ is defined in \cite[Appendix~A]{DM-AP-KB-KK-PLB-MJW:18}, which notably depend on $\C (K)$. Subsequently, by replacing said $\C (K)$ with $\sup_{K \in \G^{\text{lqr}}} \C (K)$, they obtain a set of constants $\{ \widetilde{c_{K_0}}, \widetilde{c_{K_1}}, \dotsc, \widetilde{c_{K_9}} \}$ which are independent of $K$. For ease of access for the reader, we point out that
    \begin{equation}
        \omega_{\text{lqr}} = \widetilde{c_{K_9}}, \quad \phi_{\text{lqr}} = \widetilde{c_{K_7}}, \quad \text{and} \quad \lambda_{\text{lqr}} = \widetilde{c_{K_8}}.
    \end{equation}
    Moreover, it holds that $\max \{ \lVert K \rVert, \| \nabla \C (K) \|_F \} \leq \widetilde{c_{K_1}}$ for any $K \in \G^{\text{lqr}}$, see \cite[Appendix~A]{DM-AP-KB-KK-PLB-MJW:18} and \cite[Lemma~22]{MF-RG-SK-MM:18}.
    Note that the only required modification in the values of $\widetilde{c_{K_0}}, \widetilde{c_{K_1}}, \dotsc, \widetilde{c_{K_9}}$ for our case is having $10 \C (K_0) + \C (K^*)$ as $\sup_{K \in \G^{\text{lqr}}} \C (K)$ instead of \cite{DM-AP-KB-KK-PLB-MJW:18}'s $10 \C (K_0) - 9 \C (K^*)$, due to the difference in our definition of $\G^{\text{lqr}}$ in~\eqref{eq: G_LQR}.
    \oprocend}
\end{remark}

We now provide an informal statement of our main result, which shows that our proposed algorithm obtains an $\eps$-optimal policy after $\widetilde{\mathcal{O}} (1/\eps)$ iterations. As we outline precisely later, this algorithm forms an estimate $ \widehat{\nabla \C} (K_t) $ of the gradient at a given time $ t $ and updates the policy $ K_t $ with a time-varying learning rate $ \alpha_t $. 
\begin{theorem}\longthmtitle{Informal Statement of Our Main Result}
    If the step-size is chosen as $ \alpha_t =C\frac{1}{t+N}$ with $ N $ ``large enough'', i.e., $N \sim \mathcal{O} \left((\log{\frac{1}{\delta}})^{3/2}\right)$ for any chosen $\delta$, and $C$ being a known constant, then after $ T = \mathcal{O} \left( \frac{1}{\eps}(\log{\frac{1}{\delta}})^{3/2} \right)$ iterations, provided the discount factor exceeds a constant threshold strictly less than $1$,  we have that 
    \begin{equation} 
        \C(K_T) - \C(K^*) \leq \eps
    \end{equation}
    with a probability of at least $4/5 - \delta T$. In particular, choosing $\delta$ proportional to $1/T$, we attain $\C(K_T)-\C(K^*)$ with a constant probability with a sample complexity of $\widetilde{\mathcal{O}}\left( 1/\eps \right)$. 
\end{theorem}

A precise version of this result is given later in Theorem~\ref{thm: lqr policy gradient}, with the corresponding algorithm formally stated in Algorithm~\ref{alg: LQR_policy_gradient}.

Let us first point out that this result substantially improves the ones in the literature by achieving a $\widetilde{\mathcal{O}} (1/\eps)$ rate without any additional assumptions. The best previous result achieves a convergence rate of $\widetilde{\mathcal{O}} (1/\eps^2)$ \cite{DM-AP-KB-KK-PLB-MJW:18} in this setting. Indeed, $\widetilde{\mathcal{O}} (1/\eps)$ rates were only available using so-called two-point estimates which re-use randomness (e.g., require being able to initialize the system at a given $x_0$). Note that the limitations of this assumption become especially evident in the noisy dynamics setting, where access to cost evaluations of two different policies is required under the exact same infinite sequence of additive noise. 
This is significantly more restrictive than in the random initialization setting, which only requires matching a single random variable---namely, the initial condition.
In both cases, however, this coupling is difficult to realize in practice,
as one must have perfect control over a simulator to use such estimates; one cannot implement them for black-box systems with unknown dynamics which need to learn in the real world, for example. In contrast, our result only uses gradient estimates with a single zero-order evaluation at each step. 

We now begin the process of collecting the essentials needed to articulate our theorem precisely and to prove this result, beginning with a  fresh examination of the policy gradient that we employ for gradient estimation.

\section{Policy gradient}

Most formulations of the policy gradient require probabilistic policies; in contrast, as can be seen in \eqref{eq:sys-no-noise}, we have used a deterministic policy given by $u_t = -K x_t$. To remedy, we utilize the control input $u_{\hat{t}}$, to be defined shortly, where $ \hat{t} $ is sampled at random 
from the distribution $\mu_\gamma (t) := (1-\gamma) \gamma^t$, where $t \in \{0,1,2,\cdots\}$. Keeping this in mind, we now compute
\begin{equation}
    \widehat{\nabla \C} (K) := \frac{1}{1 - \gamma} Q^{K} (x_{\hat{t}}, u_{\hat{t}}) \nabla_{K} \log \pi_K (u_{\hat{t}} | x_{\hat{t}}), \label{eq: gradient estimate definition}
\end{equation}
where the control input $u_{\hat{t}}$ is randomly chosen from the Gaussian distribution $\N (- K x_{\hat{t}}, \sigma^2 I_m)$ for some $\sigma>0$ only for the selected iteration $\hat{t}$, and $x_{\hat{t}} = (A-BK)^{\hat{t}} x_0$ with $x_0 \sim \D$ as before. Note that
\begin{equation}
    \E_{\hat{t} \sim \mu_\gamma} \left[ \widehat{\nabla \C} (K) \right] = \sum_{t=0}^{\infty} \gamma^t Q^{K} (x_t, u_t) \nabla_{K} \log \pi_K (u_t | x_t), \label{eq: policy gradient formulation}
\end{equation}
where 
\begin{equation}
    \pi_K (u_t | x_t) = \frac{1}{\sqrt{(2 \pi)^m (\sigma^2)^m}} e^{- \frac{(u_t + K x_t)^\top (u_t + K x_t)}{2 \sigma^2}}, \label{eq: gaussian policy}
\end{equation}
and
\begin{align}
    Q^{K} (x_t, u_t) &:= x_t ^\top Q x_t + u_t^\top R u_t + \gamma \Cinit (K; x_{t+1}) \cr
    &=x_t ^\top Q x_t + u_t^\top R u_t + \gamma \Cinit (K; A x_t + B u_t) \cr
    &\overset{\mathrm{(i)}}{=} x_t ^\top Q x_t + u_t^\top R u_t + \gamma (A x_t + B u_t)^\top P_K (A x_t + B u_t), \label{eq: Q function definition}
\end{align}
where (i) is on account of \eqref{eq: cost wrt state via P_K}. Note that we can also rewrite $u_{\hat{t}} \sim \N (-K x_{\hat{t}}, \sigma^2 I_m)$ as
\begin{equation}
    u_{\hat{t}} = -K x_{\hat{t}} + \sigma \eta_{\hat{t}}, \label{eq: u_t rewrite using eta_t}
\end{equation}
where $\eta_{\hat{t}} \sim \N (0, I_m)$. Moreover, we have the following lemma to provide an alternative way of representing~\eqref{eq: gradient estimate definition}. 
\begin{lemma} \label{lem: alternative estimate formulation}
    The gradient estimate in~\eqref{eq: gradient estimate definition} can be modified to get
    \begin{equation}
    \widehat{\nabla \C} (K) = - \frac{1}{\sigma (1 - \gamma)} Q^{K} (x_{\hat{t}}, -K x_{\hat{t}} + \sigma \eta_{\hat{t}}) \eta_{\hat{t}} x_{\hat{t}}^\top. \label{eq: gradient estimate practical formulation}
    \end{equation}
\end{lemma}
\begin{proof} Following \eqref{eq: gradient estimate definition}, we have that
\begin{align}
    \widehat{\nabla \C} (K) =& \frac{1}{1 - \gamma} Q^{K} (x_{\hat{t}}, u_{\hat{t}}) \nabla_{K} \log \pi_K (u_{\hat{t}} | x_{\hat{t}}) \cr
    \overset{\mathrm{(i)}}{=} &\frac{1}{1 - \gamma} Q^{K} (x_{\hat{t}}, u_{\hat{t}}) \nabla_{K} \left( - \frac{(u_{\hat{t}} + K x_{\hat{t}})^\top (u_{\hat{t}} + K x_{\hat{t}})}{2 \sigma^2} \right) \cr
    =& \frac{1}{1 - \gamma} Q^{K} (x_{\hat{t}}, u_{\hat{t}}) \nabla_{K} \left( - \frac{ u_{\hat{t}}^\top u_{\hat{t}} + 2 u_{\hat{t}}^\top K x_{\hat{t}} + x_{\hat{t}}^\top K^\top K x_{\hat{t}}}{2 \sigma^2} \right) \cr
    =& \frac{1}{1 - \gamma} Q^{K} (x_{\hat{t}}, u_{\hat{t}}) \nabla_{K} \left( - \frac{ \tr \left( 2 x_{\hat{t}} u_{\hat{t}}^\top K \right) + \tr \left( x_{\hat{t}} x_{\hat{t}}^\top K^\top K \right)}{2 \sigma^2} \right), \label{eq: gradient other form prelim}
\end{align}
where (i) follows from \eqref{eq: gaussian policy}. Now note that
\begin{equation}
\nabla_K \tr \left( 2 x_{\hat{t}} u_{\hat{t}}^\top K \right) = \nabla_K \tr \left( \left(2 u_{\hat{t}} x_{\hat{t}}^\top \right)^\top K \right) = \nabla_K \left\langle 2 u_{\hat{t}} x_{\hat{t}}^\top, K \right\rangle = 2 u_{\hat{t}} x_{\hat{t}}^\top, \label{eq: linear gradient wrt matrix}
\end{equation}
and
\begin{align}
\nabla_K \tr \left( x_{\hat{t}} x_{\hat{t}}^\top K^\top K \right) =& \nabla_{K_1} \tr \left( x_{\hat{t}} x_{\hat{t}}^\top K^\top K_1 \right) + \nabla_{K_2} \tr \left( x_{\hat{t}} x_{\hat{t}}^\top K_2^\top K \right) \cr
=& \nabla_{K_1} \tr \left( \left( K x_{\hat{t}} x_{\hat{t}}^\top \right)^\top K_1 \right) + \nabla_{K_2} \tr \left( K_2^\top \left( K x_{\hat{t}} x_{\hat{t}}^\top \right) \right) \cr
=& \nabla_{K_1} \left\langle K x_{\hat{t}} x_{\hat{t}}^\top, K_1 \right\rangle + \nabla_{K_2} \left\langle K x_{\hat{t}} x_{\hat{t}}^\top, K_2 \right\rangle \cr
=& 2 K x_{\hat{t}} x_{\hat{t}}^\top. \label{eq: quadratic gradient wrt matrix}
\end{align}
As a result, combining \eqref{eq: linear gradient wrt matrix} and \eqref{eq: quadratic gradient wrt matrix} with \eqref{eq: gradient other form prelim} yields
\begin{align*}
    \widehat{\nabla \C} (K) =& \frac{1}{1 - \gamma} Q^{K} (x_{\hat{t}}, u_{\hat{t}}) \left( -\frac{1}{2 \sigma^2} \left( 2 (K x_{\hat{t}} x_{\hat{t}}^\top + u_{\hat{t}} x_{\hat{t}}^\top) \right) \right) \\
    =& \frac{1}{1 - \gamma} Q^{K} (x_{\hat{t}}, u_{\hat{t}}) \left( - \frac{(u_{\hat{t}} + K x_{\hat{t}})}{\sigma^2} x_{\hat{t}}^\top \right) \\
    \overset{\mathrm{(i)}}{=}& - \frac{1}{\sigma (1 - \gamma)} Q^{K} (x_{\hat{t}}, -K x_{\hat{t}} + \sigma \eta_{\hat{t}}) \eta_{\hat{t}} x_{\hat{t}}^\top,
\end{align*}
where (i) follows from~\eqref{eq: u_t rewrite using eta_t}. This finishes the proof.
\end{proof}
We now provide the following remark on the computation of $Q^K (x_{\hat{t}},u_{\hat{t}})$.
\begin{remark} \label{rem: Q computation}
    The $Q$-function in~\eqref{eq: gradient estimate practical formulation} represents the cost-to-go from time step~$\hat{t}$. Using the quadratic stage cost $c_t := x_t^\top Q x_t + u_t^\top R u_t$, we can write
    \[
    Q^K (x_{\hat{t}}, u_{\hat{t}}) = \sum_{t=\hat{t}}^\infty \gamma^{t - \hat{t}} c_t,
    \]
    where the dynamics follow~\eqref{eq:sys-no-noise} with control
    \[
    u_t = \begin{cases}
        -K x_t + \sigma \eta_t, & \text{if } t = \hat{t}, \\
        -K x_t, & \text{otherwise},
    \end{cases}
    \]
    and $x_0 \sim \D$. This is analogous to the zero-order oracle in~\cite{DM-AP-KB-KK-PLB-MJW:18}, which computes
    \[
    \C(K; x_0) := \sum_{t=0}^\infty \gamma^t c_t \quad \text{with } u_t = -K x_t.
    \]
    Accordingly, we also assume access to an oracle that returns a single noisy evaluation of such costs under the given policy. 
\end{remark}

Taking the alternative formulation of our gradient estimate provided in Lemma~\ref{lem: alternative estimate formulation} into consideration, we introduce the algorithm
\begin{algorithm}
\caption{LQR with Policy Gradient}
\label{alg: LQR_policy_gradient}
\begin{algorithmic}[1]
\State{Given iteration number $T \geq 1$, initial policy $K_0 \in
  \R^{m \times n}$, noise parameter $\sigma$, and step size $\alpha_t > 0$}
  \For{$t \in \{ 0, 1, \ldots, T-1 \}$}
\State{Sample $x_0 \sim \D$, $\hat{t} \sim \mu_\gamma$, and
  $\eta_{\hat{t}} \sim \N (0, I_m)$}
  
  \State{Simulate $K_t$ for $\hat{t}$ steps starting from $x_0$ and observe $x_{\hat{t}}$.}
  
  \State{$u_{\hat{t}} \gets - K_t x_{\hat{t}} + \sigma \eta_{\hat{t}}$}
  
  \State{$\widehat{\nabla \C} (K_t) \gets - \frac{1}{\sigma (1 - \gamma)} \eta_{\hat{t}} x_{\hat{t}}^\top Q^{K_t} (x_{\hat{t}}, u_{\hat{t}}) $} \label{line: alg_LQR_policy_gradient_6}
  
  \State {$K_{t+1} \gets K_t - \alpha_t \widehat{\nabla \C} (K_t)$} \label{line: alg_LQR_policy_gradient_7} \EndFor
  \Return $K_T$
\end{algorithmic}
\end{algorithm}

Before we state the next result, note that one can compute 
\begin{equation}
\nabla \C (K) = 2 ((R + \gamma B^\top P_K B)K - \gamma B^\top P_K A) \E_{x_0 \sim \D} \left[\sum_{t=0}^{\infty} \gamma^t x_t x_t^\top \right]; \label{eq: actual gradient}
\end{equation}
a proof can be found in~\cite{MF-RG-SK-MM:18} for the undiscounted case, where $\gamma = 1$,  and in~\cite{DM-AP-KB-KK-PLB-MJW:18} for the discounted case. The following proposition plays a key role in our constructions. 
\begin{proposition}
    \label{prop: gradient estimate expectation}
    Suppose $u_{\hat{t}} \sim \N (- K x_{\hat{t}}, \sigma^2 I_m)$ as before. Then for any given $K$,
    \begin{equation}
        \E [\widehat{\nabla \C} (K)] = \nabla \C (K).
    \end{equation}    
\end{proposition} 
\begin{proof}
    Following \eqref{eq: gradient estimate practical formulation}, 
    \begin{align}
        &\E [\widehat{\nabla \C} (K)] \cr
        = \ &\E_{\hat{t} \sim \mu_\gamma} \left[ \E_{x_0 \sim \D} \left[ \E_{\eta_{\hat{t}} \sim \N (0, I_m)} \left[ \widehat{\nabla \C} (K) \big| \hat{t}, x_0 \right] \bigg| \hat{t} \right] \right] \cr
        \overset{\mathrm{(i)}}{=} \ &\E_{\hat{t} \sim \mu_\gamma} \left[ \E_{x_0 \sim \D} \left[ - \frac{1}{\sigma^2 (1 - \gamma)} \E_{\eta_{\hat{t}} \sim \N (0, I_m)} \left[ Q (x_{\hat{t}}, - K x_{\hat{t}} + \sigma \eta_{\hat{t}}) (\sigma \eta_{\hat{t}}) \big| \hat{t}, x_0 \right] x_{\hat{t}}^\top \bigg| \hat{t} \right] \right] \cr
        \overset{\mathrm{(ii)}}{=} \ &\frac{1}{1 - \gamma} \E_{\hat{t} \sim \mu_\gamma} \left[ \E_{x_0 \sim \D} \left[ \E_{\eta_{\hat{t}} \sim \N (0, I_m)} \left[ - \nabla_{u} Q^{K} (x_{\hat{t}}, u) \bigg|_{u = -K x_{\hat{t}} + \sigma \eta_{\hat{t}}} \big| \hat{t}, x_0 \right] x_{\hat{t}}^\top \bigg| \hat{t} \right] \right], \label{eq: policy gradient proof before Stein}
    \end{align} 
    where (i) follows from $x_{\hat{t}}$ being determined when given $x_0$ and $\hat{t}$, and (ii) from Stein's lemma~\cite{CMS:81}. Using~\eqref{eq: Q function definition}, we compute 
    \begin{align*}
        \nabla_{u} Q^{K} (x_{\hat{t}}, u) =& \nabla_{u} \left( x_{\hat{t}}^\top Q x_{\hat{t}} + u^\top R u + \gamma (A x_{\hat{t}} + B u)^\top P_K (A x_{\hat{t}} + B u) \right) \\
        =& 2 R u + 2 \gamma B^\top P_K B u + 2 \gamma B^\top P_K A x_{\hat{t}},
    \end{align*}
    which evaluated at $u = -K x_{\hat{t}} + \sigma \eta_{\hat{t}}$ yields
    \begin{equation*}
        \nabla_{u} Q^{K} (x_{\hat{t}}, u) \bigg|_{u = -K x_{\hat{t}} + \sigma \eta_{\hat{t}}} = 2 \left( (R + \gamma B^\top P_K B)(-K x_{\hat{t}} + \sigma \eta_{\hat{t}}) + \gamma B^\top P_K A x_{\hat{t}} \right).
    \end{equation*}
    Substituting in \eqref{eq: policy gradient proof before Stein}, we obtain
    \begin{align*}
        &\E [\widehat{\nabla \C} (K)] \\
        = \ &\frac{1}{1 - \gamma} \E_{\hat{t} \sim \mu_\gamma} \left[ \E_{x_0 \sim \D} \left[ 2 \left( (R + \gamma B^\top P_K B)K - \gamma B^\top P_K A \right) x_{\hat{t}} x_{\hat{t}}^\top \bigg| \hat{t} \right] \right] \\
        = \ &\frac{2}{1 - \gamma} \E_{\hat{t} \sim \mu_\gamma} \left[ \left( (R + \gamma B^\top P_K B)K - \gamma B^\top P_K A \right) (A - BK)^{\hat{t}} \E_{x_0 \sim \D} [x_0 x_0^\top] \left( (A - BK)^{\hat{t}} \right)^\top \right] \\
        = \ &2 \left( (R + \gamma B^\top P_K B)K - \gamma B^\top P_K A \right) \sum_{t = 0}^{\infty} \gamma^{t} (A - BK)^{t} \E_{x_0 \sim \D} [x_0 x_0^\top] \left( (A - BK)^{t} \right)^\top \\
        \overset{\mathrm{(i)}}{=} \ &2 \left( (R + \gamma B^\top P_K B)K - \gamma B^\top P_K A \right) \E_{x_0 \sim \D} \left[ \sum_{t = 0}^{\infty} \gamma^{t} (A - BK)^{t} x_0 x_0^\top \left( (A - BK)^{t} \right)^\top \right] \\
        \overset{\mathrm{(ii)}}{=} \ &2 \left( (R + \gamma B^\top P_K B)K - \gamma B^\top P_K A \right) \E_{x_0 \sim \D} \left[ \sum_{t = 0}^{\infty} \gamma^{t} x_t x_t^\top \right] \\
        \overset{\mathrm{(iii)}}{=} \ &\nabla \C (K),
    \end{align*}
    where (i) is done by utilizing the linearity of expectation along with replacing $\hat{t}$ by $t$ as it is just a sum variable from that equation forward, (ii) is due to $x_t = (A - BK)^t x_0$, and (iii) follows from~\eqref{eq: actual gradient}.
\end{proof}

\begin{remark}[Extension beyond LQR]\footnote{Notation in this remark follows \cite{DS-GL-NH-TD-DW-MR:14} rather than the LQR‑specific symbols used elsewhere in the paper: $s\!\in\!\mathcal S$ is the state, 
$a\!\in\!\mathbb R^{m}$ the action, 
$\mu_\theta$ the (deterministic) policy, 
$\rho^{\mu}$ the (improper) discounted state distribution, 
$Q^{\mu}$ the action–value function, $J(\theta)=\E\bigl[\sum_{t\ge0}\gamma^{t}r_t\bigr]$ the performance objective,
and $\eta \sim \mathcal{N}(0,I_m)$ the Gaussian exploration noise.} \label{rem: extension beyond LQR}
The construction in~\eqref{eq: gradient estimate practical formulation} is not automatically restricted to linear-quadratic control, but instead relies on the following assumption on the $Q$-values which is satisfied in the LQR setting.
Suppose the action-value function satisfies
\begin{equation} \label{eq:q-structure}
  Q^{\mu}(s,a)=a^{\top}H(s)a+b(s)^{\top}a+c(s)
\end{equation}
with $H(s)=H(s)^{\!\top} \in \mathbb{R}^{m \times m}$.
Then $\nabla_a Q^{\mu}(s,a)=2H(s)a+b(s)$ is \emph{affine} in~$a$.
Let $\eta\sim\mathcal N(0,I_m)$, independent of~$s$, and write
$a_\theta(s)=\mu_\theta(s)$.  
For
\(f(\eta)\coloneqq Q^{\mu}\!\bigl(s,a_\theta(s)+\sigma\eta\bigr)\)
we have
\(
  \nabla_\eta f(\eta)=\sigma\,\nabla_a Q^{\mu}\bigl(s,a_\theta(s)+\sigma\eta\bigr).
\)
Stein’s lemma~\cite{CMS:81} yields
\[
  \E_\eta \left[\eta f(\eta)\right] = \E_\eta \left[ \nabla_{\eta} f(\eta)\right] = \sigma\,\E_\eta\!\bigl[\nabla_a Q^{\mu}(s,a_\theta(s)+\sigma\eta)\bigr] =
  \sigma\,\nabla_a Q^{\mu}(s,a_\theta(s)),
\]
where the last equality uses linearity of the integrand in~$a$.
Hence
\[
  \E_\eta\!\bigl[\sigma^{-1}
     Q^{\mu}(s,a_\theta(s)+\sigma\eta)\,\eta\bigr]
  =\nabla_a Q^{\mu}(s,a_\theta(s)).
\]

Combining this with the deterministic policy gradient
of~\cite[Theorem~1]{DS-GL-NH-TD-DW-MR:14},
\[
  \nabla_\theta J(\theta)
  \;=\;
  \E_{s\sim\rho^{\mu}}\!\bigl[
       \nabla_\theta \mu_\theta(s)^{\!\top}\,
       \nabla_a Q^{\mu}(s,a)\bigr]_{a=a_\theta(s)},
\]
gives the unbiased estimator
\[
  \widehat{\nabla J}(\theta)
  \;=\;
  \nabla_\theta \mu_\theta(s)^{\!\top}\,
  \bigl[\sigma^{-1}Q^{\mu}(s,a_\theta(s)+\sigma\eta)\,\eta\bigr],
  \qquad
  s\sim\rho^{\mu},\;
  \eta\sim\mathcal N(0,I_m).
\]
\noindent\textit{Linear actor:}
If $a_\theta(s)=\Theta s$, then
\(
  \nabla_\theta a_\theta(s)=I_m\!\otimes s^{\top}
\),
so that
\[
  \nabla_\theta a_\theta(s)^{\!\top}\,
  \nabla_a Q^{\mu}(s,a)
  \;=\;
  (I_m\!\otimes s)\,\nabla_a Q^{\mu}(s,a)
  \;=\;
  \operatorname{vec}\!\bigl[\nabla_a Q^{\mu}(s,a)\,s^{\top}\bigr],
\]
using the identity
\( (I_m\!\otimes s)u=\operatorname{vec}(us^{\top}) \).
Unvectorising recovers the familiar matrix form
\( \nabla_\Theta J(\Theta)
   =\E_{s\sim\rho^{\mu}}
      [\nabla_a Q^{\mu}(s,a)\,s^{\top}]\),
and the estimator in~\eqref{eq: gradient estimate practical formulation}
follows by substituting the Stein-based replacement for
\(\nabla_a Q^{\mu}\). 

It may therefore be possible to extend the gradient estimators discussed here beyond the LQR setting by establishing that equation~\eqref{eq:q-structure} holds (perhaps approximately) for various classes of nonlinear systems.
\end{remark}

Before moving on to the next result, we define the undiscounted cost
\begin{equation}
    \Cund (K) = \E_{x_0 \sim \D} \left[ \sum_{t=0}^{\infty} (x_t^\top Q x_t + u_t^\top R u_t) \right],
\end{equation}
subject to~\eqref{eq:sys-no-noise}.

\begin{lemma} \label{lem: gamma condition}
    Suppose $K_0$ is stable and
    suppose that 
    \[
    \gamma \in \left(1-\frac{\sigma_{\mathrm{min}} (Q)}{11 \Cund (K_0)},1 \right). 
    \]
    Then 
    \begin{equation}
        \sup_{K \in \G^{\text{lqr}}} \rho (A-BK) \leq \frac{1}{\sqrt{\gamma}} \sqrt{1 - \frac{\sigma_{\mathrm{min}} (Q)}{10 \C(K_0) + \C(K^*)}}; \label{eq: rho upper bound}
    \end{equation}
    in particular, the set $\G^{\text{lqr}}$ in~\eqref{eq: G_LQR} only contains stable policies.
\end{lemma}
This result shows that this assumption on $\gamma$ ensures stability of the policies in the $\G^\text{lqr}$ set. When $\gamma$ is small, the cost becomes heavily concentrated on early time steps and places less emphasis on the asymptotic behavior, which can lead to the optimal policy being unstable~\cite[Example~1]{RP-LB-DN-JD:17}. The assumption on $\gamma$ serves to exclude such degenerate behavior by making instability more costly. Moreover, this condition on $\gamma$ is tied to the particular definition of $\G^{\text{lqr}}$, and can be relaxed by tightening the required upper bound on the optimality gap in its definition---provided the resulting set still allows the analysis to achieve a sufficiently high confidence level. A more detailed discussion is given in Remark~\ref{rem: gamma condition flexibility} in Appendix~\ref{app: prob of failure}.

Before we provide the proof, we point out that the condition on stability of $ K_0$ readily implies that $\Cund (K_0)$ is finite. 
\begin{proof}
    Suppose $\tilde{K}$ satifies $\rho (A - B \tilde{K}) \geq 1$. Then we have
    \begin{align}
        \C(\tilde{K}) &= \E_{x_0 \sim \D} \left[ \sum_{t = 0}^{\infty} \gamma^t (x_t^\top Q x_t + u_t^\top R u_t) \right] \cr
        &\geq \sum_{t = 0}^{\infty} \gamma^t \sigma_{\mathrm{min}} (Q) \E \lVert (A-B\tilde{K})^t x_0 \rVert^2 \cr
        &= \sum_{t = 0}^{\infty} \gamma^t \sigma_{\mathrm{min}} (Q) \E [\tr ( ((A-B\tilde{K})^{t})^\top (A-B\tilde{K})^{t} x_0 x_0^\top )] \cr
        &\overset{\mathrm{(i)}}{=} \sum_{t = 0}^{\infty} \gamma^t \sigma_{\mathrm{min}} (Q) \lVert (A-B\tilde{K})^{t} \rVert_F^2 \cr
        &\geq \sum_{t = 0}^{\infty} \gamma^t \sigma_{\mathrm{min}} (Q) \rho ((A-B\tilde{K})^{t})^2 \cr
        &\overset{\mathrm{(ii)}}{\geq} \sum_{t = 0}^{\infty} \gamma^t \sigma_{\mathrm{min}} (Q) \cr
        &= \frac{\sigma_{\mathrm{min}} (Q)}{1-\gamma}, \label{eq: cost lower bound}
    \end{align}
    where (i) comes from the linearity of expectation along with the assumption on the noise from \eqref{eq: noise_assumption}, and (ii) follows from the instability of $\tilde{K}$ and that $\rho (A^t) = (\rho(A))^t$ which holds for any square matrix A. Now as a result of this, if we also show $\sup_{K \in \G^{\text{lqr}}} \C(K) < \frac{\sigma_{\mathrm{min}} (Q)}{1-\gamma}$, we have proved stability of every $K$ in the set $\G^{\text{lqr}}$. We do so as follows:
    \begin{align*}
        \frac{\sigma_{\mathrm{min}} (Q)}{1-\gamma} \overset{\mathrm{(i)}}{>} 11 \Cund (K_0)  \overset{\mathrm{(ii)}}{\geq} 11 \C(K_0) \geq 10 \C(K_0) + \C(K^*) \overset{\mathrm{(iii)}}{\geq} \sup_{K \in \G^{\text{lqr}}} \C(K),
    \end{align*}
    where (i) comes from the assumption on $\gamma$, (ii) from the fact that for a given policy, the undiscounted cost is not less than the discounted cost, and (iii) from the definition of the set $\G^{\text{lqr}}$ from \eqref{eq: G_LQR}. This proves the second claim. 
    
    For the first part, since for any $K \in \G^{\text{lqr}}$ we have that $\rho (A-BK) < 1$, we conclude that
    \begin{align*}
        \C(K) &= \E_{x_0 \sim \D} \left[ \sum_{t = 0}^{\infty} \gamma^t (x_t^\top Q x_t + u_t^\top R u_t) \right] \\
        &\overset{\mathrm{(i)}}{\geq} \sum_{t = 0}^{\infty} \gamma^t \sigma_{\mathrm{min}} (Q) \rho ((A-BK)^{t})^2 \\
        &= \sigma_{\mathrm{min}} (Q) \sum_{t = 0}^{\infty} (\gamma (\rho (A-BK))^2)^t \\
        &\overset{\mathrm{(ii)}}{=} \frac{\sigma_{\mathrm{min}} (Q)}{1 - \gamma (\rho (A-BK))^2},
    \end{align*}
    where (i) is done the same way as \eqref{eq: cost lower bound} and (ii) follows from $\gamma (\rho (A-BK))^2 < 1$ for $K \in \G^{\text{lqr}}$. As a result, for $K \in \G^{\text{lqr}}$, we have that
    \begin{align*}
        1 - \gamma (\rho (A-BK))^2 \geq \frac{\sigma_{\mathrm{min}} (Q)}{\C(K)} \Rightarrow \\
        \gamma (\rho (A-BK))^2 \leq 1 - \frac{\sigma_{\mathrm{min}} (Q)}{\C(K)} \Rightarrow \\
        \rho (A-BK) \leq \frac{1}{\sqrt{\gamma}} \sqrt{1 - \frac{\sigma_{\mathrm{min}} (Q)}{\C(K)}},
    \end{align*}
    which after taking a supremum gives
    \begin{align*}
    \sup_{K \in \G^{\text{lqr}}} \rho (A-BK) \leq \frac{1}{\sqrt{\gamma}} \sup_{K \in \G^{\text{lqr}}} \sqrt{1 - \frac{\sigma_{\mathrm{min}} (Q)}{\C(K)}} = \frac{1}{\sqrt{\gamma}} \sqrt{1 - \frac{\sigma_{\mathrm{min}} (Q)}{10 \C(K_0) + \C(K^*)}},
    \end{align*}
    concluding the proof.
\end{proof}
We next introduce a high probability upper bound on our gradient estimate on any $K \in \G^\text{lqr}$.

\begin{lemma} \label{lem: gradient estimate bounds}
    Suppose $\delta \in (0, \frac{1}{e}]$, and $\gamma$ is chosen as in Lemma~\ref{lem: gamma condition}. 
    Then for any $K \in \G^{\text{lqr}}$, we have that
    \begin{equation}
        \lVert \widehat{\nabla \C} (K) \rVert_F \leq \frac{\xi_3}{1 - \gamma} \left(\log \frac{1}{\delta}\right)^{3/2}
    \end{equation}
    with probability at least $1-\delta$, where 
$\xi_1, \xi_2, \xi_3 \in \real$ are given by
    \begin{align}
        \xi_1 &:= \left( \lVert Q \rVert + 2 \lVert R \rVert \widetilde{c_{K_1}}^2 + 2 \gamma (10 \C(K_0) + \C(K^*)) \right) e^3 n^3 \Bar{\K}^3  C_m^{3/2} \\
        \xi_2 &:= \left( 2 \lVert R \rVert + 2 \gamma \lVert B \rVert^2 (10 \C(K_0) + \C(K^*)) \right) e n \Bar{\K} C_m^{1/2} \\
        \xi_3 &:= \frac{1}{\sigma} \left(\xi_1 5^{1/2} m^{1/2}\right) + \sigma \left( \xi_2 5^{3/2} m^{3/2} \right), \label{eq: xi_3 definition}
    \end{align}
    where $\Bar{\K}$ is a positive constant. Moreover, 
    \begin{equation}
        \E \lVert \widehat{\nabla \C} (K) \rVert_F^2 \leq \frac{\xi_4}{(1 - \gamma)^2}, \label{eq: gradient_estimate_variance_bound}
    \end{equation}
    where    
    \begin{align}\label{eq: xi_4 definition}
        \xi_4 &:= \frac{1}{\sigma^2} \xi_1^2 m + 2 \xi_1 \xi_2 m (m+2) + \sigma^2 \xi_2^2 m (m+2) (m+4).
    \end{align}
\end{lemma}
\begin{proof}
    Using the formulation of $\widehat{\nabla \C} (K)$ derived in \eqref{eq: gradient estimate practical formulation}, we have
    \begin{align}
        \lVert \widehat{\nabla \C} (K) \rVert_F =& \left| \left| \frac{1}{\sigma (1-\gamma)} \eta_{\hat{t}} x_{\hat{t}}^\top  Q^{K} (x_{\hat{t}}, -K x_{\hat{t}} + \sigma \eta_{\hat{t}}) \right| \right|_F \cr
        \leq & \frac{1}{\sigma (1-\gamma)} \lVert \eta_{\hat{t}} \rVert \lVert x_{\hat{t}} \rVert  Q^{K} (x_{\hat{t}}, -K x_{\hat{t}} + \sigma \eta_{\hat{t}}). \label{eq: gradient norm initial bound}
    \end{align}
    First, note that
    \begin{equation}
        \lVert x_{\hat{t}} \rVert = \lVert (A-BK)^{\hat{t}} x_0 \rVert \leq \lVert (A-BK)^{\hat{t}} \rVert \lVert x_0 \rVert \overset{\mathrm{(i)}}{\leq} \sup_{t \geq 0} \lVert (A-BK)^{t} \rVert C_m^{1/2}, \label{eq: estimate bound tmp}
    \end{equation}
    where (i) follows from the assumption on the initial state noise mentioned in \eqref{eq: noise_assumption}. 

    \begin{sublemma} \label{sublem: sup_sup finite}
        We have that 
    \begin{equation}
    \sup_{K \in \G^{\text{lqr}}} \sup_{t \geq 0} \lVert (A-BK)^{t} \rVert \label{eq: sup_sup_norm direct}
    \end{equation}
    is finite.
    \end{sublemma}
   \emph{Proof of Sublemma~\ref{sublem: sup_sup finite}.}
    We start by arguing that $\G^{\text{lqr}}$ is a compact set. First, note that since $\| K \| \leq \widetilde{c_{K_1}}$ (see Remark~\ref{rem: regularity parameters}) for any $K \in \G^{\text{lqr}}$, the set $\G^{\text{lqr}}$ is bounded. Secondly, since $\C (K)$ is locally Lipschitz in $\G^{\text{lqr}}$, it is also continuous, and hence, by the definition of $\G^{\text{lqr}}$ in~\eqref{eq: G_LQR}, we have that $\G^{\text{lqr}}$ is the pre-image of the closed interval $[0, 10 \C (K_0) + \C (K^*)]$ under a continuous map $\map{\C}{\G^{\text{lqr}}}{\R}$, implying $\G^{\text{lqr}}$ is closed as well. As a result of this, we have that $\G^{\text{lqr}}$ is compact. Now we move on to show why \eqref{eq: sup_sup_norm direct} is finite. 
    
    First, let us define
    \begin{align*}
        S(x_0; K) := \sum_{t = 0}^{\infty} \| x_t \|^2,
    \end{align*}
    where $x_{t+1} = (A-BK) x_t$. Moreover, we let
    \begin{align*}
    S(K) &:= \E_{x_0 \sim \D} S(x_0; K) \\
    &= \E_{x_0 \sim \D} \left[ \sum_{t = 0}^{\infty} \| x_t \|^2 \right] \\
    &= \E_{x_0 \sim \D} \left[ \sum_{t=0}^{\infty} \| (A - BK)^t x_0 \|^2 \right] \\
    &= \sum_{t=0}^{\infty} \E_{x_0 \sim \D} \left[ \tr \left( \left((A - BK)^t \right)^\top (A - BK)^t x_0 x_0^\top \right) \right] \\
    &= \sum_{t=0}^{\infty} \lVert (A-BK)^t \rVert_F^2 \\
    &\geq \sum_{t=0}^{\infty} \lVert (A-BK)^t \rVert^2,
    \end{align*}
    which after taking the square root of both sides gives
    \begin{align*}
        \sqrt{S(K)} &\geq \sqrt{\sum_{t=0}^{\infty} \lVert (A-BK)^t \rVert^2} \\
        &\geq \sup_{t \geq 0} \| (A-BK)^t \|.
    \end{align*}
    As a result, we have that
    \[
    \sup_{t \geq 0} \| (A-BK)^t \| \leq \sqrt{S(K)},
    \]
    which after taking a supremum over $\G^{\text{lqr}}$ yields
    \begin{equation}
    \sup_{K \in \G^{\text{lqr}}} \sup_{t \geq 0} \| (A-BK)^t \| \leq \sup_{K \in \G^{\text{lqr}}} \sqrt{S(K)}. \label{eq: sup_sup upper bound}
    \end{equation}
    Now it suffices to show $\sup_{K \in \G^{\text{lqr}}} \sqrt{S(K)}$ is finite, which we prove by contradiction. Suppose that this is not the case. Therefore, there exists a sequence $\{K_j\}_{j=1}^{\infty}$ such that $\sqrt{S(K_j)} \xrightarrow[]{j \to \infty} \infty$. By compactness, we can pick a convergent subsequence whose limit we denote by $\bar{K}$. We will abuse notation and henceforth use $K_j$ to refer to the subsequence; observe that $K_j$ should also satisfy $\sqrt{S(K_j)} \xrightarrow[]{j \to \infty} \infty$. 

    Now since $\bar{K} \in \G^{\text{lqr}}$, we have from Lemma~\ref{lem: gamma condition} that $A - B \bar{K}$ is strictly stable, and thus, there exists a Lyapunov function $V(x) = x^\top \bar{P} x$ where $\bar{P}$ is a positive definite matrix that satisfies
    \[
    (A - B \Bar{K})^\top \bar{P} (A - B \Bar{K}) - \bar{P} = -I.
    \]
    Therefore, for $j$ large enough,
    \[
    (A - B K_j)^\top \bar{P} (A - B K_j) - \bar{P} \preceq -(1/2)I.
    \]
    Then
    \begin{align*}
    V ((A - B K_j) x) - V(x) &= x^T (A - B K_j) \bar{P} (A - B K_j) x - x^T \bar{P} x  \\
    &\leq -(1/2) \| x \|^2 \\
    &= - \frac{1}{2 \lambda_{\mathrm{max}} (\bar{P})} \left( \lambda_{\mathrm{max}} (\bar{P}) \| x \|^2 \right) \\
    &\overset{\mathrm{(i)}}{\leq} - \frac{1}{2 \lambda_{\mathrm{max}} (\bar{P})} V(x),
    \end{align*}
    where (i) is due to the fact that $V(x) \leq \lambda_{\mathrm{max}} (\bar{P}) \| x \|^2$. Thus,
    \begin{equation}
    V((A - B K_j) x) \leq \left( 1 - \frac{1}{2 \lambda_{\mathrm{max}} (\bar{P})} \right) V(x). \label{eq: Lyapunov V exp}
    \end{equation}
    As a result, we have that
    \begin{align*}
        S(x_0; K_j) &= \sum_{t = 0}^{\infty} \| x_t \|^2 \\
        &\overset{\mathrm{(i)}}{\leq} \frac{1}{\lambda_{\mathrm{min}} (\bar{P})} \sum_{t = 0}^{\infty} V((A - B K_j)^t x_0) \\
        &\overset{\mathrm{(ii)}}{\leq} \frac{1}{\lambda_{\mathrm{min}} (\bar{P})} \sum_{t = 0}^{\infty} \left( 1 - \frac{1}{2 \lambda_{\mathrm{max}} (\bar{P})} \right)^t V(x_0) \\
        &\leq \frac{2 \lambda_{\mathrm{max}} (\bar{P})}{\lambda_{\mathrm{min}} (\bar{P})} V(x_0) \\
        &\leq \frac{2 \lambda_{\mathrm{max}}^2 (\bar{P})}{\lambda_{\mathrm{min}} (\bar{P})} \| x_0 \|^2,
    \end{align*}
    where (i) follows from $V(x) \geq \lambda_{\mathrm{min}} (\bar{P}) \| x \|^2$ and (ii) from~\eqref{eq: Lyapunov V exp}. Now taking an expectation over $x_0 \sim \D$ yields
    \begin{align*}
    S(K_j) &\leq \frac{2 \lambda_{\mathrm{max}}^2 (\bar{P})}{\lambda_{\mathrm{min}} (\bar{P})} \E_{x_0 \sim \D} \| x_0 \|^2 \\
    &= \frac{2 \lambda_{\mathrm{max}}^2 (\bar{P})}{\lambda_{\mathrm{min}} (\bar{P})} \E_{x_0 \sim \D} \tr (x_0 x_0^\top) \\
    &= \frac{2 \lambda_{\mathrm{max}}^2 (\bar{P})}{\lambda_{\mathrm{min}} (\bar{P})} \tr \left( \E_{x_0 \sim \D} [x_0 x_0^\top] \right) \\
    &= \frac{2 \lambda_{\mathrm{max}}^2 (\bar{P})}{\lambda_{\mathrm{min}} (\bar{P})} \tr (I_n) \\
    &= \frac{2 n \lambda_{\mathrm{max}}^2 (\bar{P})}{\lambda_{\mathrm{min}} (\bar{P})},
    \end{align*}
    and hence, 
    \[
    \sqrt{S(K_j)} \leq \sqrt{\frac{2 n \lambda_{\mathrm{max}}^2 (\bar{P})}{\lambda_{\mathrm{min}} (\bar{P})}},
    \]
    which is finite, resulting in a contradiction,
    concluding the proof of Sublemma~\ref{sublem: sup_sup finite}. \oprocend
    
    We now continue with the proof of Lemma~\ref{lem: gradient estimate bounds}. Let us first make a remark. By the Kreiss matrix theorem \cite{JRL-LNT:84,MNS:91}, we have that
    \begin{equation}
        \K (A-BK) \leq \sup_{t \geq 0} \lVert (A-BK)^{t} \rVert \leq e \ n \ \K (A-BK). \label{eq: Kreiss matrix theorem}
    \end{equation}  
    Consequently, we can define the following constant
    \begin{equation}
        \Bar{\K} := \sup_{K \in \G^{\text{lqr}}} \K (A-BK), \label{eq: K_bar definition}
    \end{equation} 
    which is finite as a result of \eqref{eq: Kreiss matrix theorem} and Sublemma~\ref{sublem: sup_sup finite}.
    Combining \eqref{eq: Kreiss matrix theorem} and \eqref{eq: K_bar definition} with~\eqref{eq: estimate bound tmp} gives
    \begin{equation}
        \lVert x_{\hat{t}} \rVert \leq e \ n \ C_m^{1/2} \ \K (A-BK) \leq e \ n \ C_m^{1/2} \ \Bar{\K}, \label{eq: state_norm_bound}
    \end{equation}
    for any $\hat{t} \geq 0$.
    Moreover,
     \begin{align}
        Q^{K} (x_{\hat{t}}, -K x_{\hat{t}} + \sigma \eta_{\hat{t}}) =& x_{\hat{t}}^\top Q x_{\hat{t}} + (-K x_{\hat{t}} + \sigma \eta_{\hat{t}})^\top R (-K x_{\hat{t}} + \sigma \eta_{\hat{t}}) \cr
        &+ \gamma ((A - BK) x_{\hat{t}} + \sigma B \eta_{\hat{t}})^\top P_K ((A - BK) x_{\hat{t}} + \sigma B \eta_{\hat{t}}) \cr
        \overset{\mathrm{(i)}}{\leq}& \lVert Q \rVert e^2 n^2 \Bar{\K}^2 C_m + \lVert R \rVert \lVert -K x_{\hat{t}} + \sigma \eta_{\hat{t}} \rVert^2 + \gamma \lVert P_K \rVert \lVert x_{\hat{t}+1} + \sigma B \eta_{\hat{t}} \rVert^2 \cr
        \overset{\mathrm{(ii)}}{\leq}& \lVert Q \rVert e^2 n^2 \Bar{\K}^2 C_m + 2 \lVert R \rVert \left( \lVert K \rVert^2 \lVert x_{\hat{t}} \rVert^2 + \sigma^2 \lVert \eta_{\hat{t}} \rVert^2 \right) \cr
        &+ 2 \gamma \C(K) \left( \lVert x_{\hat{t}+1} \rVert^2 + \sigma^2 \lVert B \rVert^2 \lVert \eta_{\hat{t}} \rVert^2 \right) \cr
        \overset{\mathrm{(iii)}}{\leq}& \left( \lVert Q \rVert + 2 \lVert R \rVert \widetilde{c_{K_1}}^2 + 2 \gamma (10 \C(K_0) + \C(K^*)) \right) e^2 n^2 \Bar{\K}^2 C_m \cr
        &+ \left( 2 \sigma^2 \lVert R \rVert + 2 \gamma \sigma^2 \lVert B \rVert^2 (10 \C(K_0) + \C(K^*)) \right) \lVert \eta_{\hat{t}} \rVert^2, \label{eq: Q-function bound}
    \end{align}
    where (i) follows from \eqref{eq: state_norm_bound}, (ii) from $\lVert P_K \rVert \leq \tr (P_K)$ along with $\tr (P_K) = \C(K)$ as shown in~\eqref{eq: cost defined via P_K}, and (iii) from the fact that $\lVert K \rVert \leq \widetilde{c_{K_1}}$ for any $K \in \G^{\text{lqr}}$ (see Remark~\ref{rem: regularity parameters}) along with reapplying \eqref{eq: state_norm_bound} and utilizing the upper bound obtained on $\C(K)$ by the definition of the set $\G^{\text{lqr}}$. Now applying the derived bounds~\eqref{eq: state_norm_bound} and~\eqref{eq: Q-function bound} on~\eqref{eq: gradient norm initial bound}, we conclude that 
    \begin{align}
        \lVert \widehat{\nabla \C} (K) \rVert_F \leq & \frac{\left( \lVert Q \rVert + 2 \lVert R \rVert \widetilde{c_{K_1}}^2 + 2 \gamma (10 \C(K_0) + \C(K^*)) \right) e^3 n^3 \Bar{\K}^3 C_m^{3/2}}{\sigma (1 - \gamma)} \lVert \eta_{\hat{t}} \rVert \cr
        &+ \frac{ \sigma^2 \left( 2 \lVert R \rVert + 2 \gamma \lVert B \rVert^2 (10 \C(K_0) + \C(K^*)) \right) e n \Bar{\K} C_m^{1/2}}{\sigma (1 - \gamma)} \lVert \eta_{\hat{t}} \rVert^3 \cr
        =& \frac{1}{1 - \gamma} \left( \frac{1}{\sigma} \xi_1 \lVert \eta_{\hat{t}} \rVert + \sigma \xi_2 \lVert \eta_{\hat{t}} \rVert^3 \right). \label{eq: gradient estimate norm bound before LM}
    \end{align}
    Furthermore, since $\eta_{\hat{t}} \sim \N (0, I_m)$ for any $\hat{t}$, $\lVert \eta_{\hat{t}} \rVert^2$ is distributed according to the chi-squared distribution with $m$ degrees of freedom ($\lVert \eta_{\hat{t}} \rVert^2 \sim \chi^2 (m)$ for any $\hat{t}$). Therefore, the standard~\cite{BL-PM:00} bounds suggest that for arbitrary $y>0$, we have that
    \begin{equation}
        \PP \{ \lVert \eta_{\hat{t}} \rVert^2 \geq m + 2 \sqrt{m y} + 2y \} \leq e^{-y}. \label{eq: Lauren-Massart Chi concentration}
    \end{equation}
    Now since by our assumption $0 < \delta \leq 1/e$, it holds that $y = m \log \frac{1}{\delta} \geq m$ and thus
    \begin{align*}
        \PP \{ \lVert \eta_{\hat{t}} \rVert^2 \geq 5y \} \leq \PP \{ \lVert \eta_{\hat{t}} \rVert^2 \geq m + 2 \sqrt{m y} + 2y \} \leq e^{-y},
    \end{align*}
    which after substituting $y$ with its value $m \log \frac{1}{\delta}$ gives
    \begin{align*}
        \PP \{ \lVert \eta_{\hat{t}} \rVert^2 \geq 5 m \log \frac{1}{\delta} \} \leq e^{-m \log \frac{1}{\delta}} = \delta^m \leq \delta.
    \end{align*}
    As a result, we have $\lVert \eta_{\hat{t}} \rVert \leq 5^{1/2} m^{1/2} (\log \frac{1}{\delta})^{1/2}$ and consequently 
    \[
    \lVert \eta_{\hat{t}} \rVert^3 \leq 5^{3/2} m^{3/2} (\log \frac{1}{\delta})^{3/2}
    \]
    with probability at least $1-\delta$, which after applying on \eqref{eq: gradient estimate norm bound before LM} yields
    \begin{align*}
    \lVert \widehat{\nabla \C} (K) \rVert_F \leq & \frac{1}{1 - \gamma} \left(\frac{1}{\sigma} \xi_1 5^{1/2} m^{1/2} \left(\log \frac{1}{\delta}\right)^{1/2} + \sigma \xi_2 5^{3/2} m^{3/2} \left(\log \frac{1}{\delta}\right)^{3/2} \right) \\
    \leq & \frac{1}{1 - \gamma} \left(\frac{1}{\sigma} \xi_1 5^{1/2} m^{1/2} + \sigma \xi_2 5^{3/2} m^{3/2} \right) \left(\log \frac{1}{\delta}\right)^{3/2} \\
    =& \frac{\xi_3}{1 - \gamma} \left(\log \frac{1}{\delta}\right)^{3/2},
    \end{align*}
    proving the first claim.

    As for the second claim, note that using \eqref{eq: gradient estimate norm bound before LM}, we have
    \begin{align}
        \lVert \widehat{\nabla \C} (K) \rVert_F^2 \leq \frac{1}{(1 - \gamma)^2} \left( \frac{1}{\sigma^2} \xi_1^2 \lVert \eta_{\hat{t}} \rVert^2 + 2 \xi_1 \xi_2 \lVert \eta_{\hat{t}} \rVert^4 + \sigma^2 \xi_2^2 \lVert \eta_{\hat{t}} \rVert^6\right). \label{eq: before chi moments}
    \end{align}
    Now since $\lVert \eta_{\hat{t}} \rVert \sim \chi (m)$ whose moments are known, taking an expectation on both sides of \eqref{eq: before chi moments} results in
    \begin{align*}
        \E \lVert \widehat{\nabla \C} (K) \rVert_F^2 \leq & \frac{1}{(1 - \gamma)^2} \left( \frac{1}{\sigma^2} \xi_1^2 \E \lVert \eta_{\hat{t}} \rVert^2 + 2 \xi_1 \xi_2 \E \lVert \eta_{\hat{t}} \rVert^4 + \sigma^2 \xi_2^2 \E \lVert \eta_{\hat{t}} \rVert^6 \right) \\
        =& \frac{1}{(1 - \gamma)^2} \left( \frac{1}{\sigma^2} \xi_1^2 m + 2 \xi_1 \xi_2 m (m+2) + \sigma^2 \xi_2^2 m (m+2) (m+4) \right) \\
        =& \frac{\xi_4}{(1 - \gamma)^2},
    \end{align*}
    concluding the proof.
\end{proof}

Following Lemma~\ref{lem: gradient estimate bounds}, we now define the following event for each iteration $t$ of Algorithm~\ref{alg: LQR_policy_gradient}:
\begin{equation}
    \mathcal{A}_t = \left\{ \| \widehat{\nabla \C} (K_t) \|_F \leq \frac{\xi_3}{1 - \gamma} \left( \log{\frac{1}{\delta}} \right)^{3/2} \right\}. \label{eq: A_t definition}
\end{equation}
Having this, we introduce the following lemma:
\begin{lemma} \label{lem: gradient estimate conditioned}
    Suppose $\delta \in (0,e^{-3/2}]$, and $\gamma$ is chosen as in Lemma~\ref{lem: gamma condition}. Then for any given $K_t \in \G^{\text{lqr}}$, we have that
    \begin{equation}
        \| \E [\widehat{\nabla \C} (K_t) 1_{\mathcal{A}_t}] - \nabla \C (K_t) \|_F \leq \frac{3 \xi_3}{1 - \gamma} \delta \left( \log{\frac{1}{\delta}} \right)^{3/2}.
    \end{equation}
\end{lemma}
\begin{proof}
    Following Proposition~\ref{prop: gradient estimate expectation}, we have that
    \begin{align*}
        \nabla \C (K_t) =& \E [\widehat{\nabla \C} (K_t)] \\
        =& \E [\widehat{\nabla \C} (K_t) 1_{\mathcal{A}_t}] + \E [\widehat{\nabla \C} (K_t) 1_{\mathcal{A}^c_t}].
    \end{align*}
    Therefore, 
    \begin{align}
        &\| \E [\widehat{\nabla \C} (K_t) 1_{\mathcal{A}_t}] - \nabla \C (K_t) \|_F \cr
        = \ & \| \E [\widehat{\nabla \C} (K_t) 1_{\mathcal{A}^c_t}] \|_F \cr
        \overset{\mathrm{(i)}}{\leq} \ & \E \left[ \| \widehat{\nabla \C} (K_t) 1_{\mathcal{A}^c_t} \|_F \right] \cr
        = \ & \E \left[ \| \widehat{\nabla \C} (K_t) \|_F 1_{\mathcal{A}^c_t} \right] \cr
        \overset{\mathrm{(ii)}}{\leq} \ &\E \left[ \| \widehat{\nabla \C} (K_t) \|_F 1_{\left\{ \| \widehat{\nabla \C} (K_t) \|_F \geq \frac{\xi_3 \left( \log{\frac{1}{\delta}} \right)^{3/2}}{1 - \gamma} \right\}} \right] \cr
        = \ &\PP \left\{ \| \widehat{\nabla \C} (K_t) \|_F \geq \frac{\xi_3 \left( \log{\frac{1}{\delta}} \right)^{3/2}}{1 - \gamma} \right\} \E \left[ \| \widehat{\nabla \C} (K_t) \|_F \bigg| \| \widehat{\nabla \C} (K_t) \|_F \geq \frac{\xi_3 \left( \log{\frac{1}{\delta}} \right)^{3/2}}{1 - \gamma} \right], \label{eq: conditional estimate difference prelim}
    \end{align}
    where (i) follows from Jensen's inequality and (ii) from the fact that
    \[
    \mathcal{A}^c_t = \left\{ \| \widehat{\nabla \C} (K_t) \|_F > \frac{\xi_3}{1 - \gamma} \left( \log{\frac{1}{\delta}} \right)^{3/2} \right\} \subseteq \left\{ \| \widehat{\nabla \C} (K_t) \|_F \geq \frac{\xi_3}{1 - \gamma} \left( \log{\frac{1}{\delta}} \right)^{3/2} \right\}.
    \]
    Moreover, it holds that
    \begin{align}
        &\E \left[ \| \widehat{\nabla \C} (K_t) \|_F \bigg| \| \widehat{\nabla \C} (K_t) \|_F \geq \frac{\xi_3 \left( \log{\frac{1}{\delta}} \right)^{3/2}}{1 - \gamma} \right] \cr
        =& \frac{\xi_3 \left( \log{\frac{1}{\delta}} \right)^{3/2}}{1 - \gamma} + \frac{\bigints_{\frac{\xi_3}{1 - \gamma} \left( \log{\frac{1}{\delta}} \right)^{3/2}}^{\infty} \PP \{ \| \widehat{\nabla \C} (K_t) \|_F \geq z \} \ dz}{\PP \left\{ \| \widehat{\nabla \C} (K_t) \|_F \geq \frac{\xi_3 \left( \log{\frac{1}{\delta}} \right)^{3/2}}{1 - \gamma} \right\}}. \label{eq: estimate conditional second part}
    \end{align}
    Now recall from Lemma~\ref{lem: gradient estimate bounds} that 
    \begin{equation}
        \PP \left\{ \| \widehat{\nabla \C} (K_t) \|_F \geq \frac{\xi_3}{1 - \gamma} \left( \log{\frac{1}{\delta}} \right)^{3/2} \right\} \leq \delta \label{eq: estimate size prob bound alternate}
    \end{equation}
    for arbitrary $\delta$, which implies
    \begin{equation}
        \PP \left\{ \| \widehat{\nabla \C} (K_t) \|_F \geq z \right\} \leq e^{- \left( \frac{z (1 - \gamma)}{\xi_3} \right)^{2/3}}. \label{eq: estimate size prob dist}
    \end{equation}
    Now combining~\eqref{eq: estimate size prob dist}, \eqref{eq: estimate conditional second part}, and \eqref{eq: conditional estimate difference prelim} yields
    \begin{align}
        &\| \E [\widehat{\nabla \C} (K_t) 1_{\mathcal{A}_t}] - \nabla \C (K_t) \|_F \cr
        \leq & \PP \left\{ \| \widehat{\nabla \C} (K_t) \|_F \geq \frac{\xi_3 \left( \log{\frac{1}{\delta}} \right)^{3/2}}{1 - \gamma} \right\} \frac{\xi_3 \left( \log{\frac{1}{\delta}} \right)^{3/2}}{1 - \gamma} + \bigintss_{\frac{\xi_3}{1 - \gamma} \left( \log{\frac{1}{\delta}} \right)^{3/2}}^{\infty} e^{- \left( \frac{z (1 - \gamma)}{\xi_3} \right)^{2/3}} dz \cr
        \overset{\mathrm{(i)}}{\leq}& \frac{\xi_3}{1 - \gamma} \delta \left( \log{\frac{1}{\delta}} \right)^{3/2} + \frac{\xi_3}{1 - \gamma} \bigintss_{\left( \log{\frac{1}{\delta}} \right)^{3/2}}^{\infty} e^{- u^{2/3}} du \cr
        =& \frac{\xi_3}{1 - \gamma} \delta \left( \log{\frac{1}{\delta}} \right)^{3/2} + \frac{\xi_3}{1 - \gamma} \left( \frac{3}{2} \delta \left( \log{\frac{1}{\delta}} \right)^{1/2} + \frac{3}{4} \sqrt{\pi} \erfc \left( \sqrt{\log{\frac{1}{\delta}}} \right) \right) \cr
        \overset{\mathrm{(ii)}}{\leq}& \frac{\xi_3}{1 - \gamma} \left( \delta \left( \log{\frac{1}{\delta}} \right)^{3/2} + \frac{3}{2} \delta \left( \log{\frac{1}{\delta}} \right)^{1/2} + \frac{3}{4} \sqrt{\pi} \delta \right) \cr
        \overset{\mathrm{(iii)}}{\leq}& \frac{3 \xi_3}{1 - \gamma} \delta \left( \log{\frac{1}{\delta}} \right)^{3/2},
    \end{align}
    where (i) follows from~\eqref{eq: estimate size prob bound alternate} along with a change of variables $u = \left( \frac{1 - \gamma}{\xi_3} \right) z$ in the integral, (ii) from the fact that $\erfc \left( \sqrt{\log{\frac{1}{\delta}}} \right) \leq \delta$, and (iii) from $\delta \leq e^{- 3/2}$. This concludes the proof.
\end{proof}
Before introducing the next lemma, let us denote the optimality gap of iterate $t$ of the algorithm by 
\begin{equation}\label{eq:Delta}
    \Delta_t := \C (K_t) - \C (K^*).
\end{equation}
Moreover, let $\F_t$ denote the $\sigma$-algebra containing the randomness up to iteration $t$ of Algorithm~\ref{alg: LQR_policy_gradient} (including $K_t$ but not $\widehat{\nabla \C} (K_t)$). We then define
\begin{equation}
    \tau_1 := \min \left\{ t \ |  \ \Delta_t > 10 \C(K_0) \right\}, \label{eq: tau_1 definition}
\end{equation}
which is a stopping time with respect to $\F_t$. 
\begin{lemma}
    \label{lem: policy gradient recursive}
    Suppose $\delta \in (0, e^{-3/2}]$, $\gamma$ is as suggested in Lemma~\ref{lem: gamma condition}, and the update rule follows
    \begin{equation}
        K_{t+1} = K_{t} - \alpha_t \widehat{\nabla \C} (K_t)
    \end{equation} 
    with a step-size $\alpha_t$ such that for all $t \in \{0,1,2,\dotsc\}$,
    \[
    \alpha_t \leq \frac{\omega_{\text{lqr}}}{\frac{\xi_3}{1 - \gamma} \left( \log{\frac{1}{\delta}} \right)^{3/2}}.
    \]
    Then for any $t \in \{0,1,2,\dotsc\}$, we have
    \begin{equation}
        \E[\Delta_{t+1} 1_{\mathcal{A}_t}| \F_t] 1_{\tau_1 > t} \leq \left( \left(1 - \mu_{\text{lqr}} \alpha_t \right) \Delta_t + \frac{3 \xi_3 \widetilde{c_{K_1}}}{1 - \gamma} \delta \left( \log{\frac{1}{\delta}} \right)^{3/2} \alpha_t + \frac{\phi_{\text{lqr}} \alpha_t^2}{2} \frac{\xi_4}{(1 - \gamma)^2} \right) 1_{\tau_1 > t}, \label{eq: policy gradient recursive}
    \end{equation}
    where $\Delta_t$ and $\mathcal{A}_t$ are defined in \eqref{eq:Delta} and \eqref{eq: A_t definition} respectively.
\end{lemma}

\begin{proof}
    First, note that by the definition of $\tau_1$ in~\eqref{eq: tau_1 definition}, $\tau_1 > t$ implies $K_t \in \G^{\text{lqr}}$. In addition, since $\alpha_t \leq \frac{\omega_{\text{lqr}}}{\frac{\xi_3}{1 - \gamma} \left( \log{\frac{1}{\delta}} \right)^{3/2}}$, the event $\mathcal{A}_t$ implies that
    \[
    \| K_{t+1} - K_t \|_F = \| \alpha_t \widehat{\nabla \C} (K_t) \|_F \leq \omega_{\text{lqr}}.
    \]
    Thus, by local smoothness of $\C (K_t)$, see Lemma~\ref{lem:lipschitz_gradient_lqr}, it holds that
    \begin{align*}
        (\Delta_{t+1} - \Delta_t) 1_{\tau_1 > t} 1_{\mathcal{A}_t} =& (\C (K_{t+1} - \C (K_t)) 1_{\tau_1 > t} 1_{\mathcal{A}_t} \\
        \leq & \left( \left\langle \nabla \C (K_t),  K_{t+1} - K_t \right\rangle + \frac{\phi_{\text{lqr}}}{2} \lVert K_{t+1} - K_t \rVert_F^2 \right) 1_{\tau_1 > t} 1_{\mathcal{A}_t} \\
        =& \left( - \alpha_t \left\langle \nabla \C (K_t), \widehat{\nabla \C} (K_t) \right\rangle + \frac{\phi_{\text{lqr}} \alpha_t^2}{2} \lVert \widehat{\nabla \C} (K_t) \rVert_F^2 \right) 1_{\tau_1 > t} 1_{\mathcal{A}_t},
    \end{align*}
    which after taking an expectation conditioned on $\F_t$ gives
    \begin{align*}
        &\E [\Delta_{t+1} 1_{\tau_1 > t} 1_{\mathcal{A}_t} | \F_t] - \E [\Delta_{t} 1_{\tau_1 > t} 1_{\mathcal{A}_t} | \F_t] \\
        \leq & - \alpha_t \left\langle \nabla \C (K_t), \E[\widehat{\nabla \C} (K_t) 1_{\tau_1 > t} 1_{\mathcal{A}_t} | \F_t] \right\rangle + \frac{\phi_{\text{lqr}}}{2} \alpha_t^2 \E[\| \widehat{\nabla \C} (K_t) \|_F^2 1_{\tau_1 > t} 1_{\mathcal{A}_t} | \F_t].
    \end{align*}
    Since $\Delta_t$ and $1_{\tau_1 > t}$ are determined by $\F_t$, 
    \begin{align*}
        &\E [\Delta_{t+1} 1_{\mathcal{A}_t} | \F_t] 1_{\tau_1 > t} \\
        \leq & \left( \Delta_t \E[1_{\mathcal{A}_t} | \F_t] - \alpha_t \left\langle \nabla \C (K_t), \E[\widehat{\nabla \C} (K_t) 1_{\mathcal{A}_t} | \F_t] \right\rangle + \frac{\phi_{\text{lqr}}}{2} \alpha_t^2 \E[\| \widehat{\nabla \C} (K_t) \|_F^2 1_{\mathcal{A}_t} | \F_t] \right) 1_{\tau_1 > t} \\
        \overset{\mathrm{(i)}}{\leq}& \left( \Delta_t - \alpha_t \left\langle \nabla \C (K_t), \E[\widehat{\nabla \C} (K_t) 1_{\mathcal{A}_t} | \F_t] \right\rangle + \frac{\phi_{\text{lqr}}}{2} \alpha_t^2 \E[\| \widehat{\nabla \C} (K_t) \|_F^2 | \F_t] \right) 1_{\tau_1 > t} \\
        =& \Delta_t 1_{\tau_1 > t} - \alpha_t \left\langle \nabla \C (K_t), \nabla \C (K_t) + \E[\widehat{\nabla \C} (K_t) 1_{\mathcal{A}_t} | \F_t] - \nabla \C (K_t) \right\rangle 1_{\tau_1 > t} \\
        &+ \frac{\phi_{\text{lqr}}}{2} \alpha_t^2 \E[\| \widehat{\nabla \C} (K_t) \|_F^2 | \F_t] 1_{\tau_1 > t} \\
        =& \Delta_t 1_{\tau_1 > t} - \alpha_t \left\langle \nabla \C (K_t), \nabla \C (K_t) \right\rangle 1_{\tau_1 > t} \\
        &- \alpha_t \left\langle \nabla \C (K_t), \E[\widehat{\nabla \C} (K_t) 1_{\mathcal{A}_t} | \F_t] - \nabla \C (K_t) \right\rangle 1_{\tau_1 > t} + \frac{\phi_{\text{lqr}}}{2} \alpha_t^2 \E[\| \widehat{\nabla \C} (K_t) \|_F^2 | \F_t] 1_{\tau_1 > t} \\
        \overset{\mathrm{(ii)}}{\leq}& \Delta_t 1_{\tau_1 > t} - \alpha_t \| \nabla \C (K_t) \|_F^2 1_{\tau_1 > t} \\
        &+ \alpha_t \| \nabla \C (K_t) \|_F \| \E[\widehat{\nabla \C} (K_t) 1_{\mathcal{A}_t} | \F_t] - \nabla \C (K_t) \|_F 1_{\tau_1 > t} + \frac{\phi_{\text{lqr}}}{2} \alpha_t^2 \frac{\xi_4}{(1 - \gamma)^2} 1_{\tau_1 > t} \\
        \overset{\mathrm{(iii)}}{\leq}& \Delta_t 1_{\tau_1 > t} - \alpha_t \mu_{\text{lqr}} \Delta_t 1_{\tau_1 > t} + \frac{3 \xi_3 \widetilde{c_{K_1}}}{1 - \gamma} \delta \left( \log{\frac{1}{\delta}} \right)^{3/2} \alpha_t 1_{\tau_1 > t} + \frac{\phi_{\text{lqr}}}{2} \alpha_t^2 \frac{\xi_4}{(1 - \gamma)^2} 1_{\tau_1 > t} \\
        =& \left( \left(1 - \mu_{\text{lqr}} \alpha_t \right) \Delta_t + \frac{3 \xi_3 \widetilde{c_{K_1}}}{1 - \gamma} \delta \left( \log{\frac{1}{\delta}} \right)^{3/2} \alpha_t + \frac{\phi_{\text{lqr}} \alpha_t^2}{2} \frac{\xi_4}{(1 - \gamma)^2} \right) 1_{\tau_1 > t},
    \end{align*}
    where (i) follows from $1_{\mathcal{A}_t} \leq 1$, (ii) from Lemma~\ref{lem: gradient estimate bounds}, and (iii) from applying the PL inequality~\eqref{eq: PL_LQR}, the fact that $\| \nabla \C (K_t) \|_F \leq \widetilde{c_{K_1}}$ for any $K_t \in \G^{\text{lqr}}$ (see Remark~\ref{rem: regularity parameters}), and Lemma~\ref{lem: gradient estimate conditioned}. This finishes the proof of Lemma~\ref{lem: policy gradient recursive}.
\end{proof}
We are now in a position to state a precise version of our main result. 
\begin{theorem}
    \label{thm: lqr policy gradient}
    Suppose $K_0$ is stable and $\gamma$ is as suggested in Lemma~\ref{lem: gamma condition}. If the step-size $\alpha_t$ is chosen as
    \begin{align}
        \alpha_t = \frac{2}{\mu_{\text{lqr}}} \frac{1}{t+N} \quad &\text{for} \quad N = \max\left\{N_1, \frac{2}{\mu_{\text{lqr}}} \frac{\xi_3 \left(\log{\frac{1}{\delta}}\right)^{3/2}}{(1 - \gamma)\omega_{\text{lqr}}}\right\}, \label{eq: theorem parameter setup N}
    \end{align}
    where
    \begin{align}
        N_1 = \max \left\{ 2,\frac{4 \phi_{\text{lqr}} \xi_4}{\mu_{\text{lqr}}^2 (1 - \gamma)^2} \frac{2}{\C(K_0)} \right\}, \label{eq: theorem parameter setup N_1}
    \end{align}
    then for a given error tolerance $\eps$ such that $\C (K_0) \geq \frac{\eps}{20}$, and $\delta$ chosen arbitrarily to satisfy
    \begin{align}
        \delta \leq \min &\Bigg\{ 2 \times 10^{-5}, \left( \frac{\phi_{\text{lqr}} \xi_4 \omega_{\text{lqr}}}{960 \xi_3^2 \widetilde{c_{K_1}} \C (K_0)} \right)^3 \eps^3, \cr
        &\left( \frac{\phi_{\text{lqr}} \xi_4}{480 (1 - \gamma)\mu_{\text{lqr}} \xi_3 \widetilde{c_{K_1}} N_1 \C (K_0)} \right)^3 \eps^3, \left( \frac{\mu_{\text{lqr}} (1 - \gamma)}{240 \xi_3 \widetilde{c_{K_1}}} \right)^3 \eps^3 \Bigg\}, \label{eq: theorem parameter setup delta}
    \end{align}
    the iterate $K_T$ of Algorithm~\ref{alg: LQR_policy_gradient} after 
    \begin{equation}
        T = \frac{40}{\eps} N \C (K_0) \label{eq: theorem parameter setup T}
    \end{equation}
    steps satisfies
    \begin{equation} 
        \C(K_T) - \C(K^*) \leq \eps
    \end{equation}
    with a probability of at least $4/5 - \delta T$.
\end{theorem}

It is essential to re-emphasize that, as also evident from the statement of Theorem~\ref{thm: lqr policy gradient}, there is no reliance on an assumption that the policy remains stable throughout the algorithm; rather, the result is proven to hold with a certain probability. In particular, the instances of the algorithm that lead to instability at any iteration before $T$ are factored into the failure probability $1/5 + \delta T$. 

In Appendix~\ref{app: prob of failure}, we show that the success probability in Theorem~\ref{thm: lqr policy gradient} can be improved from \( 4/5 - \delta T \) to \( 1 - \delta T \) by averaging a batch of gradient estimates to reduce variance and obtain an estimate that is close to the true gradient with high probability.
Moreover, in Appendix~\ref{app: extension_to_noisy_dynamics}, we show how our gradient estimation method and convergence analysis, developed for the random initialization setting, can be naturally extended to the noisy dynamics setting.

The proof of Theorem~\ref{thm: lqr policy gradient} relies on an intermediate result, namely Proposition~\ref{prop: lqr policy gradient}, which we establish next. Before doing so, we provide some observations regarding the statement of the theorem. First, we have the following remark for $\delta$: 
\begin{remark}[Selection of $\delta$ for the probability of failure]{\em
    The $\delta T$ term in the probability of failure stated in Theorem~\ref{thm: lqr policy gradient} can be adjusted arbitrarily; however, since $T$ depends on $N$ which depends on $\delta$ itself, we add some further discussion here. If we want the $\delta T$ term to be less than some arbitrary small $\delta'$, it needs to hold that
    \[
    \delta T = \delta \frac{40}{\eps} \max \left\{ N_1 \C (K_0), \frac{2 \xi_3 \C (K_0)}{\mu_{\text{lqr}} \omega_{\text{lqr}} (1 - \gamma)} \left( \log{\frac{1}{\delta}} \right)^{3/2} \right\} \leq \delta'.
    \]
    Therefore, $\delta$ first needs to satisfy
    \begin{align}
        \frac{40}{\eps} N_1 \C (K_0) \delta \leq \delta' \Rightarrow \delta \leq \frac{\delta' \eps}{40 N_1 \C (K_0)}, \label{eq: delta-delta' first prelim}
    \end{align}
    and secondly,
    \begin{align}
        \frac{80 \xi_3 \C (K_0)}{\mu_{\text{lqr}} \omega_{\text{lqr}} (1 - \gamma)} \frac{1}{\eps} \delta \left( \log{\frac{1}{\delta}} \right)^{3/2} \leq \delta' \Rightarrow \delta \left( \log{\frac{1}{\delta}} \right)^{3/2} \leq \frac{\mu_{\text{lqr}} \omega_{\text{lqr}} (1 - \gamma)}{80 \xi_3 \C (K_0)} \delta' \eps. \label{eq: delta-delta' second prelim}
    \end{align}
    Now since $a^3 \left( \log{\frac{1}{a^3}} \right)^{3/2} \leq a$ for any $a \in (0,1)$, for \eqref{eq: delta-delta' second prelim} to hold, it would suffice to have
    \begin{equation}
        \delta \leq \left( \frac{\mu_{\text{lqr}} \omega_{\text{lqr}} (1 - \gamma)}{80 \xi_3 \C (K_0)} \right)^3 (\delta' \eps)^3. \label{eq: delta-delta' third prelim}
    \end{equation}
    Note that \eqref{eq: delta-delta' third prelim} is only a loose sufficient bound on $\delta$ that can be improved (for instance, the exponents in~\eqref{eq: delta-delta' third prelim} can be reduced from $3$ to $2$ considering the other requirements on $\delta$ in~\eqref{eq: theorem parameter setup delta}); however, since the dependence of $T$ on $\delta$ is logarithmic, the looser requirement only adds a constant and does not change the order.

    As a result, adding \eqref{eq: delta-delta' first prelim} and \eqref{eq: delta-delta' third prelim} to the existing requirements on $\delta$ in~\eqref{eq: theorem parameter setup delta}, we will have
    \begin{align}
        \delta \leq \min &\Bigg\{ 2 \times 10^{-5}, \left( \frac{\phi_{\text{lqr}} \xi_4 \omega_{\text{lqr}}}{960 \xi_3^2 \widetilde{c_{K_1}} \C (K_0)} \right)^3 \eps^3, \left( \frac{\phi_{\text{lqr}} \xi_4}{480 (1 - \gamma)\mu_{\text{lqr}} \xi_3 \widetilde{c_{K_1}} N_1 \C (K_0)} \right)^3 \eps^3, \cr
        &\left( \frac{\mu_{\text{lqr}} (1 - \gamma)}{240 \xi_3 \widetilde{c_{K_1}}} \right)^3 \eps^3, \frac{\delta' \eps}{40 N_1 \C (K_0)}, \left( \frac{\mu_{\text{lqr}} \omega_{\text{lqr}} (1 - \gamma)}{80 \xi_3 \C (K_0)} \right)^3 (\delta' \eps)^3 \Bigg\}, \label{eq: remark delta setup wrt delta'}
    \end{align}
    which will lead to the result of Theorem~\ref{thm: lqr policy gradient} holding with probability $4/5 - \delta'$ after
    \[
    T \sim \frac{N}{\eps} \sim \mathcal{O} \left( \frac{1}{\eps} \left( \log{\frac{1}{(\delta' \eps)^3}} \right)^{3/2} \right) = \mathcal{O} \left( \frac{1}{\eps} \left(\log{\frac{1}{\delta'} + \log{\frac{1}{\eps}}} \right)^{3/2} \right) = \widetilde{\mathcal{O}} \left( \frac{1}{\eps} \right)
    \]
    iterations of Algorithm~\ref{alg: LQR_policy_gradient}. \oprocend}
\end{remark}
Secondly, we find it worthwile to provide the following observation on the choice of $\sigma$:
\begin{remark}[Selection of $\sigma$ and its impact on $T$]{\em
    Note that the value of $\sigma$ in \eqref{eq: gaussian policy} is at our discretion, so one natural question would be regarding the asymptotic analysis of $\sigma$ and its impact on our rate $T$. Observe that the only effect of $\sigma$ on $T$ is through $\xi_3$ and $\xi_4$ defined in \eqref{eq: xi_3 definition} and \eqref{eq: xi_4 definition} respectively. Taking everything else as constants, following the choice of $T$ and $N$ suggested in Theorem~\ref{thm: lqr policy gradient}, we have that $T \geq \mathcal{O} \left(\max\{\xi_3, \xi_4\}\right)$. Now since both $\xi_3$ and $\xi_4$ will grow unbounded as $\sigma$ approaches either zero or infinity, so does $T$. Therefore, we choose a non-zero value for $\sigma$ instead. An optimal value can be derived, but given that this only affects the constants in the rate, we opt for $\sigma = 1$. \oprocend}
\end{remark}

Thirdly, note that for any $K_t \in \G^{\text{lqr}}$, by our choice of $\alpha_t$ and $N$ in Theorem~\ref{thm: lqr policy gradient}, we have
\begin{align}
    \lVert K_{t+1} - K_{t} \rVert_F =& \lVert \alpha_t \widehat{\nabla \C} (K_t) \rVert_F \cr
    =& \frac{2}{\mu_{\text{lqr}}} \frac{1}{t+N} \lVert \widehat{\nabla \C} (K_t) \rVert_F \cr
    \leq & \frac{2}{\mu_{\text{lqr}}} \frac{1}{N} \lVert \widehat{\nabla \C} (K_t) \rVert_F \cr
    \overset{\mathrm{(i)}}{\leq}& \omega_{\text{lqr}} \frac{\lVert \widehat{\nabla \C} (K_t) \rVert_F}{\frac{\xi_3}{1 - \gamma} \left(\log \frac{1}{\delta}\right)^{3/2}}, \label{eq: E_1 reasoning prelim}
\end{align}
where (i) follows from~\eqref{eq: theorem parameter setup N}. Now applying Lemma~\ref{lem: gradient estimate bounds} on \eqref{eq: E_1 reasoning prelim} yields
\begin{equation}
    \lVert K_{t+1} - K_{t} \rVert_F \leq \omega_{\text{lqr}} = \inf_{K \in \G^{\text{lqr}}} \omega_K \label{eq: E_1 reasoning}
\end{equation}
with probability at least $1 - \delta$, where $\omega_K = \min\{\beta_K,\zeta_K\}$. This implies that the local Lipschitzness and local smoothness properties of the cost hold for the update at iteration $t$ with probability at least $1 - \delta$.

Fourthly, to help unravel the logical reasoning elucidated in the proof, we introduce the following stopping times:
\begin{align}
    \tau_2 &:= \min \left\{ t \geq 1 \ \Big|  \ \lVert \widehat{\nabla \C} (K_{t-1}) \rVert_F > \frac{\xi_3}{1 - \gamma} \left(\log \frac{1}{\delta}\right)^{3/2} \right\} \cr
    \tau &:= \min \{ \tau_1, \tau_2 \}, \label{eq:stopping-times}
\end{align}
with $\tau_1$ previously defined in~\eqref{eq: tau_1 definition}. Essentially, one can  observe that as long as $t < \tau_1$ and $t+1 < \tau_2$, it holds that $K_t \in \G^{\text{lqr}}$ and $\lVert K_{t+1} - K_{t} \rVert_F \leq \omega_{\text{lqr}}$, implying that local Lipschitzness and local smoothness properties of the cost hold until that iteration. By the definition of $\tau$ in~\eqref{eq:stopping-times}, we have that
\begin{equation}
    1_{\tau > t} = 1_{\tau_1 > t} 1_{\tau_2 > t}. \label{eq: tau indicator wrt tau_1-tau_2 indicators}
\end{equation}
Moreover, following the definition of $\mathcal{A}_t$ in~\eqref{eq: A_t definition}, it also holds that
\begin{equation}
    1_{\tau_2 > t+1} = 1_{\tau_2 > t} 1_{\mathcal{A}_t}. \label{eq: tau_2 - A_t}
\end{equation}

Finally, we note that the idea of introducing a stopping time~\eqref{eq:stopping-times}, which helps identify the failure of the algorithm and is also used to define a stopped process later on, is inspired by \cite{DM-AP-KB-KK-PLB-MJW:18}. However, despite the similarity of our forthcoming statements to those in the proof of \cite[Theorem 8]{DM-AP-KB-KK-PLB-MJW:18}, the paths we take to prove said statements are considerably different due to the differences in how we defined our stopping time (and subsequently the stopped process to be defined later on), our gradient estimation method, the time-varying learning rate, etc.

Having covered all of the above, we are now ready to present the following proposition:
\begin{proposition}
    \label{prop: lqr policy gradient}
    Under the parameter settings of Theorem~\ref{thm: lqr policy gradient}, we have
    \begin{equation*}
        \E[\Delta_T 1_{\tau > T}] \leq \frac{\eps}{20}.
    \end{equation*}
    Furthermore, the event $\left\{ \tau > T \right\}$ happens with a probability of at least $\frac{17}{20} - \delta T$.
\end{proposition}

\begin{proof}
The following provides us with a stepping stone for proving the first claim: 
\begin{sublemma}
    \label{sublem: lqr_pg_inductive}
    Under the parameter settings of Theorem~\ref{thm: lqr policy gradient}, we have that
    \begin{equation}
        \E[\Delta_t 1_{\tau > t}] \leq \frac{\eps}{40} + \frac{N \C (K_0)}{t+N},
    \end{equation}
    for all $t \in [T]$.
\end{sublemma}
\emph{Proof of Sublemma~\ref{sublem: lqr_pg_inductive}.} We prove this result by induction on $t$ as follows:

\textbf{Base case ($t = 0$):} 
\begin{align*}
    \Delta_0 1_{\tau > 0} \leq \Delta_0 \leq \C (K_0) = \frac{N \C (K_0)}{0 + N} \leq \frac{\eps}{40} + \frac{N \C (K_0)}{0+N},
\end{align*}
which after taking expectation proves the claim for $t = 0$.

\textbf{Inductive step:} Let $k \in [T-1]$ be fixed and assume that
\begin{equation}
    \E[\Delta_k 1_{\tau > k}] \leq \frac{\eps}{40} + \frac{N \C (K_0)}{k+N} \label{eq: sublem_pg_hyp}
\end{equation}
holds (the inductive hypothesis). Observe that
\begin{align}
    \E[\Delta_{k+1} 1_{\tau > k+1}] &\overset{\mathrm{(i)}}{=} \E [\Delta_{k+1} 1_{\tau_1 > k+1} 1_{\tau_2 > k+1}] \cr
    &\overset{\mathrm{(ii)}}{\leq} \E [\Delta_{k+1} 1_{\tau_1 > k} 1_{\tau_2 > k} 1_{\mathcal{A}_k}] \cr
    &= \E [\E[\Delta_{k+1} 1_{\tau_1 > k} 1_{\tau_2 > k} 1_{\mathcal{A}_k} | \F_k]] \cr
    &\overset{\mathrm{(iii)}}{=} \E [\E[\Delta_{k+1} 1_{\mathcal{A}_k} | \F_k] 1_{\tau_1 > k} 1_{\tau_2 > k}], \label{eq: sublem_pg_conditional_expectation_bound}
\end{align}
where (i) follows from~\eqref{eq: tau indicator wrt tau_1-tau_2 indicators}, (ii) from equation~\eqref{eq: tau_2 - A_t} along with the fact that $1_{\tau_1 > k+1} \leq 1_{\tau_1 > k}$, and (iii) is due to $1_{\tau_2 > k}$ and $1_{\tau_1 > k}$ being determined by $\F_k$. By Lemma~\ref{lem: policy gradient recursive}, we have that
\begin{align}
    &\left( \E[\Delta_{k+1} 1_{\mathcal{A}_k} | \F_k] 1_{\tau_1 > k} \right) 1_{\tau_2 > k} \cr
    \leq & \left( \left( \left(1 - \mu_{\text{lqr}} \alpha_k \right) \Delta_k + \frac{3 \xi_3 \widetilde{c_{K_1}}}{1 - \gamma} \delta \left( \log{\frac{1}{\delta}} \right)^{3/2} \alpha_k + \frac{\phi_{\text{lqr}} \alpha_k^2}{2} \frac{\xi_4}{(1 - \gamma)^2} \right) 1_{\tau_1 > k} \right) 1_{\tau_2 > k} \cr
    \overset{\mathrm{(i)}}{=}& \left( \left(1 - \frac{2}{k+N} \right) \Delta_k + \frac{6 \xi_3 \widetilde{c_{K_1}} \delta \left( \log{\frac{1}{\delta}} \right)^{3/2}}{\mu_{\text{lqr}} (1 - \gamma)} \frac{1}{k + N} + \frac{2 \phi_{\text{lqr}} \xi_4}{(1 - \gamma)^2 \mu^2_{\text{lqr}}} \frac{1}{(k+N)^2} \right) 1_{\tau > k}, \label{eq: recursive prelim 1}
\end{align}
where (i) follows from\eqref{eq: tau indicator wrt tau_1-tau_2 indicators} along with replacing $\alpha_k$ with its value in~\eqref{eq: theorem parameter setup N}. Now due to the choice of $\delta$ in~\eqref{eq: theorem parameter setup delta}, we have that
\[
\delta \leq \left( \frac{\mu_{\text{lqr}} (1 - \gamma)}{240 \xi_3 \widetilde{c_{K_1}}} \right)^3 \eps^3,
\]
which after noting that $a^3 \left( \log{\frac{1}{a^3}} \right)^{3/2} \leq a$ for any $a \in (0,1)$ implies 
\begin{equation}
\delta \left( \log{\frac{1}{\delta}} \right)^{3/2} \leq \frac{\mu_{\text{lqr}} (1 - \gamma)}{240 \xi_3 \widetilde{c_{K_1}}} \eps \Rightarrow \frac{6 \xi_3 \widetilde{c_{K_1}} \delta \left( \log{\frac{1}{\delta}} \right)^{3/2}}{\mu_{\text{lqr}} (1 - \gamma)} \leq \frac{\eps}{40}. \label{eq: recursive prelim 2}
\end{equation}
Applying~\eqref{eq: recursive prelim 2} on \eqref{eq: recursive prelim 1} yields
\begin{align*}
    &\E[\Delta_{k+1} 1_{\mathcal{A}_k} | \F_k] 1_{\tau_1 > k} 1_{\tau_2 > k} \\
    \leq & \left( \left(1 - \frac{2}{k+N} \right) \Delta_k + \frac{\eps}{40} \frac{1}{k + N} + \frac{2 \phi_{\text{lqr}} \xi_4}{(1 - \gamma)^2 \mu^2_{\text{lqr}}} \frac{1}{(k+N)^2} \right) 1_{\tau > k} \\
    \leq & \left(1 - \frac{2}{k+N} \right) \Delta_k 1_{\tau > k} + \frac{\eps}{40} \frac{1}{k + N} + \frac{2 \phi_{\text{lqr}} \xi_4}{(1 - \gamma)^2 \mu^2_{\text{lqr}}} \frac{1}{(k+N)^2},
\end{align*}
which after taking expectation results in
\begin{align}
    &\E [\E[\Delta_{k+1} 1_{\mathcal{A}_k} | \F_k] 1_{\tau_1 > k} 1_{\tau_2 > k}] \cr
    \leq & \left(1 - \frac{2}{k+N} \right) \E[\Delta_k 1_{\tau > k}] + \frac{\eps}{40} \frac{1}{k + N} + \frac{2 \phi_{\text{lqr}} \xi_4}{(1 - \gamma)^2 \mu^2_{\text{lqr}}} \frac{1}{(k+N)^2}. \label{eq: recursive prelim 3}
\end{align}
Combining the hypothesis (inequality~\eqref{eq: sublem_pg_hyp}) and inequality~\eqref{eq: sublem_pg_conditional_expectation_bound} with~\eqref{eq: recursive prelim 3}, we obtain
\begin{align*}
    &\E[\Delta_{k+1} 1_{\tau > k+1}] \\
    \leq & \left( 1 - \frac{2}{k+N} \right) \left( \frac{\eps}{40} + \frac{N \C (K_0)}{k+N} \right) + \frac{\eps}{40} \frac{1}{k + N} + \frac{2 \phi_{\text{lqr}} \xi_4}{(1 - \gamma)^2 \mu^2_{\text{lqr}}} \frac{1}{(k+N)^2} \\
    \leq & \frac{\eps}{40} + \left( 1 - \frac{1}{k + N} \right) \frac{N \C (K_0)}{k + N} - \frac{1}{(k + N)^2} \left( N \C (K_0) - \frac{2 \phi_{\text{lqr}} \xi_4}{(1 - \gamma)^2 \mu^2_{\text{lqr}}} \right) \\
    \overset{\mathrm{(i)}}{\leq} &\frac{\eps}{40} + \left( \frac{k + N - 1}{(k + N)^2} \right) N \C (K_0) \\
    \leq & \frac{\eps}{40} + \frac{N \C (K_0)}{k+N+1},
\end{align*}
where (i) follows from the fact that
\begin{equation*}
    N \C (K_0) \geq N_1 \C (K_0) \geq \left( \frac{4 \phi_{\text{lqr}} \xi_4}{\mu_{\text{lqr}}^2 (1 - \gamma)^2} \frac{2}{\C(K_0)} \right) \C (K_0) = \frac{8 \phi_{\text{lqr}} \xi_4}{(1 - \gamma)^2 \mu^2_{\text{lqr}}} \geq \frac{2 \phi_{\text{lqr}} \xi_4}{(1 - \gamma)^2 \mu^2_{\text{lqr}}}.
\end{equation*}

This proves the claim for $k+1$, completing the inductive step. \oprocend

Now utilizing Sublemma~\ref{sublem: lqr_pg_inductive} and the choice of $T$ from \eqref{eq: theorem parameter setup T} in Theorem~\ref{thm: lqr policy gradient}, 
\begin{align*}
    \E[\Delta_T 1_{\tau > T}] \leq \frac{\eps}{40} + \frac{N \C (K_0)}{T+N} \leq \frac{\eps}{40} + \frac{N \C (K_0)}{T} = \frac{\eps}{20},
\end{align*}
which finishes the proof of the first claim of Proposition~\ref{prop: lqr policy gradient}. Now before moving on to the second claim, we introduce the following sublemma:
\begin{sublemma} \label{sublem: prelim to martingale}
    Under the parameter setup of Theorem~\ref{thm: lqr policy gradient}, we have that for all $t \in [T]$,
    \begin{equation}
        \frac{3 \xi_3 \widetilde{c_{K_1}}}{1 - \gamma} \delta \left( \log{\frac{1}{\delta}} \right)^{3/2} \alpha_t + \frac{\phi_{\text{lqr}} \xi_4}{2 (1 - \gamma)^2} \alpha_t^2 + \frac{4 \phi_{\text{lqr}} \xi_4}{(1 - \gamma)^2 \mu^2_{\text{lqr}}} \frac{1}{t + N + 1} \leq \frac{4 \phi_{\text{lqr}} \xi_4}{(1 - \gamma)^2 \mu^2_{\text{lqr}}} \frac{1}{t + N}. \label{eq: sublemma prelim to martingale}
    \end{equation}
\end{sublemma}
\emph{Proof of Sublemma~\ref{sublem: prelim to martingale}.} First, substituting $\alpha_t$ with its value in~\eqref{eq: theorem parameter setup N}, inequality~\eqref{eq: sublemma prelim to martingale} becomes
\begin{align}
    &\frac{6 \xi_3 \widetilde{c_{K_1}} \delta \left( \log{\frac{1}{\delta}} \right)^{3/2}}{(1 - \gamma) \mu_{\text{lqr}}} \frac{1}{t + N} + \frac{2 \phi_{\text{lqr}} \xi_4}{(1 - \gamma)^2 \mu^2_{\text{lqr}}} \left( \frac{1}{(t+N)^2} + \frac{2}{t + N + 1} \right) \leq \frac{2 \phi_{\text{lqr}} \xi_4}{(1 - \gamma)^2 \mu^2_{\text{lqr}}} \left( \frac{2}{t + N} \right) \cr
    &\iff \frac{6 \xi_3 \widetilde{c_{K_1}} \delta \left( \log{\frac{1}{\delta}} \right)^{3/2}}{(1 - \gamma) \mu_{\text{lqr}}} \frac{1}{t + N} \leq \frac{2 \phi_{\text{lqr}} \xi_4}{(1 - \gamma)^2 \mu^2_{\text{lqr}}} \left( \frac{2}{t+N} - \frac{2}{t+N+1} - \frac{1}{(t+N)^2} \right) \cr
    &\iff \frac{6 \xi_3 \widetilde{c_{K_1}} \delta \left( \log{\frac{1}{\delta}} \right)^{3/2}}{(1 - \gamma) \mu_{\text{lqr}}} \frac{1}{t + N} \leq \frac{2 \phi_{\text{lqr}} \xi_4}{(1 - \gamma)^2 \mu^2_{\text{lqr}}} \left( \frac{2}{(t+N)(t+N+1)} - \frac{1}{(t+N)^2} \right) \cr
    &\iff \frac{6 \xi_3 \widetilde{c_{K_1}} \delta \left( \log{\frac{1}{\delta}} \right)^{3/2}}{(1 - \gamma) \mu_{\text{lqr}}} \frac{1}{t + N} \leq \frac{2 \phi_{\text{lqr}} \xi_4}{(1 - \gamma)^2 \mu^2_{\text{lqr}}} \left( \frac{t + N - 1}{(t+N)^2 (t+N+1)} \right) \cr
    &\iff \delta \left( \log{\frac{1}{\delta}} \right)^{3/2} \leq \frac{\phi_{\text{lqr}} \xi_4}{3 \xi_3 \widetilde{c_{K_1}} (1 - \gamma) \mu_{\text{lqr}}} \left( \frac{t + N - 1}{(t+N) (t+N+1)} \right). \label{eq: sublemma prelim to martingale pre1}
\end{align}
Note that for the right-hand side of~\eqref{eq: sublemma prelim to martingale pre1}, we have for all $t \in [T]$ that
\begin{align}
    \frac{\phi_{\text{lqr}} \xi_4}{3 \xi_3 \widetilde{c_{K_1}} (1 - \gamma) \mu_{\text{lqr}}} \left( \frac{t + N - 1}{t+N} \frac{1}{t+N+1} \right) &\overset{\mathrm{(i)}}{\geq} \frac{\phi_{\text{lqr}} \xi_4}{6 \xi_3 \widetilde{c_{K_1}} (1 - \gamma) \mu_{\text{lqr}}} \left( \frac{1}{t+N+1} \right) \cr
    &\geq \frac{\phi_{\text{lqr}} \xi_4}{6 \xi_3 \widetilde{c_{K_1}} (1 - \gamma) \mu_{\text{lqr}}} \left( \frac{1}{T+N+1} \right) \cr
    &\overset{\mathrm{(ii)}}{\geq} \frac{\phi_{\text{lqr}} \xi_4}{12 \xi_3 \widetilde{c_{K_1}} (1 - \gamma) \mu_{\text{lqr}}} \left( \frac{1}{T} \right), \label{eq: sublemma prelim to martingale pre2}
\end{align}
where (i) follows from the fact that $\frac{t+N-1}{t+N} \geq \frac{1}{2}$ which is due to $N \geq 2$ (see \eqref{eq: theorem parameter setup N} and \eqref{eq: theorem parameter setup N_1}), and (ii) from $\C (K_0) \geq \frac{\eps}{20}$ under the settings of Theorem~\ref{thm: lqr policy gradient}, which results in
\[
T = \frac{40}{\eps} N \C (K_0) \geq 2 N \geq N+1 \Rightarrow \frac{1}{T+N+1} \geq \frac{1}{2T}.
\]
As a result of \eqref{eq: sublemma prelim to martingale pre1} and \eqref{eq: sublemma prelim to martingale pre2}, in order to conclude the proof Sublemma~\ref{sublem: prelim to martingale}, it would suffice to show that
\begin{align}
    \delta \left( \log{\frac{1}{\delta}} \right)^{3/2} \leq & \frac{\phi_{\text{lqr}} \xi_4}{12 \xi_3 \widetilde{c_{K_1}} (1 - \gamma) \mu_{\text{lqr}}} \left( \frac{1}{T} \right) \cr
    =& \frac{\phi_{\text{lqr}} \xi_4}{12 \xi_3 \widetilde{c_{K_1}} (1 - \gamma) \mu_{\text{lqr}}} \frac{\eps}{40} \frac{1}{N \C (K_0)} \cr
    =& \frac{\phi_{\text{lqr}} \xi_4}{12 \xi_3 \widetilde{c_{K_1}} (1 - \gamma) \mu_{\text{lqr}}} \frac{\eps}{40} \frac{1}{\max\left\{N_1 \C (K_0), \frac{2 \C (K_0)}{\mu_{\text{lqr}}} \frac{\xi_3 \left(\log{\frac{1}{\delta}}\right)^{3/2}}{(1 - \gamma)\omega_{\text{lqr}}}\right\}} \cr
    =& \frac{\phi_{\text{lqr}} \xi_4}{12 \xi_3 \widetilde{c_{K_1}} (1 - \gamma) \mu_{\text{lqr}}} \frac{\eps}{40} \min \left\{ \frac{1}{N_1 \C (K_0)}, \frac{\mu_{\text{lqr}} \omega_{\text{lqr}} (1 - \gamma)}{2 \C (K_0) \xi_3 \left( \log{\frac{1}{\delta}} \right)^{3/2}} \right\}. \label{eq: sublemma prelim to martingale pre3}
\end{align}
For~\eqref{eq: sublemma prelim to martingale pre3} to hold, we need two inequalities to hold as a result of the $\min \{.,.\}$ operator. First, we require
\begin{equation}
    \delta \left( \log{\frac{1}{\delta}} \right)^{3/2} \leq \frac{\phi_{\text{lqr}} \xi_4}{480 \xi_3 \widetilde{c_{K_1}} (1 - \gamma) \mu_{\text{lqr}} N_1 \C (K_0)} \eps. \label{eq: sublemma prelim to martingale pre4}
\end{equation}
Now since $a^3 \left( \log{\frac{1}{a^3}} \right)^{3/2} \leq a$ for all $a \in (0,1)$ and the choice of $\delta$ in~\eqref{eq: theorem parameter setup delta}, i.e., 
\[
\delta \leq \left( \frac{\phi_{\text{lqr}} \xi_4}{480 \xi_3 \widetilde{c_{K_1}} (1 - \gamma) \mu_{\text{lqr}} N_1 \C (K_0)} \right)^3 \eps^3,
\]
we conclude that~\eqref{eq: sublemma prelim to martingale pre4} holds for the parameter setup of Theorem~\ref{thm: lqr policy gradient}.

Secondly, it needs to hold that
\begin{align}
    &\delta \left( \log{\frac{1}{\delta}} \right)^{3/2} \leq \frac{\phi_{\text{lqr}} \xi_4 \omega_{\text{lqr}}}{960 \xi_3^2 \widetilde{c_{K_1}} \C (K_0) \left( \log{\frac{1}{\delta}} \right)^{3/2}} \eps \cr
    &\iff \delta \left( \log{\frac{1}{\delta}} \right)^{3} \leq \frac{\phi_{\text{lqr}} \xi_4 \omega_{\text{lqr}}}{960 \xi_3^2 \widetilde{c_{K_1}} \C (K_0)} \eps. \label{eq: sublemma prelim to martingale pre5}
\end{align}
Now if
\[
\frac{\phi_{\text{lqr}} \xi_4 \omega_{\text{lqr}}}{960 \xi_3^2 \widetilde{c_{K_1}} \C (K_0)} \eps \leq 0.028,
\]
for any $\delta \leq \left( \frac{\phi_{\text{lqr}} \xi_4 \omega_{\text{lqr}}}{960 \xi_3^2 \widetilde{c_{K_1}} \C (K_0)}\right)^3 \eps^3$, we have that
\[
 \delta \left( \log{\frac{1}{\delta}} \right)^{3} \leq \frac{\phi_{\text{lqr}} \xi_4 \omega_{\text{lqr}}}{960 \xi_3^2 \widetilde{c_{K_1}} \C (K_0)} \eps,
\]
and if 
\[
\frac{\phi_{\text{lqr}} \xi_4 \omega_{\text{lqr}}}{960 \xi_3^2 \widetilde{c_{K_1}} \C (K_0)} \eps > 0.028,
\]
it would suffice to have that
\[
\delta \left( \log{\frac{1}{\delta}} \right)^{3} \leq 0.028,
\]
which would hold for any $\delta \leq 2 \times 10^{-5}$. As a result, due to the choice of $\delta$ in~\eqref{eq: theorem parameter setup delta}, i.e., 
\[
\delta \leq \min \left\{2 \times 10^{-5}, \left( \frac{\phi_{\text{lqr}} \xi_4 \omega_{\text{lqr}}}{960 \xi_3^2 \widetilde{c_{K_1}} \C (K_0)}\right)^3 \eps^3 \right\},
\]
we have that~\eqref{eq: sublemma prelim to martingale pre5} will also hold under the parameter setup of Theorem~\ref{thm: lqr policy gradient}. Finally, since both~\eqref{eq: sublemma prelim to martingale pre4} and \eqref{eq: sublemma prelim to martingale pre5} hold for $\delta$ as chosen in~\eqref{eq: theorem parameter setup delta}, inequality~\eqref{eq: sublemma prelim to martingale pre3} is satisfied, finishing the proof. \oprocend

We now prove the second claim. Even though our proof strategy mimics the one in~\cite{DM-AP-KB-KK-PLB-MJW:18}, the structure of the stopping times in~\eqref{eq: tau_1 definition} and \eqref{eq:stopping-times} makes the arguments more involved. Note that this difference in the definition of the stopping time (and subsequently the stopped process) can be attributed to the fact that in contrast to \cite{DM-AP-KB-KK-PLB-MJW:18}'s one scenario (leaving the stable region) which may lead their algorithm to fail, there are two possible scenarios that may cause the failure of our algorithm. 
We start by introducing the stopped process
\begin{equation}
    Y_t : = \Delta_{\tau_1 \wedge t} 1_{\tau_2 > t} + \frac{4 \phi_{\text{lqr}} \xi_4}{(1 - \gamma)^2 \mu^2_{\text{lqr}}} \frac{1}{t+N} \quad \text{for each } t \in [T].
\end{equation}
We next show that this process is a supermartingale. First, we have that
\begin{align}
    &\E[Y_{t+1} | \F_t] \cr
    =& \E[\Delta_{\tau_1 \wedge t+1} 1_{\tau_2 > t+1} | \F_t] + \frac{4 \phi_{\text{lqr}} \xi_4}{(1 - \gamma)^2 \mu^2_{\text{lqr}}} \frac{1}{t+N+1} \cr
    =& \E[\Delta_{\tau_1 \wedge t+1} 1_{\tau_2 > t+1} \left( 1_{\tau_1 \leq t} + 1_{\tau_1 > t} \right)| \F_t] + \frac{4 \phi_{\text{lqr}} \xi_4}{(1 - \gamma)^2 \mu^2_{\text{lqr}}} \frac{1}{t+N+1} \cr
    =& \E[\Delta_{\tau_1 \wedge t+1} 1_{\tau_2 > t+1} 1_{\tau_1 \leq t} | \F_t] + \E[\Delta_{\tau_1 \wedge t+1} 1_{\tau_2 > t+1} 1_{\tau_1 > t} | \F_t] + \frac{4 \phi_{\text{lqr}} \xi_4}{(1 - \gamma)^2 \mu^2_{\text{lqr}}} \frac{1}{t+N+1}. \label{eq: martingale prelim 1}
\end{align}
Then for the first term on the right-hand side of~\eqref{eq: martingale prelim 1}, it holds that
\begin{align}
    \E[\Delta_{\tau_1 \wedge t+1} 1_{\tau_2 > t+1} 1_{\tau_1 \leq t} | \F_t] \leq & \E[\Delta_{\tau_1 \wedge t+1} 1_{\tau_2 > t} 1_{\tau_1 \leq t} | \F_t] \cr
    =& 1_{\tau_2 > t} \E[\Delta_{\tau_1 \wedge t+1} 1_{\tau_1 \leq t} | \F_t] \cr
    =& 1_{\tau_2 > t} \E[\Delta_{\tau_1 \wedge t} 1_{\tau_1 \leq t} | \F_t] \cr
    =& \Delta_{\tau_1 \wedge t} 1_{\tau_2 > t} 1_{\tau_1 \leq t}. \label{eq: pg_martingale_first_term}
\end{align}
As for the second term, we have
\begin{align}
    &\E[\Delta_{\tau_1 \wedge t+1} 1_{\tau_2 > t+1} 1_{\tau_1 > t} | \F_t] \cr
    \overset{\mathrm{(i)}}{=} &\E[\Delta_{\tau_1 \wedge t+1} 1_{\tau_1 > t} 1_{\tau_2 > t} 1_{\mathcal{A}_t} | \F_t] \cr
    =& \E[\Delta_{t+1} 1_{\tau_1 > t} 1_{\tau_2 > t} 1_{\mathcal{A}_t} | \F_t] \cr
    =& \E[\Delta_{t+1} 1_{\mathcal{A}_t} | \F_t] 1_{\tau_1 > t} 1_{\tau_2 > t} \cr
    \overset{\mathrm{(ii)}}{\leq} &\left( \left(1 - \mu_{\text{lqr}} \alpha_t \right) \Delta_t + \frac{3 \xi_3 \widetilde{c_{K_1}}}{1 - \gamma} \delta \left( \log{\frac{1}{\delta}} \right)^{3/2} \alpha_t + \frac{\phi_{\text{lqr}} \alpha_t^2}{2} \frac{\xi_4}{(1 - \gamma)^2} \right) 1_{\tau_1 > t} 1_{\tau_2 > t} \cr
    =& \left( \left(1 - \frac{2}{t+N} \right) \Delta_t + \frac{3 \xi_3 \widetilde{c_{K_1}}}{1 - \gamma} \delta \left( \log{\frac{1}{\delta}} \right)^{3/2} \alpha_t + \frac{\phi_{\text{lqr}} \alpha_t^2}{2} \frac{\xi_4}{(1 - \gamma)^2} \right) 1_{\tau_1 > t} 1_{\tau_2 > t} \cr
    \overset{\mathrm{(iii)}}{\leq} &\Delta_t 1_{\tau_1 > t} 1_{\tau_2 > t} + \frac{3 \xi_3 \widetilde{c_{K_1}}}{1 - \gamma} \delta \left( \log{\frac{1}{\delta}} \right)^{3/2} \alpha_t + \frac{\phi_{\text{lqr}} \alpha_t^2}{2} \frac{\xi_4}{(1 - \gamma)^2} \cr
    \overset{\mathrm{(iv)}}{=} &\Delta_{\tau_1 \wedge t} 1_{\tau_1 > t} 1_{\tau_2 > t} + \frac{3 \xi_3 \widetilde{c_{K_1}}}{1 - \gamma} \delta \left( \log{\frac{1}{\delta}} \right)^{3/2} \alpha_t + \frac{\phi_{\text{lqr}} \alpha_t^2}{2} \frac{\xi_4}{(1 - \gamma)^2}, \label{eq: pg_martingale_second_term}
\end{align}
where (i) follows from~\eqref{eq: tau_2 - A_t}, (ii) from Lemma~\ref{lem: policy gradient recursive}, (iii) from $1_{\tau_1 > t} 1_{\tau_2 > t} \leq 1$ along with the fact that $\frac{2}{t + N} \leq 1$ for all $t \in [T]$, and (iv) from $\Delta_t 1_{\tau_1 > t} = \Delta_{\tau_1 \wedge t} 1_{\tau_1 > t}$.

Combining \eqref{eq: martingale prelim 1}, \eqref{eq: pg_martingale_first_term}, and \eqref{eq: pg_martingale_second_term}, we obtain that for all $t \in [T]$,
\begin{align*}
    \E[Y_{t+1} | \F_t] \leq & \Delta_{\tau_1 \wedge t} 1_{\tau_2 > t} 1_{\tau_1 \leq t} + \Delta_{\tau_1 \wedge t} 1_{\tau_1 > t} 1_{\tau_2 > t} \\
    &+ \frac{3 \xi_3 \widetilde{c_{K_1}}}{1 - \gamma} \delta \left( \log{\frac{1}{\delta}} \right)^{3/2} \alpha_t + \frac{\phi_{\text{lqr}} \alpha_t^2}{2} \frac{\xi_4}{(1 - \gamma)^2} + \frac{4 \phi_{\text{lqr}} \xi_4}{(1 - \gamma)^2 \mu^2_{\text{lqr}}} \frac{1}{t+N+1} \\
    \overset{\mathrm{(i)}}{\leq} &\Delta_{\tau_1 \wedge t} 1_{\tau_2 > t} (1_{\tau_1 \leq t} + 1_{\tau_1 > t}) + \frac{4 \phi_{\text{lqr}} \xi_4}{(1 - \gamma)^2 \mu^2_{\text{lqr}}} \frac{1}{t + N} \\
    =& \Delta_{\tau_1 \wedge t} 1_{\tau_2 > t} + \frac{4 \phi_{\text{lqr}} \xi_4}{(1 - \gamma)^2 \mu^2_{\text{lqr}}} \frac{1}{t + N} \\
    =& Y_t,
\end{align*}
where (i) follows from Sublemma~\ref{sublem: prelim to martingale}. This proves the claim that $Y_t$ is a supermartingale. Moreover, define the following events:
\begin{align}
    \mathcal{E}_1 &:= \{ \tau_2 \leq \tau_1 \; \text{and} \; \tau_2 \in [T] \} \\
    \mathcal{E}_2 &:= \{ \tau_1 < \tau_2 \; \text{and} \; \tau_1 \in [T] \} \\
    \mathcal{E}_3 &:= \left\{ \max_{t \in [T]} \Delta_{\tau_1 \wedge t} 1_{\tau_2 > t} \geq 10 \C (K_0) \right\},
\end{align}
and hence, we have $\PP \{ \tau \leq T \} = \PP (\mathcal{E}_1) + \PP (\mathcal{E}_2)$. Now since $\tau_2 \leq \tau_1$ in $\mathcal{E}_1$ suggests that $\lVert \widehat{\nabla \C} (K_{\tau_2 - 1}) \rVert_F > \frac{\xi_3}{1 - \gamma} \left(\log \frac{1}{\delta}\right)^{3/2}$ despite $\Delta_{\tau_2 - 1} \leq 10 \C(K_0)$ (which implies $\K_{\tau_2 - 1} \in \G^{\text{lqr}}$), after applying union bound on the result of Lemma~\ref{lem: gradient estimate bounds}, we have
\begin{equation}
    \PP (\mathcal{E}_1) \leq \delta T. \label{eq: probability bound on E_1}
\end{equation}
Furthermore, note that $\tau_1 < \tau_2$ in $\mathcal{E}_2$ implies that $\Delta_{\tau_1 \wedge \tau_1} 1_{\tau_2 > \tau_1} = \Delta_{\tau_1}$ and since $\tau_1 \in [T]$, it holds that
\begin{equation*}
    \max_{t \in [T]} \Delta_{\tau_1 \wedge t} 1_{\tau_2 > t} \geq  \Delta_{\tau_1 \wedge \tau_1} 1_{\tau_2 > \tau_1} = \Delta_{\tau_1} \overset{\mathrm{(i)}}{>} 10 \C(K_0),
\end{equation*}
where (i) follows the definition of $\tau_1$. As a result of this, we have that $\mathcal{E}_2$ implies $\mathcal{E}_3$, and consequently, $\PP (\mathcal{E}_2) \leq \PP (\mathcal{E}_3)$. Finally, since $Y_t \geq \Delta_{\tau_1 \wedge t} 1_{\tau_2 > t}$ for all $t \in [T]$, we have that
\begin{align}
    \PP (\mathcal{E}_2) \leq & \PP(\mathcal{E}_3) \cr
    =& \PP \left\{ \max_{t \in [T]} \Delta_{\tau_1 \wedge t} 1_{\tau_2 > t} \geq 10 \C(K_0) \right\} \cr
    \leq & \PP \left\{ \max_{t \in [T]} Y_t \geq 10 \C(K_0) \right\} \cr
    \overset{\mathrm{(i)}}{\leq} &\frac{\E[Y_0]}{10 \C(K_0)} \cr
    =& \frac{\Delta_{\tau_1 \wedge 0} 1_{\tau_2 > 0} + \frac{4 \phi_{\text{lqr}} \xi_4}{(1 - \gamma)^2 \mu^2_{\text{lqr}}} \frac{1}{N}}{10 \C(K_0)} \cr
    \overset{\mathrm{(ii)}}{\leq} &\frac{\Delta_0 + \C(K_0)/2}{10 \C(K_0)} \cr
    \leq & \frac{\C(K_0) + \C(K_0)/2}{10 \C(K_0)} \cr
    =& \frac{3}{20}, \label{eq: probability bound on E_2}
\end{align}
where (i) follows from applying Doob/Ville's inequality for supermartingales, and (ii) from the condition on the choice of $N$ in Theorem~\ref{thm: lqr policy gradient}. Utilizing the acquired probability bounds \eqref{eq: probability bound on E_1} and \eqref{eq: probability bound on E_2}, we observe that
\begin{align*}
    \PP \{ \tau \leq T \} =& \PP (\mathcal{E}_1) + \PP (\mathcal{E}_2) \\
    \leq & \delta T + \frac{3}{20},
\end{align*}
which verifies the second claim of Proposition~\ref{prop: lqr policy gradient}, concluding the proof.
\end{proof}

The proof of our main result is a straightforward corollary: 
\begin{proof}[Proof of Theorem~\ref{thm: lqr policy gradient}]
    We now show how Proposition~\ref{prop: lqr policy gradient} can be employed to validate the claims of Theorem~\ref{thm: lqr policy gradient}. Note that 
\begin{align*}
    \PP \left\{ \Delta_T \geq \eps \right\} \leq & \PP \left\{ \Delta_T 1_{\tau > T} \geq \eps \right\} + \PP \left\{ 1_{\tau \leq T} = 1 \right\} \\
    \overset{\mathrm{(i)}}{\leq} &\frac{1}{\eps} \E[\Delta_T 1_{\tau > T}] + \PP \left\{ \tau \leq T \right\} \\
    \overset{\mathrm{(ii)}}{\leq} &\frac{1}{20} + \frac{3}{20} + \delta T = \frac{1}{5} + \delta T,
\end{align*}
where (i) follows from Markov's inequality and (ii) follows from Proposition~\ref{prop: lqr policy gradient}.
\end{proof}

In the next section, we present a brief simulation study using two representative examples from~\cite{DM-AP-KB-KK-PLB-MJW:18} to empirically validate our theoretical guarantees and compare convergence rates.

\section{Simulation studies}

We now revisit several examples introduced from the previous literature (specifically from ~\cite{DM-AP-KB-KK-PLB-MJW:18}) and show empirically that our performance does indeed match $\tilde{O}(\epsilon^{-1})$ guaranteed by our theoretical results.

We begin with the following LQR problem:
\begin{align*}
    A = \begin{bmatrix}
        1 & 0 & -10 \\ -1 & 1 & 0 \\ 0 & 0 & 1
    \end{bmatrix}, \;
    B = \begin{bmatrix}
        1 & -10 & 0 \\ 0 & 1 & 0 \\ -1 & 0 & 1
    \end{bmatrix}, \;
    Q = \begin{bmatrix}
        2 & -1 & 0 \\ -1 & 2 & -1 \\ 0 & -1 & 2
    \end{bmatrix}, \;
    R = \begin{bmatrix}
        5 & -3 & 0 \\ -3 & 5 & -2 \\ 0 & -2 & 5
    \end{bmatrix},
\end{align*}
under the random initialization setting, where the initial state \( x_0 \) is sampled uniformly from the set of signed canonical basis vectors, yielding a mean-zero distribution. The discount factor is set to \( \gamma = 0.9 \). We initialize with a policy \( K_0 \) satisfying \( \Cinit(K_0) - \Cinit(K^*) = 11.716 \), use a constant step-size of \( 10^{-4} \), and set the batch size to \( N_s = 10^3 \).
This example was previously considered in~\cite{DM-AP-KB-KK-PLB-MJW:18} under a two-point gradient estimation scheme, where an empirical sample complexity of approximately \( \mathcal{O}(\epsilon^{-1}) \) was observed (see their Figure~2~(b)). As shown in Figure~\ref{fig:sample_complexity_ex1}, our method achieves a fitted rate of approximately \( \mathcal{O}(\epsilon^{-1.03}) \), in line with our theoretical guarantees of \( \widetilde{\mathcal{O}}(\epsilon^{-1}) \), with the small discrepancy likely due to logarithmic factors.

We next consider a second example from~\cite{DM-AP-KB-KK-PLB-MJW:18}, this time under the noisy dynamics setting:
\[
    A = 0.1 I_3, \quad B = 0.01 I_3, \quad Q = 100 I_3, \quad R = 100 I_3,
\]
with the discount factor again set to \( \gamma = 0.9 \). The system is subject to additive Gaussian noise with zero mean and covariance \( \frac{1}{25} I_3 \). We initialize with a policy \( K_0 \) such that \( \Cdyn(K_0) = \Cdyn(K^*) + 3.12 \), and apply a time-varying step-size given by
\[
    \alpha_t = \max\left( \frac{1}{60t + 2000}, 2 \cdot 10^{-5} \right),
\]
along with a batch size of \( N_s = 3000 \). This choice allows us to apply the time-varying step-size scheme from Theorem~\ref{thm: lqr policy gradient} (or, equivalently, from Corollary~\ref{cor: lqr policy gradient - noisy} in Appendix~\ref{app: extension_to_noisy_dynamics} for the noisy dynamics setting), although we note that a constant step-size performs similarly well in practice. The same problem was studied in~\cite{DM-AP-KB-KK-PLB-MJW:18} under a one-point estimation scheme, where a sample complexity of approximately \( \mathcal{O}(\epsilon^{-2}) \) was observed (see their Figure~2~(c)). As can be seen from Figure~\ref{fig:sample_complexity_ex2}, our empirical rate is approximately \( \mathcal{O}(\epsilon^{-0.87}) \), satisfying our theoretical guarantee of at most \( \widetilde{\mathcal{O}}(\epsilon^{-1}) \).

\begin{figure}[t]
    \centering
    \subfigure[Random initialization setting.]{
        \includegraphics[width=0.47\linewidth]{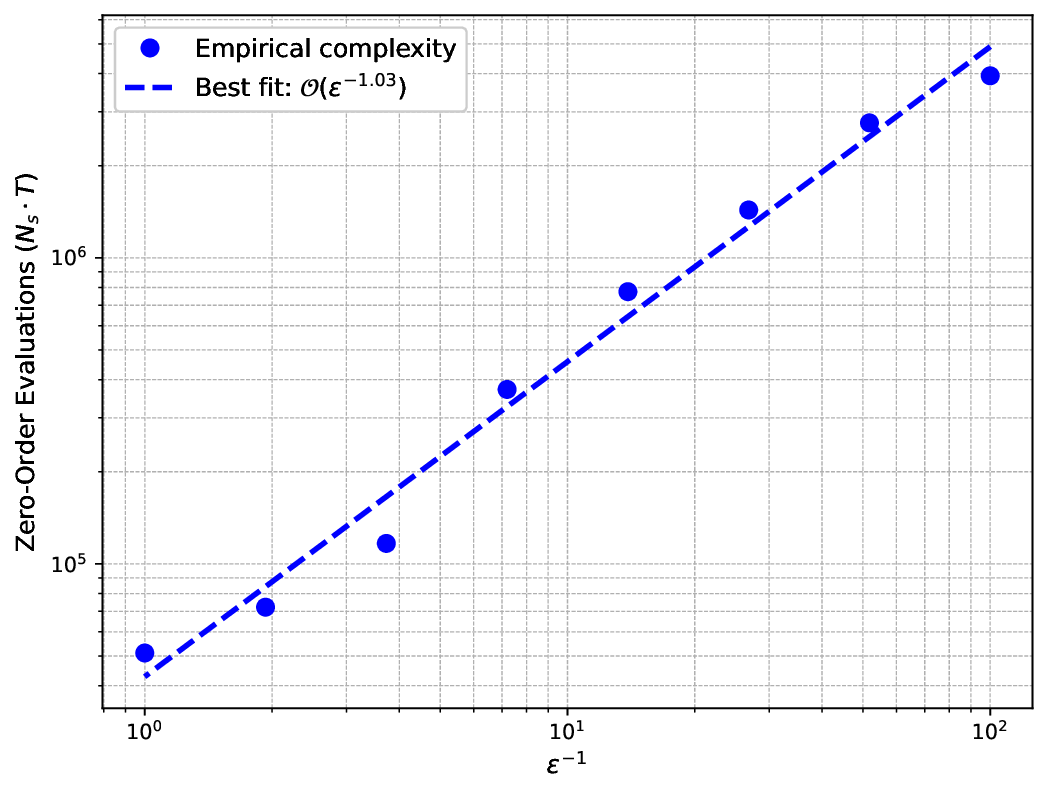}
        \label{fig:sample_complexity_ex1}
    }
    \hfill
    \subfigure[Noisy dynamics setting.]{
        \includegraphics[width=0.47\linewidth]{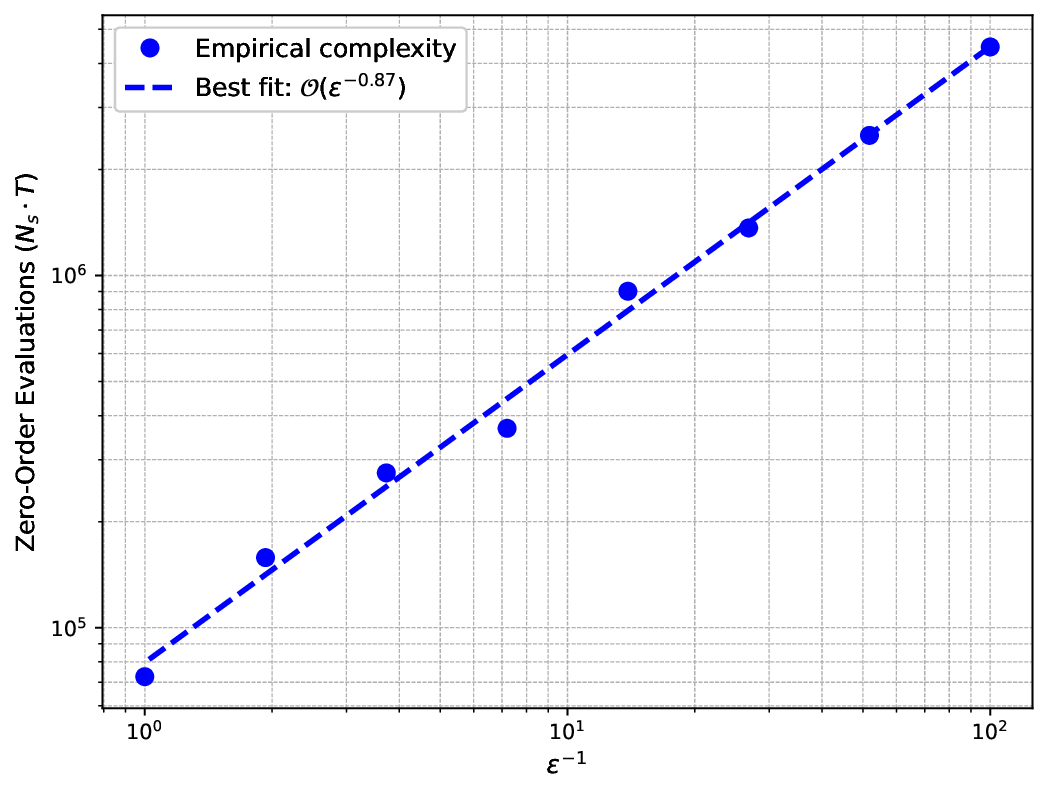}
        \label{fig:sample_complexity_ex2}
    }
    \caption{Empirical zero-order evaluations required by the policy gradient method to achieve $\epsilon$-optimality. Dashed lines indicate the best-fit lines in the log-log scale. The plots were generated by averaging 20 runs of Algorithm~\ref{alg: LQR_policy_gradient}.}
    \label{fig:sample_complexity_comparison}
\end{figure}

\section{Summary and discussion}
We have provided an algorithm with $\eps$-optimality guarantees with a provable convergence rate of $\widetilde{\mathcal{O}} (1/ \eps)$ for the discounted discrete-time LQR problem in the model-free setting. This was made possible by employing a gradient estimation technique inspired by REINFORCE, combined with a time-varying step-size. Our results contrast from the ones obtained by 
two-point methods---which 
make the stronger assumption of access to cost for two different policies with the same realization of all system randomness---as well as results that assume stability of the obtained policies throughout the algorithm.

An interesting future direction would be to investigate an actor-critic approach that could maintain the rate without requiring further assumptions. Moreover, one could consider an extension of the presented results for the undiscounted case; in particular, the current analysis of gradient estimation with one zero-order evaluation per iteration 
heavily relies on sampling from a distribution whose definition relies on the discount factor be strictly less that one. 
\appendix
\section{Probabilty of failure argument} \label{app: prob of failure}

We dedicate this section to addressing our constant probability guarantees in Theorem~\ref{thm: lqr policy gradient}. To that end, and inspired by the approach in~\cite[Appendix E]{DM-AP-KB-KK-PLB-MJW:18}, we propose a mini-batched gradient estimation method, in which we average a sufficiently large number of i.i.d. copies of our original gradient estimate to obtain a more accurate approximation of the true gradient with high probability. Consider the mini-batch gradient estimate
\begin{equation}
\overline{\nabla \C}_{N_s} (K) := \frac{1}{N_s} \sum_{i=1}^{N_s} \widehat{\nabla \C}_i (K),
\label{eq: minibatch_gradient_estimate}
\end{equation}
where each $\widehat{\nabla \C}_i (K)$ is an i.i.d. copy of $\widehat{\nabla \C} (K)$ in~\eqref{eq: gradient estimate practical formulation}. We provide the following lemma regarding the concentration of this averaged estimate around its expectation, which is equal to the actual gradient as shown in Propositon~\ref{prop: gradient estimate expectation}.
\begin{lemma} \label{lem: averaged_estimate_error_bound}
Suppose $K \in \G^{\text{lqr}}$, $\gamma$ is chosen as in Lemma~\ref{lem: gamma condition}, and $\delta > 0$ chosen to satisfy
\begin{equation}
    \delta \leq \min \left\{ e^{-3/2}, \frac{ 1-\gamma}{3 \xi_3} \sqrt{\frac{\mu_{\text{lqr}} \eps}{8}} \right\}. \label{eq: delta_choice_in_appendix}
\end{equation}
If $N_s$ is selected such that 
\begin{align}
    N_s \geq \Bigg\lceil \max \Bigg\{ 5000, 8 \left( \log \frac{2}{\delta} \right)^3, &\frac{2048 \xi_3^2}{9 (1-\gamma)^2 \mu_{\text{lqr}}} \frac{1}{\eps} \left( \log \frac{2(mn+1)}{\delta} \right)^2, \cr
    &\frac{128 \xi_4}{\mu_{\text{lqr}} (1-\gamma)^2} \frac{1}{\eps} \log \frac{2(mn+1)}{\delta} \Bigg\} \Bigg\rceil = \widetilde{\mathcal{O}} \left(\frac{1}{\eps}\right), \label{eq: N_s choice}
\end{align}
then the mini-batch averaged estimate~\eqref{eq: minibatch_gradient_estimate} satisfies
\begin{equation*}
    \| \overline{\nabla \C}_{N_s} (K) - \nabla \C(K) \|_F \leq \sqrt{\frac{\mu_{\text{lqr}} \eps}{8}},
\end{equation*}
with probability at least $1-\delta$.
\end{lemma}
\begin{proof}
Let us define the following event
\[
\mathcal{B}_i = \left\{ \| \widehat{\nabla \C}_i (K) \|_F \leq \frac{\xi_3}{1 - \gamma} \left(\log \frac{2 N_s}{\delta}\right)^{3/2} \right\},
\]
which holds with probability at least $1 - \frac{\delta}{2 N_s}$ for each $i$. As a result, following Lemma~\ref{lem: gradient estimate conditioned}, we have for all $i \in \{ 1,2,\dotsc,N_s \}$ that
\begin{equation}
\| \E [\widehat{\nabla \C}_i (K) 1_{\mathcal{B}_i}] - \nabla \C(K) \|_F \leq \frac{3 \xi_3}{1 - \gamma} \frac{\delta}{2 N_s} \left(\log \frac{2 N_s}{\delta}\right)^{3/2}, \label{eq: gradient_est_conditioned_recall}
\end{equation}
where $\E [\widehat{\nabla \C}_i (K) 1_{\mathcal{B}_i}]$ holds the same value for all $i$. Moreover, note that
\begin{align}
    \overline{\nabla \C}_{N_s} (K) - \nabla \C(K) &= \frac{1}{N_s} \sum_{i=1}^{N_s} \left( \widehat{\nabla \C}_i (K) 1_{\mathcal{B}^c_i} + \widehat{\nabla \C}_i (K) 1_{\mathcal{B}_i} - \nabla \C(K) \right) \cr
    &= \frac{1}{N_s} \sum_{i=1}^{N_s} \left( \widehat{\nabla \C}_i (K) 1_{\mathcal{B}^c_i} + \widehat{\nabla \C}_i (K) 1_{\mathcal{B}_i} - \E [\widehat{\nabla \C}_i (K) 1_{\mathcal{B}_i}] \right) \cr
    &\quad + \frac{1}{N_s} \sum_{i=1}^{N_s} \E [\widehat{\nabla \C}_i (K) 1_{\mathcal{B}_i}] - \nabla \C(K) \cr
    &= \frac{1}{N_s} \sum_{i=1}^{N_s} \left( \widehat{\nabla \C}_i (K) 1_{\mathcal{B}^c_i} + \widehat{\nabla \C}_i (K) 1_{\mathcal{B}_i} - \E [\widehat{\nabla \C}_i (K) 1_{\mathcal{B}_i}] \right) \cr
    &\quad + \E [\widehat{\nabla \C}_1 (K) 1_{\mathcal{B}_1}] - \nabla \C(K). \label{eq: Bernstein_prelim1}
\end{align}
Let us now define
\[
S_i := \widehat{\nabla \C}_i (K) 1_{\mathcal{B}_i} - \E [\widehat{\nabla \C}_i (K) 1_{\mathcal{B}_i}].
\]
so we can utilize~\eqref{eq: Bernstein_prelim1} to write
\begin{align}
    &\| \overline{\nabla \C}_{N_s} (K) - \nabla \C(K) \|_F \cr
    \leq &\frac{1}{N_s} \sum_{i=1}^{N_s} \|\widehat{\nabla \C}_i (K) 1_{\mathcal{B}^c_i} \|_F + \| \frac{1}{N_s} \sum_{i=1}^{N_s} S_i \|_F + \| \E [\widehat{\nabla \C}_1 (K) 1_{\mathcal{B}_1}] - \nabla \C(K) \|_F \cr
    \overset{\mathrm{(i)}}{\leq} &\frac{1}{N_s} \sum_{i=1}^{N_s} \|\widehat{\nabla \C}_i (K) \|_F 1_{\mathcal{B}^c_i} + \| \frac{1}{N_s} \sum_{i=1}^{N_s} S_i \|_F + \frac{3 \xi_3}{1 - \gamma} \frac{\delta}{2 N_s} \left(\log \frac{2 N_s}{\delta}\right)^{3/2}, \label{eq: Bernstein_prelim2}
\end{align}
where (i) follows from~\eqref{eq: gradient_est_conditioned_recall}. For the first term in~\eqref{eq: Bernstein_prelim2}, we have
\begin{align}
    \PP \left\{ \frac{1}{N_s} \sum_{i=1}^{N_s} \|\widehat{\nabla \C}_i (K) \|_F 1_{\mathcal{B}^c_i} = 0 \right\} \geq \PP \left\{ \bigcap\limits_{i=1}^{N_s} \mathcal{B}_i \right\} \geq 1 - \sum_{i=1}^{N_s} \PP \{\mathcal{B}_i^c \} \geq 1 - \frac{\delta}{2}. \label{eq: averaging_error_first_term_bound}
\end{align}
Additionally, we can use the matrix Bernstein theorem to bound the second term in~\eqref{eq: Bernstein_prelim2} with high probability~\cite[Theorem 1.6.2]{JAT:15}. In order to do so, first observe that $S_i$'s are i.i.d. random matrices and satisfy
\[
\E [S_i] = 0, \; \; \text{and} \; \; \|S_i \|_F \leq \frac{2 \xi_3}{1 - \gamma} \left(\log \frac{2 N_s}{\delta}\right)^{3/2}
\]
for all $i \in \{1,2,\dotsc,N_s\}$. Now let
\[
Z := \sum_{i=1}^{N_s} S_i.
\]
We have
\begin{align*}
    \E [\| Z \|_F^2] &= \E [\tr (Z^\top Z)] \\
    &= \sum_{i=1}^{N_s} \sum_{j=1}^{N_s} \E [\tr (S_i^\top S_j)] \\
    &= \sum_{i=1}^{N_s} \E[\| S_i \|_F^2] \\
    &= N_s \E[\| S_1 \|_F^2] \\
    &\leq N_s \E [\| \widehat{\nabla \C}_1 (K) 1_{\mathcal{B}_1} \|_F^2] \\
    &\overset{\mathrm{(i)}}{\leq} N_s \frac{\xi_4}{(1-\gamma)^2},
\end{align*}
where (i) follows from~\eqref{eq: gradient_estimate_variance_bound} in Lemma~\ref{lem: gradient estimate bounds}. As a result, following~\cite[Theorem 1.6.2]{JAT:15}, along with an additional vectorization of the matrices (to transfer the results from 2-norm to Frobenius norm), we have
\begin{align}
\PP \{ \| \frac{1}{N_s} \sum_{i=1}^{N_s} S_i \|_F \geq t \} &= \PP \{ \| Z \|_F \geq N_s t \}  \cr
&\leq (mn + 1) \exp \left( - \frac{N_s^2 t^2}{2 N_s \frac{\xi_4}{(1-\gamma)^2} + \frac{4 \xi_3}{3(1 - \gamma)} \left(\log \frac{2 N_s}{\delta}\right)^{3/2} N_s t} \right) \cr
&= (mn + 1) \exp \left( - \frac{N_s t^2}{2 \frac{\xi_4}{(1-\gamma)^2} + \frac{4 \xi_3}{3(1 - \gamma)} \left(\log \frac{2 N_s}{\delta}\right)^{3/2} t} \right). \label{eq: Bernstein_inequality}
\end{align}
Now letting $t = \frac{1}{2}\sqrt{\frac{\mu_{\text{lqr}} \eps}{8}}$ and selecting $N_s$ as suggested in~\eqref{eq: N_s choice} lets us write~\eqref{eq: Bernstein_inequality} as
\begin{align}
    \PP \left\{ \| \frac{1}{N_s} \sum_{i=1}^{N_s} S_i \|_F \geq \frac{1}{2}\sqrt{\frac{\mu_{\text{lqr}} \eps}{8}} \right\} \leq \frac{\delta}{2}. \label{eq: averaging_error_second_term_bound}
\end{align}

For the third term in~\eqref{eq: Bernstein_prelim2}, note that due to the choice of $N_s$ in~\eqref{eq: N_s choice}, we have
\begin{align}
\frac{3 \xi_3}{1 - \gamma} \frac{\delta}{2} \frac{\left(\log \frac{2 N_s}{\delta}\right)^{3/2}}{N_s} &\overset{\mathrm{(i)}}{\leq} \frac{3 \xi_3}{1 - \gamma} \frac{\delta}{2} \frac{1}{\sqrt{N_s}} \cr
&\leq \frac{3 \xi_3}{1 - \gamma} \frac{\delta}{2} \cr
&\overset{\mathrm{(ii)}}{\leq} \frac{1}{2}\sqrt{\frac{\mu_{\text{lqr}} \eps}{8}}, \label{eq: averaging_error_third_term_bound}
\end{align}
where (i) follows from the choice of $N_s$ in~\eqref{eq: N_s choice}, and (ii) from~\eqref{eq: delta_choice_in_appendix}. Finally, applying~\eqref{eq: averaging_error_first_term_bound}, \eqref{eq: averaging_error_second_term_bound}, and \eqref{eq: averaging_error_third_term_bound}, along with union bound, on~\eqref{eq: Bernstein_prelim2} concludes the proof.
\end{proof}

We are now in a position to present the following result:
\begin{theorem} \label{thm: lqr policy gradient - averaged}
    Suppose $K_0$ is stable, $\gamma$ is as suggested in Lemma~\ref{lem: gamma condition}, and the update rule follows
    \begin{equation}
        K_{t+1} = K_t - \alpha \overline{\nabla \C}_{N_s}(K_t) \label{eq: update_rule_appendix}
    \end{equation}
    with a constant step-size $\alpha$ satisfying
    \begin{equation}
        \alpha \leq \min \left\{ \frac{\omega_{\text{lqr}}}{\widetilde{c_{K_1}} + \sqrt{\frac{\mu_{\text{lqr}} \C (K_0)}{8}}}, \frac{1}{4 \phi_{\text{lqr}}}, \frac{4}{\mu_{\text{lqr}}} \right\}. \label{eq: alpha_choice_appendix}
    \end{equation}
    Then for a given error tolerance $\eps \in (0,\C (K_0)]$, and for any $\delta$ satisfying~\eqref{eq: delta_choice_in_appendix}, the update rule~\eqref{eq: update_rule_appendix}, with $N_s \sim \widetilde{\mathcal{O}} (1/\eps)$ chosen according to~\eqref{eq: N_s choice}, guarantees that after
    \begin{equation}
        T = \frac{4}{\alpha \mu_{\text{lqr}}} \log \left( \frac{2 \C(K_0)}{\eps} \right) \label{eq: T_choice_appendix}
    \end{equation}
    iterations, we have
    \[
    \C(K_T) - \C(K^*) \leq \eps,
    \]
    with a probability of at least $1 - \delta T$.
\end{theorem}
\begin{proof}
    First, assume that $K_t \in \G^{\text{lqr}}$, then since $N_s$ is chosen as in~\eqref{eq: N_s choice}, we have from Lemma~\ref{lem: averaged_estimate_error_bound} that
    \begin{align*}
        \| \overline{\nabla \C}_{N_s} (K_t) - \nabla \C(K_t) \|_F \leq \sqrt{\frac{\mu_{\text{lqr}} \eps}{8}},
    \end{align*}
    with probability at least $1 - \delta$. Hence, conditioned on this event, we have the following bound
    \begin{align}
        \| \alpha \overline{\nabla \C}_{N_s} (K_t) \|_F &\leq \alpha \| \overline{\nabla \C}_{N_s} (K_t) - \nabla \C(K_t) + \nabla \C(K_t) \|_F \cr
        &\leq \alpha \left(\| \overline{\nabla \C}_{N_s} (K_t) - \nabla \C(K_t) \|_F + \| \nabla \C(K_t) \|_F \right) \cr
        &\leq \alpha \left( \sqrt{\frac{\mu_{\text{lqr}} \eps}{8}} + \widetilde{c_{K_1}} \right) \cr
        &\overset{\mathrm{(i)}}{\leq} \alpha \left( \sqrt{\frac{\mu_{\text{lqr}} \C(K_0)}{8}} + \widetilde{c_{K_1}} \right) \cr
        &\overset{\mathrm{(ii)}}{\leq} \omega_{\text{lqr}}, \label{eq: local_smoothness_holds_appendix}
    \end{align}
    where (i) follows from $\eps \leq \C(K_0)$ and (ii) from the choice of $\alpha$ in~\eqref{eq: alpha_choice_appendix}. Note that~\eqref{eq: local_smoothness_holds_appendix} ensures that our step-size is small enough for Lipschitz and smoothness properties to hold. Consequently, we can utilize smoothness to write
    \begin{align}
        \Delta_{t+1} - \Delta_t &= \C(K_{t+1}) - \C(K_t) \cr
        &\leq - \left\langle \nabla \C(K_t), \alpha \overline{\nabla \C}_{N_s} (K_t) \right\rangle + \frac{\phi_{\text{lqr}}}{2} \alpha^2 \| \overline{\nabla \C}_{N_s} (K_t) \|_F^2 \cr
        &= - \alpha \left\langle \nabla \C (K_t), \nabla \C(K_t) + \overline{\nabla \C}_{N_s} (K_t) - \nabla \C(K_t) \right\rangle \cr
        &\quad + \frac{\phi_{\text{lqr}}}{2} \alpha^2 \left( \| \nabla \C(K_t) + ( \overline{\nabla \C}_{N_s} (K_t) - \nabla \C(K_t) ) \|_F^2 \right) \cr
        &\leq - \alpha \| \nabla \C(K_t) \|_F^2 + \alpha \| \nabla \C(K_t) \|_F \| \overline{\nabla \C}_{N_s} (K_t) - \nabla \C(K_t) \|_F \cr
        &\quad + \phi_{\text{lqr}} \alpha^2 \| \overline{\nabla \C}_{N_s} (K_t) - \nabla \C(K_t) \|_F^2 + \phi_{\text{lqr}} \alpha^2 \| \nabla \C(K_t) \|_F^2 \cr
        &\leq - \alpha \| \nabla \C(K_t) \|_F^2 + \frac{\alpha}{2} \left( \| \nabla \C(K_t) \|_F^2 + \| \overline{\nabla \C}_{N_s} (K_t) - \nabla \C(K_t) \|_F^2 \right) \cr
        &\quad + \phi_{\text{lqr}} \alpha^2 \| \overline{\nabla \C}_{N_s} (K_t) - \nabla \C(K_t) \|_F^2 + \phi_{\text{lqr}} \alpha^2 \| \nabla \C(K_t) \|_F^2 \cr
        &= - \frac{\alpha}{2} \| \nabla \C(K_t) \|_F^2 + \phi_{\text{lqr}} \alpha^2 \| \nabla \C(K_t) \|_F^2 \cr
        &\quad + \left( \frac{\alpha}{2} + \phi_{\text{lqr}} \alpha^2 \right) \| \overline{\nabla \C}_{N_s} (K_t) - \nabla \C(K_t) \|_F^2 \cr
        &\overset{\mathrm{(i)}}{\leq} - \frac{\alpha}{2} \| \nabla \C(K_t) \|_F^2 + \frac{\alpha}{4} \| \nabla \C(K_t) \|_F^2 + \left( \frac{\alpha}{2} + \frac{\alpha}{4} \right) \| \overline{\nabla \C}_{N_s} (K_t) - \nabla \C(K_t) \|_F^2 \cr
        &\leq - \frac{\alpha}{4} \| \nabla \C(K_t) \|_F^2 + \alpha \| \overline{\nabla \C}_{N_s} (K_t) - \nabla \C(K_t) \|_F^2 \cr
        &\overset{\mathrm{(ii)}}{\leq} - \frac{\alpha \mu_{\text{lqr}}}{4} \Delta_t + \alpha \frac{\mu_{\text{lqr}} \eps}{8}, \label{eq: recursive_inequality_appendix_prelim}
    \end{align}
    where (i) follows from the fact that $\alpha \phi_{\text{lqr}} \leq  1/4$ due to the choice of $\alpha$ in~\eqref{eq: alpha_choice_appendix}, and (ii) from the PL inequality~\eqref{eq: PL_LQR}. Rearranging~\eqref{eq: recursive_inequality_appendix_prelim} yields
    \begin{equation}
        \Delta_{t+1} \leq \left( 1 - \frac{\alpha \mu_{\text{lqr}}}{4} \right) \Delta_t + \alpha \frac{\mu_{\text{lqr}} \eps}{8}. \label{eq: recursive_inequality_appendix}
    \end{equation}
    With this in place, we use strong induction to finalize the proof. For each time $i \in \{1,2,\dotsc,T \}$, let $\mathscr{E}_i$ denote the event that $\Delta_i \leq 10 \C(K_0)$ (implying $K_i \in \G^{\text{lqr}}$) and $\Delta_{i} \leq \left( 1 - \frac{\alpha \mu_{\text{lqr}}}{4} \right) \Delta_{i-1} + \alpha \frac{\mu_{\text{lqr}} \eps}{8}$. We claim that for each $t \in \mathbb{N}$, it holds that
    \[
    \PP \left\{ \cap_{i=1}^t \mathscr{E}_i \right\} \geq 1 - \delta t.
    \]
    We demonstrate this by induction as follows:
    
    \textbf{Base case ($t = 0$):} Since $K_0 \in \G^{\text{lqr}}$, we have by Lemma~\ref{lem: averaged_estimate_error_bound} and inequality~\eqref{eq: recursive_inequality_appendix} that
    \[
    \Delta_{1} \leq \left( 1 - \frac{\alpha \mu_{\text{lqr}}}{4} \right) \Delta_0 + \alpha \frac{\mu_{\text{lqr}} \eps}{8}.
    \]
    Moreover, since $\alpha \leq \frac{4}{\mu_{\text{lqr}}}$ and $\eps \leq \C(K_0)$, we have that $\Delta_1 \leq \Delta_0 + \frac{1}{2} \C(K_0) \leq 10 \C(K_0)$. Thus, we have shown that $\mathscr{E}_1$ holds with probability at least $1 - \delta$, establishing the base case.
    
    \textbf{Inductive step:} By induction hypothesis, we have that the event $\cap_{i=1}^t \mathscr{E}_i$ holds with probability at least $1 - \delta t$. Conditioned on this event, we have by Lemma~\ref{lem: averaged_estimate_error_bound} and inequality~\eqref{eq: recursive_inequality_appendix} that with probability at least $1 - \delta$, the following holds
    \begin{align}
        \Delta_{t+1} &\leq \left( 1 - \frac{\alpha \mu_{\text{lqr}}}{4} \right) \Delta_t + \alpha \frac{\mu_{\text{lqr}} \eps}{8} \cr
        &\leq \left( 1 - \frac{\alpha \mu_{\text{lqr}}}{4} \right)^{t+1} \Delta_0+ \alpha \frac{\mu_{\text{lqr}} \eps}{8} \sum_{i=0}^t \left( 1 - \frac{\alpha \mu_{\text{lqr}}}{4} \right)^i \cr
        &\leq \left( 1 - \frac{\alpha \mu_{\text{lqr}}}{4} \right)^{t+1} \Delta_0 + \alpha \frac{\mu_{\text{lqr}} \eps}{8} \sum_{i=0}^{\infty} \left( 1 - \frac{\alpha \mu_{\text{lqr}}}{4} \right)^i \cr
        &= \left( 1 - \frac{\alpha \mu_{\text{lqr}}}{4} \right)^{t+1} \Delta_0 + \frac{\eps}{2}, \label{eq: inductive_step_appendix}
    \end{align}
    and since $\eps \leq \C(K_0)$, we also have $\Delta_{t+1} \leq 10 \C(K_0)$. Now combining this with a union bound shows that $\cap_{i=1}^{t+1} \mathscr{E}_i$ holds with a probability of at least $1 - (\delta t + \delta) = 1 - \delta (t+1)$, completing the inductive step.

    Finally, conditioned on $\cap_{i=1}^T \mathscr{E}_i$, similar to~\eqref{eq: inductive_step_appendix}, we obtain
    \begin{align*}
        \Delta_T &\leq \left( 1 - \frac{\alpha \mu_{\text{lqr}}}{4} \right)^{T} \Delta_0 + \frac{\eps}{2} \\
        &\overset{\mathrm{(i)}}{\leq} \left[ \left( 1 - \frac{\alpha \mu_{\text{lqr}}}{4} \right)^{\frac{4}{\alpha \mu_{\text{lqr}}}} \right]^{\log (\frac{2 \C(K_0)}{\eps})} \Delta_0 + \frac{\eps}{2} \\
        &\leq \left(e^{-1}\right)^{\log (\frac{2 \C(K_0)}{\eps})} \Delta_0 + \frac{\eps}{2} \\
        &= \frac{\eps}{2 \C(K_0)} \Delta_0 + \frac{\eps}{2} \\
        &\leq \eps,
    \end{align*}
    where (i) follows from the choice of $T$ in~\eqref{eq: T_choice_appendix}. This, along with recalling $\PP \left\{\cap_{i=1}^T \mathscr{E}_i \right\} \geq 1 - \delta T$, concludes the proof.
\end{proof}
\begin{remark} \label{rem: gamma condition flexibility}
    As discussed after Lemma~\ref{lem: gamma condition}, the condition on $\gamma$ depends only on the cost bound used to define the set $\mathcal{G}^{\text{lqr}}$. In particular, from the induction step in the proof of Theorem~\ref{thm: lqr policy gradient - averaged}, one can deduce that this bound can be tightened from $10 \mathcal{C}(K_0) + \mathcal{C}(K^*)$ to $\mathcal{C}(K_0) + \mathcal{C}(K^*)$, while still preserving the convergence guarantees (i.e., achieving $\varepsilon$-optimality with probability exceeding $1 - \delta T$ for small enough $\varepsilon$). This effectively enlarges the allowable range of $\gamma$; for example, the alternative set
    \[
    \mathcal{G}'^{\text{lqr}} = \{ K \mid \mathcal{C}(K) - \mathcal{C}(K^*) \leq \mathcal{C}(K_0) \}
    \]
    admits any $\gamma$ in the interval $\left(1 - \frac{\sigma_{\min}(Q)}{2\,\mathcal{C}_{\text{und}}(K_0)},\, 1\right)$, which is more permissive than the condition stated in Lemma~\ref{lem: gamma condition}.
\end{remark}

\section{Extension to noisy dynamics setting} \label{app: extension_to_noisy_dynamics}
In this section, we show how everything from the Random Initialization setting trasnfers into the noisy dynamics setup. We begin by establishing an exponential decay bound on $\|(A - BK)^t\|$, which serves as a key technical tool for the results that follow. 
\subsection{Exponential decay in the closed-loop system}
Before we introduce the next result, let us define
\begin{align}
M &:= \sqrt{\frac{10 \Cinit(K_0) + \Cinit(K^*)}{\lambda_{\min} (Q)}}, \quad \text{and} \cr
r &:= \sqrt{1 - \frac{0.5 \lambda_{\min}(Q)}{10 \Cinit(K_0) + \Cinit(K^*) - 0.5 \lambda_{\min}(Q)}} \; \in (0,1). \label{eq: M, r definition}
\end{align}
\begin{lemma} \label{lem: exponential decay}
    Suppose $\gamma \in \left( 1 - \frac{0.5 \sigma_{\min} (Q)}{11 \C_{\text{und}} (K_0)},1 \right)$. Then for any $K \in \G^{\text{lqr}}$, it holds that
    \begin{equation*}
        \|(A-BK)^{t}\|_2 \leq M r^t.
        \end{equation*}
    \end{lemma}
\begin{proof}
Let \(P_K\) denote the unique positive‑definite solution of the discrete algebraic Riccati equation
\[
P_K \;=\;
Q + K^{\top} R K + \gamma (A-BK)^{\top} P_K (A-BK).
\]

Re‑arranging gives the Lyapunov inequality
\[
\gamma (A-BK)^{\top} P_K (A-BK)
  = P_K - \bigl(Q + K^{\top} R K\bigr)
  \preceq (1-a_K) P_K,
\]
where
\[
a_K \;=\;
\frac{\lambda_{\min}\!\bigl(Q + K^{\top} R K\bigr)}
     {\lambda_{\max}(P_K)}
\in (0,1].
\]
Define \(\displaystyle b_K := \sqrt{1-a_K}\in(0,1)\); then
\begin{equation*}
\gamma (A-BK)^{\top} P_K (A-BK)
\;\preceq\; b_K^{2} P_K,
\end{equation*}
and hence,
\begin{align}
\left[(A-BK)^t\right]^{\top} P_K (A-BK)^t
&\preceq \frac{b_K^{2}}{\gamma} \left[(A-BK)^{t-1}\right]^{\top} P_K (A-BK)^{t-1} \cr
&\preceq \cdots \preceq \left(\frac{b_K^{2}}{\gamma}\right)^t P_K, \label{eq: prelim_1}
\end{align}

Equip \(\mathbb{R}^{n}\) with the quadratic norm
\(
\|x\|_{P_K} := \sqrt{x^{\top}P_K x}.
\)
From \eqref{eq: prelim_1} we obtain
\[
\| (A-BK)^t x \|_{P_K}
   \le \left(\frac{b_K}{\sqrt{\gamma}}\right)^t \|x\|_{P_K}
   \quad\forall x\in\mathbb{R}^{n},
\]
hence for every integer \(t\ge 0\)
\[
\|(A-BK)^{t}\|_{P_K} \;\le\; \left( \frac{b_K}{\sqrt{\gamma}} \right)^t,
\]
where $\| . \|_{P_K}$ is the $P_K$-induced matrix norm. Because all norms on a finite‑dimensional space are equivalent,
\[
\|x\|_2^2 \leq \lambda_{\min}^{-1} \|x\|_{P_K}^2,
  \qquad
\|x\|_{P_K}^2 \leq \lambda_{\max}(P_K) \|x\|_2^2,
\]
so the operator norm induced by \(\|\cdot\|_2\) satisfies
\begin{align*}
\|(A-BK)^{t}\|_2 \leq 
\sqrt{\frac{\lambda_{\max}(P_K)}{\lambda_{\min}(P_K)}} \| (A-BK)^t \|_{P_K} \leq \sqrt{\frac{\lambda_{\max}(P_K)}{\lambda_{\min}(P_K)}} \left( \frac{b_K}{\sqrt{\gamma}} \right)^t.
\end{align*}
Now note that we have that for $K \in \mathcal{G}^{\text{lqr}}$,
\begin{align*}
    10 \Cinit(K_0) + \Cinit(K^*) \geq \Cinit(K) = \mathrm{tr} (P_K) \geq \lambda_{\max} (P_K),
\end{align*}
and hence,
\[
\lambda_{\max} (P_K) \leq 10 \Cinit(K_0) + \Cinit(K^*) =: \lambda_1.
\]
As a result of this, all the previously used values for bounding $\| (A-BK)^t \| $ can be bounded by constants independent of $K$:
\begin{align*}
    &\lambda_{\min} (P_K) \geq \lambda_{\min} (Q) =: \lambda_2 \\
    &a_K \geq \frac{\lambda_{\min} (Q)}{\lambda_{\max} (P_K)} \geq \frac{\lambda_2}{\lambda_1} \\
    &b_K = \sqrt{1 - a_K} \leq \sqrt{1 - \frac{\lambda_2}{\lambda_1}} \\
    &\sqrt{\frac{\lambda_{\max}(P_K)}{\lambda_{\min}(P_K)}} \leq \sqrt{\frac{\lambda_1}{\lambda_2}}.
\end{align*}
Now since by assumption,
\[
\gamma \geq 1 - \frac{0.5 \sigma_{\min} (Q)}{11 \C_{\text{und}} (K_0)}\geq 1 - \frac{0.5 \lambda_2}{\lambda_1},
\]
we also conclude that 
\[
\left( \frac{b_K}{\sqrt{\gamma}} \right) \leq \sqrt{\frac{\lambda_1 - \lambda_2}{\lambda_1 - 0.5 \lambda_2}} = \sqrt{1 - \frac{0.5 \lambda_2}{\lambda_1 - 0.5 \lambda_2}}; 
\]
therefore,
\begin{equation}
\|(A-BK)^{t}\|_2 \leq \sqrt{\frac{\lambda_1}{\lambda_2}} \; \left(1 - \frac{0.5 \lambda_2}{\lambda_1 - 0.5 \lambda_2}\right)^{t/2}, \label{eq: (A-BK)^t norm decays exponentially}
\end{equation}
which is independent of $K$ as long as we are withing the $\mathcal{G}^{\text{lqr}}$ set. Substituting the values of $\lambda_1$ and $\lambda_2$ finishes the proof.
\end{proof}
Finally, note that for the noisy dynamics setting, if we let
\begin{align}
\mathcal{G}_{\text{dyn}}^{\text{lqr}} = \{ K \: | \: \Cdyn(K) - \Cdyn(K^*) \leq 10 \Cdyn(K_0) \}, \label{eq: G_lqr noisy dynamics}
\end{align}
since $\Cdyn(K) = \frac{\gamma}{1 - \gamma} \Cinit (K)$ due to Lemma~\ref{lem: noisy-random cost equivalence}, this set is the exact same as~\eqref{eq: G_LQR} in the random initialization setting. Therefore, all the bounds leading to~\eqref{eq: (A-BK)^t norm decays exponentially} hold with exactly the same values for the noisy dynamics case as well.

The exponential decay bound established in Lemma~\ref{lem: exponential decay} plays a crucial role in bounding the gradient estimate under the noisy dynamics setup. We now turn to this estimate, show that it remains unbiased and admits similar concentration bounds in this setting, and re-establish the main convergence guarantees for both standard and mini-batched policy updates.

\subsection{Gradient estimation and convergence results}
Suppose $K \in \G^{\text{lqr}}_{\text{dyn}}$, with $\G^{\text{lqr}}_{\text{dyn}}$ defined in~\eqref{eq: G_lqr noisy dynamics}. Now let us define $Q^K_{\text{dyn}} (x_{\hat{t}},u_{\hat{t}})$ as
\begin{align*}
Q^K_{\text{dyn}} (x_{\hat{t}},u_{\hat{t}}) := x_{\hat{t}}^\top Q x_{\hat{t}} + u_{\hat{t}}^\top R u_{\hat{t}} + \sum_{t=\hat{t}+1}^\infty \gamma^{t - \hat{t}} x_{t}^\top (Q + K^\top R K) x_t,
\end{align*}
where
\begin{align*}
    x_{\hat{t} + 1} = A x_{\hat{t}} + B u_{\hat{t}} + z_{\hat{t}},
\end{align*}
and
\begin{align*}
    x_{t + 1} = (A-BK) x_t + z_t,
\end{align*}
for all $t \neq \hat{t}$, with $x_0 = 0$ and i.i.d. additive noise sequence $z_t \sim \D$ for all $t$. As a result, for every $t \geq \hat{t} + 1$,
\[
x_t
  =(A-BK)^{\,t-\hat t-1}\bigl(A x_{\hat{t}} + B u_{\hat{t}} \bigr) + \sum_{i=0}^{t-\hat t-1}(A-BK)^{\,t-\hat t-1-i}z_{\hat t+i},
\]
which is affine in $u_{\hat t}$.  Combining this with the fact that each stage cost $x_{t}^\top (Q + K^\top R K) x_t$ is quadratic in $x_t$ yields a quadratic function of
$u_{\hat t}$. Therefore,
\[
Q^K_{\text{dyn}} (x_{\hat{t}},u_{\hat{t}}) = \underbrace{x_{\hat{t}}^\top Q x_{\hat{t}}}_{\text{independent of }u_{\hat{t}}} + \underbrace{u_{\hat{t}}^\top R u_{\hat{t}}}_{\text{quadratic in }u_{\hat{t}}} + \underbrace{\sum_{t=\hat{t}+1}^\infty \gamma^{t - \hat{t}} x_{t}^\top (Q + K^\top R K) x_t}_{\text{quadratic in }u_{\hat{t}}}
\]
is quadratic in $u_{\hat{t}}$, satisfying the condtion in Remark~\ref{rem: extension beyond LQR}. Following this, we have that the gradient estimate
\begin{equation}
    \widehat{\nabla \Cdyn} (K) := - \frac{1}{\sigma (1 - \gamma)} Q^{K}_{\text{dyn}} (x_{\hat{t}}, -K x_{\hat{t}} + \sigma \eta_{\hat{t}}) \eta_{\hat{t}} x_{\hat{t}}^\top \label{eq: gradient estimate practical formulation - noisy dynamics}
\end{equation}
satisfies
\begin{corollary} \label{cor: gradient estimate expectation noisy dynamics}
    Suppose $\hat{t} \sim \mu_\gamma$ and $\eta_{\hat{t}} \sim \N (0, I_m)$ as before. Then for any given $K$,
    \begin{equation*}
        \E [\widehat{\nabla \Cdyn} (K)] = \nabla \Cdyn (K).
    \end{equation*} 
\end{corollary}
The proof of this is a direct consequence of Remark~\ref{rem: extension beyond LQR}. We now introduce a result similar to~\ref{lem: gradient estimate bounds} where we provide some bounds on this gradient estimate in the noisy dynamics setting.
\begin{lemma} \label{lem: gradient estimate bounds - noisy}
    Suppose $\delta \in (0, \frac{1}{e}]$, and $\gamma$ is chosen as in Lemma~\ref{lem: exponential decay}. 
    Then for any $K \in \G^{\text{lqr}}$, we have that
    \begin{equation*}
        \lVert \widehat{\nabla \Cdyn} (K) \rVert_F \leq \frac{\tilde{\xi}_3}{1 - \gamma} \left(\log \frac{1}{\delta}\right)^{3/2}
    \end{equation*}
    with probability at least $1-\delta$, where 
$\tilde{\xi}_1, \tilde{\xi}_2, \tilde{\xi}_3 \in \real$ are given by
    \begin{align*}
        \tilde{\xi}_1 &:= \frac{M^3 C_m^{3/2}}{(1-r)^3} \left( \lVert Q \rVert + 2 \lVert R \rVert \widetilde{c_{K_1}}^2 + 2 \gamma \left(\lVert Q \rVert + \lVert R \rVert \widetilde{c_{K_1}}^2\right) \frac{(M^2 r+2)^2}{1 - \gamma} \right) \\
        \tilde{\xi}_2 &:= \frac{2M C_m^{1/2}}{1-r} \left( \lVert R \rVert + \gamma \left(\lVert Q \rVert + \lVert R \rVert \widetilde{c_{K_1}}^2\right) \frac{M^2 \|B\|^2}{1 - \gamma}\right) \\
        \tilde{\xi}_3 &:= \frac{1}{\sigma} \left(\tilde{\xi}_1 5^{1/2} m^{1/2}\right) + \sigma \left( \tilde{\xi}_2 5^{3/2} m^{3/2} \right), 
    \end{align*}
    where $M$ and $r$ are defined in~\eqref{eq: M, r definition}. Moreover, 
    \begin{equation*}
        \E \lVert \widehat{\nabla \C} (K) \rVert_F^2 \leq \frac{\tilde{\xi}_4}{(1 - \gamma)^2},
    \end{equation*}
    where    
    \begin{align*}
        \tilde{\xi}_4 &:= \frac{1}{\sigma^2} \tilde{\xi}_1^2 m + 2 \tilde{\xi}_1 \tilde{\xi}_2 m (m+2) + \sigma^2 \tilde{\xi}_2^2 m (m+2) (m+4).
    \end{align*}
\end{lemma}
\begin{proof}
    First, note that since $x_0=0$ in this setting, it holds that
    \[
    x_{\hat{t}} = \sum_{i=0}^{\hat{t} - 1} (A-BK)^i z_{\hat{t} - 1 - i},
    \]
    and hence,
    \begin{align}
        \| x_{\hat{t}} \| \leq \sum_{i=0}^{\hat{t} - 1} \| (A-BK)^i \| \|z_{\hat{t} - 1 - i} \| \overset{\mathrm{(i)}}{\leq} \sum_{i=0}^{\hat{t} - 1} (M r^i) C_m^{1/2} \leq M C_m^{1/2} \sum_{i=0}^{\infty} r^i = \frac{M C_m^{1/2}}{1-r}, \label{eq: x_t_hat size bound}
    \end{align}
    where (i) follows from Lemma~\ref{lem: exponential decay} and assumption~\eqref{eq: noise_assumption} on the additive noise. Moreover, we have
    \[
    x_{\hat{t}+1} = (A-BK) x_{\hat{t}} + \sigma B \eta_{\hat{t}} + z_{\hat{t}},
    \]
    and thus,
    \begin{align}
        \| x_{\hat{t}+1} \| &\leq \| A-BK\| \| x_{\hat{t}} \| + \sigma \| B \| \|\eta_{\hat{t}} \| + C_m^{1/2} \leq (Mr) \frac{M C_m^{1/2}}{1-r} + \sigma \| B \| \|\eta_{\hat{t}} \| + C_m^{1/2}. \label{eq: x_t_hat_+1 size bound}
    \end{align}
    Additionally, for all $t \geq \hat{t} + 1$, we can write
    \[
    x_t = (A-BK)^{t-\hat{t}-1} x_{\hat{t}+1} + \sum_{i=0}^{t-\hat{t}-2} (A-BK)^i z_{t-1-i},
    \]
    and hence,
    \begin{align}
        \| x_t \| &\leq M \| x_{\hat{t}+1} \| + \sum_{i=0}^{t-\hat{t}-2} (M r^i) C_m^{1/2} \cr
        &\overset{\mathrm{(i)}}{\leq} M \left( \frac{M^2 r C_m^{1/2}}{1-r} + \sigma \|B\| \|\eta_{\hat{t}}\|+C_m^{1/2} \right) + \frac{M C_m^{1/2}}{1-r} \cr
        &\leq \frac{M C_m^{1/2} (M^2 r + 2)}{1-r} + \sigma M \|B\| \| \eta_{\hat{t}} \|, \label{eq: x_t size bound after t_hat}
    \end{align}
    where (i) follows from~\eqref{eq: x_t_hat_+1 size bound}. We are now in a position to show the following upper bound:
    \begin{align}
        &Q^{K}_{\text{dyn}} (x_{\hat{t}}, -K x_{\hat{t}} + \sigma \eta_{\hat{t}}) \cr
        = &x_{\hat{t}}^\top Q x_{\hat{t}} + (-K x_{\hat{t}} + \sigma \eta_{\hat{t}})^\top R (-K x_{\hat{t}} + \sigma \eta_{\hat{t}}) + \sum_{t=\hat{t}+1}^\infty \gamma^{t - \hat{t}} x_{t}^\top (Q + K^\top R K) x_t \cr
        \leq &\|Q\| \|x_{\hat{t}}\|^2  + \|R\| \|-K x_{\hat{t}} + \sigma \eta_{\hat{t}} \|^2 + \sum_{t=\hat{t}+1}^\infty \gamma^{t - \hat{t}} \left(\|Q\| + \|R\| \widetilde{c_{K_1}}^2 \right) \|x_t\|^2 \cr
        \overset{\mathrm{(i)}}{\leq} &\|Q\| \|x_{\hat{t}}\|^2  + 2 \|R\|\left( \|K\|^2 \| x_{\hat{t}} \|^2 + \sigma^2 \|\eta_{\hat{t}}\|^2 \right) \cr
        &+ \sum_{t=\hat{t}+1}^\infty \gamma^{t - \hat{t}} \left(\|Q\| + \|R\| \widetilde{c_{K_1}}^2 \right) \left( \frac{M C_m^{1/2} (M^2 r + 2)}{1-r} + \sigma M \|B\| \| \eta_{\hat{t}} \| \right)^2 \cr
        \overset{\mathrm{(ii)}}{\leq} &\|Q\| \frac{M^2 C_m}{(1-r)^2} + 2 \|R\| \left( \widetilde{c_{K_1}}^2 \frac{M^2 C_m}{(1-r)^2} + \sigma^2 \|\eta_{\hat{t}}\|^2 \right) \cr
        &+ 2 \left(\|Q\| + \|R\| \widetilde{c_{K_1}}^2 \right) \left( \frac{M^2 C_m (M^2 r + 2)^2}{(1-r)^2} + \sigma^2 M^2 \|B\|^2 \| \eta_{\hat{t}} \|^2 \right) \sum_{t=\hat{t}+1}^\infty \gamma^{t - \hat{t}} \cr
        = &\frac{M^2 C_m}{(1-r)^2} \left(\|Q\| + 2 \|R\| \widetilde{c_{K_1}}^2 + 2 \left(\|Q\| + \|R\| \widetilde{c_{K_1}}^2 \right) (M^2 r + 2)^2 \frac{\gamma}{1-\gamma} \right) \cr
        &+ 2 \sigma^2 \left( \|R\| + \left(\|Q\| + \|R\| \widetilde{c_{K_1}}^2 \right) M^2 \|B\|^2 \frac{\gamma}{1-\gamma} \right) \|\eta_{\hat{t}}\|^2, \label{eq: Q-function bound noisy dynamics}
    \end{align}
    where (i) follows from~\eqref{eq: x_t size bound after t_hat} and (ii) from~\eqref{eq: x_t_hat size bound}. Combining~\eqref{eq: x_t_hat size bound} and~\eqref{eq: Q-function bound noisy dynamics}, we have
    \begin{align}
        &\lVert \widehat{\nabla \Cdyn} (K) \rVert_F \cr
        \leq &\frac{1}{\sigma (1-\gamma)} Q^{K}_{\text{dyn}} (x_{\hat{t}}, -K x_{\hat{t}} + \sigma \eta_{\hat{t}}) \| x_{\hat{t}} \| \|\eta_{\hat{t}}\| \cr
        \overset{\mathrm{(i)}}{\leq} &\frac{1}{\sigma (1-\gamma)} \frac{M^3 C_m^{3/2}}{(1-r)^3} \left( \lVert Q \rVert + 2 \lVert R \rVert \widetilde{c_{K_1}}^2 + 2 \gamma \left(\lVert Q \rVert + \lVert R \rVert \widetilde{c_{K_1}}^2\right) \frac{(M^2 r+2)^2}{1 - \gamma} \right) \|\eta_{\hat{t}}\| \cr
        &+ \frac{\sigma}{1-\gamma} \frac{2 M C_m^{1/2}}{1-r} \left( \lVert R \rVert + \gamma \left(\lVert Q \rVert + \lVert R \rVert \widetilde{c_{K_1}}^2\right) \frac{M^2 \|B\|^2}{1 - \gamma}\right) \|\eta_{\hat{t}}\|^3 \cr
        = &\frac{1}{1-\gamma} \left( \frac{1}{\sigma} \tilde{\xi}_1 \|\eta_{\hat{t}}\| + \sigma \tilde{\xi}_2 \|\eta_{\hat{t}}\|^3 \right)
    \end{align}
    which resembles the expression of the bound~\eqref{eq: gradient estimate norm bound before LM} shown for the random initialization setting. Therefore, the rest of the proof follows exactly like that of Lemma~\ref{lem: gradient estimate bounds}, after substituting $\xi_1,\xi_2$ with $\tilde{\xi}_1, \tilde{\xi}_2$ respectively.
\end{proof}
Now note that as a consequence of Lemma~\ref{lem: noisy-random cost equivalence}, and as also pointed out in~\cite{DM-AP-KB-KK-PLB-MJW:18}, $\Cdyn (K)$ is also $(\frac{\gamma}{1-\gamma} \phi_{K}, \beta_K)$ locally smooth, $(\frac{\gamma}{1-\gamma} \lambda_K, \zeta_{K})$ locally Lipschitz, and globally $\frac{\gamma}{1-\gamma} \mu_{\text{lqr}}$-PL. Now similar to before, we recall $\omega_K = \min \{ \beta_K, \zeta_K \}$ and define the quantities
\begin{align*}
    &\phi_\text{dyn} := \sup_{K \in \mathcal{G}^\text{lqr}_{\text{dyn}}} \frac{\gamma}{1-\gamma} \phi_K = \frac{\gamma}{1-\gamma} \sup_{K \in \mathcal{G}^\text{lqr}} \phi_K = \frac{\gamma}{1-\gamma} \phi_{\text{lqr}}, \\ &\lambda_\text{dyn} := \sup_{K \in \mathcal{G}^\text{lqr}_{\text{dyn}}} \frac{\gamma}{1-\gamma} \lambda_K = \frac{\gamma}{1-\gamma} \sup_{K \in \mathcal{G}^\text{lqr}} \lambda_K = \frac{\gamma}{1-\gamma} \lambda_{\text{lqr}}, \\ &\omega_\text{dyn} := \inf_{K \in \mathcal{G}^\text{lqr}_{\text{dyn}}} \omega_K = \inf_{K \in \mathcal{G}^\text{lqr}} \omega_K = \omega_{\text{lqr}}, \\
    &\mu_{\text{dyn}} := \frac{\gamma}{1-\gamma} \mu_{\text{lqr}},
\end{align*}
where the equalities in the first three lines follow from $\G^{\text{lqr}}_{\text{dyn}} = \G^{\text{lqr}}$, which holds due to the cost equivalence shown in Lemma~\ref{lem: noisy-random cost equivalence}. Building on this, along with utilizing Corollary~\ref{cor: gradient estimate expectation noisy dynamics} and Lemma~\ref{lem: gradient estimate bounds - noisy}, we have all the necessary tools to provide the equivalent convergence result of Theorem~\ref{thm: lqr policy gradient} for the noisy dynamics setting.
\begin{corollary}
    \label{cor: lqr policy gradient - noisy}
    Suppose $K_0$ is stable and $\gamma$ is as suggested in Lemma~\ref{lem: exponential decay}, and the update rule follows
    \begin{equation}
        K_{t+1} = K_t - \alpha_t  \widehat{\nabla \Cdyn} (K_t). \label{eq: update rule - noisy dynamics sgd}
    \end{equation}
    If the step-size $\alpha_t$ is chosen as
    \begin{align*}
        \alpha_t = \frac{2}{\mu_{\text{dyn}}} \frac{1}{t+N} \quad &\text{for} \quad N = \max\left\{N_1, \frac{2}{\mu_{\text{dyn}}} \frac{\tilde{\xi}_3 \left(\log{\frac{1}{\delta}}\right)^{3/2}}{(1 - \gamma)\omega_{\text{dyn}}}\right\},
    \end{align*}
    where
    \begin{align*}
        N_1 = \max \left\{ 2,\frac{4 \phi_{\text{dyn}} \tilde{\xi}_4}{\mu_{\text{dyn}}^2 (1 - \gamma)^2} \frac{2}{\Cdyn(K_0)} \right\},
    \end{align*}
    then for a given error tolerance $\eps$ such that $\Cdyn (K_0) \geq \frac{\eps}{20}$, and $\delta$ chosen arbitrarily to satisfy
    \begin{align*}
        \delta \leq \min &\Bigg\{ 2 \times 10^{-5}, \left( \frac{\phi_{\text{dyn}} \tilde{\xi}_4 \omega_{\text{dyn}}}{960 \tilde{\xi}_3^2 \widetilde{c_{K_1}} \Cdyn (K_0)} \right)^3 \eps^3, \cr
        &\left( \frac{\phi_{\text{dyn}} \tilde{\xi}_4}{480 (1 - \gamma)\mu_{\text{dyn}} \tilde{\xi}_3 \widetilde{c_{K_1}} N_1 \Cdyn (K_0)} \right)^3 \eps^3, \left( \frac{\mu_{\text{dyn}} (1 - \gamma)}{240 \tilde{\xi}_3 \widetilde{c_{K_1}}} \right)^3 \eps^3 \Bigg\},
    \end{align*}
    the iterate $K_T$ of~\eqref{eq: update rule - noisy dynamics sgd} after 
    \begin{equation*}
        T = \frac{40}{\eps} N \Cdyn (K_0)
    \end{equation*}
    steps satisfies
    \begin{equation*} 
        \C(K_T) - \C(K^*) \leq \eps
    \end{equation*}
    with a probability of at least $4/5 - \delta T$.
\end{corollary}
Furthermore, we can also extend the mini-batched gradient estimation argument to this setting. Let
\begin{equation}
\overline{\nabla \Cdyn}_{N_s} (K) := \frac{1}{N_s} \sum_{i=1}^{N_s} \widehat{\nabla \Cdyn}_i (K),
\label{eq: minibatch_gradient_estimate - noisy dynamics}
\end{equation}
where each $\widehat{\nabla \Cdyn}_i (K)$ is an i.i.d. copy of $\widehat{\nabla \Cdyn} (K)$ in~\eqref{eq: gradient estimate practical formulation - noisy dynamics}. We are now in a position to provide a convergence result similar to Theorem~\ref{thm: lqr policy gradient - averaged} for the noisy dynamics setting.
\begin{corollary} \label{cor: lqr policy gradient - averaged - noisy dynamics}
    Suppose $K_0$ is stable, $\gamma$ is as suggested in Lemma~\ref{lem: exponential decay}, and the update rule follows
    \begin{equation}
        K_{t+1} = K_t - \alpha \overline{\nabla \Cdyn}_{N_s}(K_t) \label{eq: update_rule_appendix - noisy dynamics}
    \end{equation}
    with a constant step-size $\alpha$ satisfying
    \begin{equation*}
        \alpha \leq \min \left\{ \frac{\omega_{\text{dyn}}}{\widetilde{c_{K_1}} + \sqrt{\frac{\mu_{\text{dyn}} \Cdyn (K_0)}{8}}}, \frac{1}{4 \phi_{\text{dyn}}}, \frac{4}{\mu_{\text{dyn}}} \right\}.
    \end{equation*}
    Then for a given error tolerance $\eps \in (0,\Cdyn (K_0)]$, and for any $\delta \leq \min \left\{ e^{-3/2}, \frac{ 1-\gamma}{3 \tilde{\xi}_3} \sqrt{\frac{\mu_{\text{dyn}} \eps}{8}} \right\}$, the update rule~\eqref{eq: update_rule_appendix - noisy dynamics}, with $N_s \sim \widetilde{\mathcal{O}} (1/\eps)$ chosen according to
    \begin{align*}
    N_s \geq \Bigg\lceil \max \Bigg\{ 5000, 8 \left( \log \frac{2}{\delta} \right)^3, &\frac{2048 \tilde{\xi}_3^2}{9 (1-\gamma)^2 \mu_{\text{dyn}}} \frac{1}{\eps} \left( \log \frac{2(mn+1)}{\delta} \right)^2, \cr
    &\frac{128 \tilde{\xi}_4}{\mu_{\text{dyn}} (1-\gamma)^2} \frac{1}{\eps} \log \frac{2(mn+1)}{\delta} \Bigg\} \Bigg\rceil = \widetilde{\mathcal{O}} \left(\frac{1}{\eps}\right),
\end{align*}
    guarantees that after
    \begin{equation*}
        T = \frac{4}{\alpha \mu_{\text{dyn}}} \log \left( \frac{2 \Cdyn(K_0)}{\eps} \right)
    \end{equation*}
    iterations, we have
    \[
    \Cdyn(K_T) - \Cdyn(K^*) \leq \eps,
    \]
    with a probability of at least $1 - \delta T$.
\end{corollary}

\end{document}